\documentclass[11pt]{article}

\usepackage[letterpaper,margin=1in]{geometry}
\usepackage{amsfonts,amsmath,amssymb,boxedminipage,color,graphicx,mathtools,multirow,nccmath,url}
\usepackage[table]{xcolor}
\usepackage[pageanchor=false,pdfstartview=FitH,colorlinks,linkcolor=blue,filecolor=blue,citecolor=blue,urlcolor=blue]{hyperref}
\usepackage{amsthm}
\usepackage{booktabs}
\usepackage[T1]{fontenc}
\usepackage{lmodern} \normalfont
\DeclareFontShape{T1}{lmr}{bx}{sc} { <-> ssub * cmr/bx/sc }{}
\usepackage{aliascnt}
\usepackage[numbers]{natbib}
\usepackage{xspace}
\usepackage{enumitem}
\usepackage[nottoc]{tocbibind}
\usepackage[hypcap=false]{caption}

\newcommand{\Tableofcontents}{
    \thispagestyle{empty}
    \pagenumbering{gobble}
    \clearpage
    \setcounter{tocdepth}{3}
    \tableofcontents
    \thispagestyle{empty}
    \clearpage
    \pagenumbering{arabic}
}

\newcommand{\func}[1][\relax]{\ensuremath{\mathcal{F}_{\mathsf{#1}}}\xspace}
\newcommand{\facast}{\func[a\mhyphen cast]}
\newcommand{\favss}{\func[a\mhyphen vss]}
\newcommand{\faba}{\func[a\mhyphen ba]}
\newcommand{\ftaba}{\func[trunc\mhyphen a\mhyphen ba]}
\newcommand{\fcaba}{\func[conc\mhyphen a\mhyphen ba]}
\newcommand{\faocc}{\func[a\mhyphen occ]}
\newcommand{\faole}{\func[a\mhyphen ole]}
\newcommand{\fasmt}{\func[a\mhyphen smt]}
\newcommand{\fedsec}{\func[ed\mhyphen sec]}

\newcommand{\prot}[1][\relax]{\ensuremath{\Pi_{\mathsf{#1}}}\xspace}
\newcommand{\paocc}{\prot[a\mhyphen occ]}
\newcommand{\pcaba}{\prot[conc\mhyphen a\mhyphen ba]}
\newcommand{\pacastp}{\prot[a\mhyphen cast\textsuperscript{+}]}
\newcommand{\pspread}{\prot[spread]}
\newcommand{\pselect}{\prot[select]}
\newcommand{\patc}{\prot[a\mhyphen tc]}

\newcommand{\zo}{\{0,1\}}
\renewcommand{\Pr}{{\mathrm {Pr}}}
\mathchardef\mhyphen="2D
\newcommand{\poly}{\mathsf{poly}}

\newcommand{\ie}  {i.e.,\xspace}

\newcommand{\etal}{\textit{et al.}}

\newcommand{\iith}[1]{$#1$\textsuperscript{th}\xspace}

\newcommand{\jth}{\iith{j}}
\newcommand{\kth}{\iith{k}}

\newcommand{\calE}{\ensuremath{\mathcal{E}}\xspace}
\newcommand\ZZ{\ensuremath{\mathbb Z}\xspace}
\newcommand\NN{\ensuremath{\mathbb N}\xspace}
\newcommand\FF{\ensuremath{\mathbb F}\xspace}

\newcommand{\vcommand}[1]{\ensuremath{\texttt{#1}}\xspace}
\newcommand{\revealcmd}{\vcommand{reveal}}
\newcommand{\inputcmd}{\vcommand{input}}
\newcommand{\outputcmd}{\vcommand{output}}
\newcommand{\sendcmd}{\vcommand{send}}
\newcommand{\delaycmd}{\vcommand{delay}}
\newcommand{\delaysetcmd}{\vcommand{delay-set}}
\newcommand{\leakagecmd}{\vcommand{leakage}}
\newcommand{\fetchcmd}{\vcommand{fetch}}
\newcommand{\fetchedcmd}{\vcommand{fetched}}
\newcommand{\replacecmd}{\vcommand{replace}}
\newcommand{\sharecmd}{\vcommand{share}}
\newcommand{\reccmd}{\vcommand{rec}}
\newcommand{\sharedcmd}{\vcommand{shared}}
\newcommand{\openedcmd}{\vcommand{opened}}
\newcommand{\earlycmd}{\vcommand{early}}
\newcommand{\setcmd}{\vcommand{set}}

\newcommand{\permutecmd}{\vcommand{permute}}

\newcommand{\sid}{\ensuremath{\mathsf{sid}}\xspace}
\newcommand{\MID}{\ensuremath{\mathsf{mid}}\xspace}
\newcommand{\continue}{\ensuremath{\mathsf{continue}}\xspace}
\newcommand{\type}{\ensuremath{\mathsf{type}}\xspace}

\newcommand{\naive}{na\"{i}ve\xspace}
\newcommand{\naively}{na\"{i}vely\xspace}

\newlength{\saveparindent}
\newlength{\saveparskip}
\newenvironment{tiret}{%
\begin{list}{\hspace{1pt}\rule[0.5ex]{6pt}{1pt}\hfill}{\labelwidth=15pt%
\labelsep=3pt \leftmargin=18pt \topsep=1pt%
\setlength{\listparindent}{\saveparindent}%
\setlength{\parsep}{\saveparskip}%
\setlength{\itemsep}{1pt}}}{\end{list}}

\newcounter{thm}
\newtheorem{theorem}[thm]{Theorem}
\newtheorem{lemma}{Lemma}
\newtheorem{claim}[lemma]{Claim}
\newtheorem{proposition}[lemma]{Proposition}
\theoremstyle{remark}
\newtheorem{remark}[lemma]{Remark}

\newcommand{\assign}{\ensuremath{\mathrel{\vcenter{\baselineskip0.5ex \lineskiplimit0pt \hbox{\scriptsize.}\hbox{\scriptsize.}}}=}}

\newcommand{\adv}{\mathcal{A}}
\renewcommand{\P}{\mathcal{P}}
\renewcommand{\sim}{\mathcal{S}}
\newcommand{\env}{\mathcal{Z}}
\renewcommand{\complement}{\mathrm{c}}

\newcommand{\exec}{\textsc{exec}}
\newcommand{\ideal}{\textsc{ideal}}
\newcommand{\eom}{\texttt{eom}}
\newcommand{\itr}{\mathsf{itr}}
\newcommand{\participated}{\mathsf{participated}}
\newcommand{\phase}{\mathsf{phase}}
\newcommand{\secret}{\mathsf{secret}}
\newcommand{\finished}{\mathsf{finished}}

\newcommand{\inputlbl}{\mathrm{input}}
\newcommand{\outputlbl}{\mathrm{output}}
\newcommand{\sharelbl}{\mathrm{share}}
\newcommand{\reclbl}{\mathrm{rec}}
\newcommand{\attachlbl}{\mathrm{attach}}
\newcommand{\readylbl}{\mathrm{ready}}
\newcommand{\termlbl}{\mathrm{term}}
\newcommand{\contlbl}{\mathrm{cont}}
\newcommand{\vectorlbl}{\mathrm{vector}}
\newcommand{\setlbl}{\mathrm{set}}
\newcommand{\electlbl}{\mathrm{elect}}

\newcommand{\lcm}{\mathrm{lcm}}
\newcommand{\rsample}{\xleftarrow{\mathrm{R}}}

\newcommand{\funcbox}[3][Functionality] {
    \fbox{
        \begin{minipage}{0.95\textwidth}
        {\footnotesize
            \begin{center}\textbf{#1} #2\end{center}
            #3
        }
        \end{minipage}
    }
}

\newcommand{\protbox}[2]{\funcbox[Protocol]{#1}{#2}}

\title{Concurrent Asynchronous Byzantine Agreement\\ in Expected-Constant Rounds, Revisited\thanks{This is the full version of a paper that appeared at \emph{TCC 2023}. The proceedings version is available at \url{https://doi.org/10.1007/978-3-031-48624-1_16}.}}

\author{Ran Cohen\thanks{Efi Arazi School of Computer Science, Reichman University. E-mail: \texttt{cohenran@runi.ac.il}.}
\and Pouyan Forghani\thanks{Department of Computer Science and Engineering, Texas A\&M University. E-mail: \texttt{\{pouyan.forghani,garay,\allowbreak rsp7\}@tamu.edu}.}
\and Juan Garay$^\ddagger$
\and Rutvik Patel$^\ddagger$
\and Vassilis Zikas\thanks{Department of Computer Science, Purdue University. E-mail: \texttt{vzikas@cs.purdue.edu}.}}

\begin{document}

\date{}

\maketitle
\thispagestyle{empty}

\begin{abstract}
It is well known that without randomization, Byzantine agreement (BA) requires a linear number of rounds in the synchronous setting, while it is flat out impossible in the asynchronous setting. The primitive which allows to bypass the above limitation is known as \emph{oblivious common coin} (OCC). It allows parties to agree with constant probability on a random coin, where agreement is oblivious, i.e., players are not aware whether or not agreement has been achieved.

The starting point of our work is the observation that no known protocol exists for information-theoretic multi-valued OCC---i.e., OCC where the coin might take a value from a domain of cardinality larger than 2---with optimal resiliency in the asynchronous setting (with eventual message delivery). This apparent hole in the literature is particularly problematic, as multi-valued OCC is implicitly or explicitly used in several constructions. (In fact, it is often falsely attributed to the asynchronous BA result by Canetti and Rabin [STOC ’93], which, however, only achieves {\bf binary} OCC and {\bf does not} translate to a multi-valued OCC protocol.)

In this paper, we present the first information-theoretic multi-valued OCC protocol in the asynchronous setting with optimal resiliency, i.e., tolerating $t<n/3$ corruptions, thereby filling this important gap. Further, our protocol efficiently implements OCC with an exponential-size domain, a property which is not even achieved by known constructions in the simpler, synchronous setting.

We then turn to the problem of round-preserving parallel composition of asynchronous BA. A protocol for this task was proposed by  Ben-Or and El-Yaniv [Distributed Computing ’03]. Their construction, however, is flawed in several ways: For starters, it relies on multi-valued OCC instantiated by Canetti and Rabin's result (which, as mentioned above, only provides binary OCC). This shortcoming can be repaired by plugging in our above multi-valued OCC construction. However, as we show, even with this fix it remains unclear whether the protocol of Ben-Or and El-Yaniv achieves its goal of expected-constant-round parallel asynchronous BA, as the proof is incorrect. Thus, as a second contribution, we provide a simpler, more modular protocol for the above task. Finally, and as a contribution of independent interest, we provide proofs in Canetti's Universal Composability framework; this makes our work the first one offering composability guarantees, which are important as BA is a core building block of secure multi-party computation protocols.
\end{abstract}

\Tableofcontents

\section{Introduction}
\label{sec:intro}

\emph{Byzantine agreement} (BA)~\cite{pease1980reaching,lamport1982byzantine} enables $n$ parties to reach agreement on one of their inputs in an adversarial setting, facing up to $t$ colluding and cheating parties. The core properties require all honest parties to eventually terminate with the same output (\emph{agreement}), which equals their input value in case they all begin with the same common input (\emph{validity}). BA and its closely related single-sender variant, \emph{broadcast}, are fundamental building blocks in the construction of cryptographic protocols, in particular for \emph{secure multi-party computation} (MPC)~\cite{yao1982protocols,goldreich1987how,ben-or1988completeness,chaum1988multiparty}, in which parties wish to privately compute a joint function over their inputs.

We consider a complete network of authenticated and private point-to-point (P2P) channels, which enables every pair of parties to communicate directly.
The central settings in which BA has been studied are:
\begin{description}
\item[The \emph{synchronous} setting.]
Here the protocol proceeds in a round-by-round fashion, and messages sent in a given round are guaranteed to be delivered by the start of the next round.
The round structure can be achieved given synchronized clocks and a known bound on message delivery, and enables the use of \emph{timeouts}. This clean abstraction allows for simpler analyses of protocols, but comes at the cost that parties must wait for the worst-case delay in each round before they can proceed.
\item[The \emph{asynchronous} setting.]
Here no assumptions are made on the clocks or on the bound of message delivery (other than that each message will be eventually delivered), and messages may be adversarially delayed by an arbitrary (yet finite) amount of time. The main challenge is that timeouts cannot be used, implying the inability to distinguish a slow honest party from a silent, noncooperative corrupted party. On the positive side, asynchronous protocols can advance as fast as the network allows, irrespectively of the worst-case delay, and each party can proceed to the next step as soon as it gets sufficiently many messages.
\end{description}

The \emph{feasibility} of BA is inherently related to the synchrony assumptions of the system. Synchronous BA with \emph{perfect} security is achievable with deterministic protocols for $t<n/3$ \cite{pease1980reaching,BGP89,garay1998fully}; this bound is tight in the plain model\footnote{That is, without setup assumptions and without imposing resource restrictions on the adversary.} even when considering weaker, computational security \cite{pease1980reaching,fischer1986easy,borcherding1996levels}, but can be overcome with setup assumptions, yielding computationally secure BA \cite{dolev1983authenticated}, and even information-theoretically secure BA \cite{pfitzmann1992unconditional}, for $t<n/2$.\footnote{The bound $t<n/2$ is tight for BA \cite{Fitzi03}; however, under the same setup assumptions, broadcast can be solved for any number of corruptions.}
On the other hand, deterministic asynchronous BA (A-BA) is impossible even facing a single crash failure \cite{fischer1985impossibility}, and randomized solutions are a necessity. Following Ben-Or \cite{ben-or1983another} and Rabin \cite{rabin1983randomized}, information-theoretic randomized A-BA has been achieved with $t<n/3$ corruptions \cite{bracha1987asynchronous,canetti1993fast,abraham2008almost}, which, as opposed to the synchronous case, is a tight bound even given setup and cryptographic assumptions \cite{DLS88,BT85}.

\paragraph{Round complexity of BA.}
In the synchronous setting, it is known that \emph{deterministic} BA requires $t+1$ rounds \cite{fischer1982lower,dolev1983authenticated}, a bound that is matched by early feasibility results \cite{pease1980reaching,dolev1983authenticated,garay1998fully}. It is also known that $t$-secure randomized BA for $t\in\Theta(n)$ cannot terminate in a \emph{strict} constant number of rounds \cite{CMS89,KY86,CPS19}; yet, \emph{expected} constant-round protocols have been constructed both for $t<n/3$ in the plain model \cite{feldman1997optimal} and for $t<n/2$ in the PKI model \cite{fitzi2003efficient,katz2009expected}.\footnote{Fitzi and Garay \cite{fitzi2003efficient} devised expected-constant-round BA for $t<n/2$ in the PKI model under number-theoretic assumptions. Katz and Koo \cite{katz2009expected} established a similar result from the minimal assumption of digital signatures, which yields an information-theoretic variant using pseudo-signatures \cite{pfitzmann1992unconditional}.} The latter protocols follow the approach of Rabin~\cite{rabin1983randomized} and rely on an \emph{oblivious common coin} (OCC), that is, a distributed coin-tossing protocol over a domain $V$ such that the output of every honest party is a common random value $v\in V$ with constant probability~$p$, but with probability $1-p$ is independently and adversarially chosen; the coin toss is ``oblivious'' since the parties cannot distinguish between a successful coin toss and an adversarial one.\footnote{This primitive is sometimes known as a ``weak'' common coin in the literature.}

Loosely speaking, round complexity in the asynchronous setting can be defined based on the expected number of times an honest party has to alternate between sending and receiving messages that are ``causally related''~\cite{coretti2016constant}. In an unpublished manuscript, Feldman~\cite{feldman1989asynchronous} generalized the expected-constant-round synchronous protocol of~\cite{feldman1988optimal} to asynchronous networks with $t<n/4$; at the core of the construction lies a {\bf binary} (asynchronous) OCC protocol which actually has resiliency $\min(t',\lceil n/3\rceil-1)$, where $t'$ is the corruption threshold of an \emph{asynchronous verifiable secret sharing} (A-VSS) scheme.\footnote{Feldman's A-VSS suffers from a negligible error probability. An \emph{errorless} A-VSS scheme for $t<n/4$ is given in~\cite{ben-or1993asynchronous} and used to construct a \emph{perfectly} secure asynchronous MPC protocol with resiliency $t<n/4$.} Canetti and Rabin~\cite{canetti1993fast} constructed A-VSS for $t<n/3$, thereby obtaining A-BA in expected-constant rounds for the optimal threshold $t<n/3$.

\paragraph{Concurrent BA.}
All of the above constructions are for solving a \emph{single} instance of BA. However, most applications, particularly MPC protocols, require composing multiple instances of BA: sequentially, in parallel, or concurrently. While the composition of synchronous, deterministic protocols is relatively simple (although care must be taken in the cryptographic setting \cite{LLR06}), composing expected-constant-round protocols with probabilistic termination is a much more challenging task.

Indeed, Ben-Or and El-Yaniv~\cite{BE88,ben-or2003resilient} observed that running $m$ instances of a probabilistic-termination protocol in parallel may incur a blow-up in the expected number of rounds until they all terminate. The technical reason is that the expectation of the maximum of $m$ independent, identically distributed random variables does not necessarily equal the maximum of their expectations. In particular, for expected-constant-round BA with a geometric termination probability (which is the case in all known protocols), the parallel composition of $m$ instances terminates after expected $\Theta(\log m)$ rounds \cite{cohen2019probabilistic}. Further, even when all parties start the protocol together, simultaneous termination is not guaranteed as the adversary can force some honest parties to terminate after the others; although this gap can be reduced to a single round, inaccurate sequential re-synchronization of $\ell$ instances of BA may lead to an exponential blow up in $\ell$ if not done with care \cite{LLR02}. Following \cite{ben-or2003resilient,LLR02,LLR06,katz2009expected}, recent works have shown how to compose synchronous BA in a round-preserving way with simulation-based security in Canetti's framework for universal composability (UC) \cite{canetti2020universally}, with optimal resiliency and \emph{perfect} security in the plain model for $t<n/3$ \cite{cohen2019probabilistic}, and with cryptographic security in the PKI model for $t<n/2$ \cite{cohen2021round}.

In the asynchronous setting, the parties are not assumed to begin the protocol at the same time, so intuitively, sequential composition is not problematic. However, running $m$ instances of expected-constant-round A-BA concurrently would yield a $\Theta(\log m)$ blowup as in the synchronous case. Ben-Or and El-Yaniv \cite{ben-or2003resilient} further showed how to execute $m$ instances of A-BA in expected-constant rounds and with optimal resiliency $t<n/3$; however, their solution is more complicated than the synchronous one. In addition, they only prove a property-based security definition of A-BA that does not necessarily address modern security requirements such as security under composition, or facing adaptive adversaries.

\subsection{Concurrent A-BA in Expected-Constant Rounds: Cracks in the Concrete}
The underlying idea behind the synchronous protocol of Ben-Or and El-Yaniv~\cite{ben-or2003resilient} for parallel BA is to execute each BA instance multiple times over the same inputs, but only for a constant number of rounds. For a suitable choice of parameters, this ensures that with high probability each party will have obtained at least one output value in each such batch. To coordinate these outputs, the parties then run an \emph{oblivious leader election} (OLE) protocol, which guarantees that with constant probability, a random leader is elected. In the event that the leader is honest and the parties obtained an output in each batch (which, again, occurs with constant probability) the parties will terminate; otherwise they repeat the process.

The same general technique underlies Ben-Or and El-Yaniv's asynchronous protocol, but great care is needed to deal with the low message dispersion inherent in asynchronous networks while maintaining optimal resiliency $t<n/3$.
A closer look at their security proof indeed raises a number of subtle issues. First, they point to Canetti and Rabin~\cite{canetti1993fast} for instantiating the (asynchronous) OLE primitive used in their construction (called \emph{A-Election()}). (Recall that Canetti and Rabin construct an OCC to obtain A-BA in expected-constant rounds; an $n$-valued OCC would indeed imply OLE.)
As it turns out, the OCC construction in~\cite[Sec.\ 8]{canetti1993fast} (as well as the more detailed versions in \cite[Sec.\ 5.7]{canetti1996studies} and \cite[Sec.\ 8]{CR98}) is only \emph{binary} (i.e., it only works for $V=\zo$), and it does not seem straightforward to generalize to larger (non-constant-sized) domains. 

Further, running $\log n$ executions of a binary OCC, with agreement probability $p$, in parallel to make it multi-valued yields only $p^{\log n}\in 1/\poly(n)$ probability of agreement, and as long as the coin is not \emph{perfectly} fair (i.e., guaranteeing that all honest parties output $0$ with probability $1/2$, and $1$ with probability $1/2$, which is impossible in the asynchronous setting~\cite{freitas2022distributed}), that would not imply OLE with constant success probability. We can also consider a different type of ``strong'' common coin: one which achieves arbitrarily small bias, possibly with guaranteed agreement. Such an $\epsilon$-biased binary coin would in fact achieve arbitrarily small bias even when run in parallel, and there are even known protocols for this task that simultaneously enjoy almost-sure termination and optimal resiliency ($t<n/3$)~\cite{DBLP:journals/dc/AbrahamDS22}, but these coins do not have constant round complexity, since at the very least there will be a dependence on $\epsilon$.\footnote{The protocol in~\cite{DBLP:journals/dc/AbrahamDS22} actually incurs an additional overhead of several factors of $n$, due to its use of ACS and almost-surely terminating A-BA.}

We note that Patra \etal~\cite{patra2014asynchronous} claim to construct a $(t+1)$-bit asynchronous OCC, but their main focus is on communication complexity, and the agreement probability of their protocol is no better than would be obtained by running $t+1$ executions of a binary OCC protocol (i.e., exponentially small). Techniques used to get OLE in the synchronous setting~\cite{feldman1988optimal,katz2009expected} do not seem to extend in asynchrony for $t<n/3$ (we elaborate on this in Section~\ref{sec:results}). Thus, to the best of our knowledge, no existing OCC construction is \emph{simultaneously} optimally resilient, multi-valued, and asynchronous, without relying on computational assumptions (in Section~\ref{sec:related} we discuss solutions in the cryptographic setting).

Second, there is a subtle issue in the logic of one of the proofs in~\cite{ben-or2003resilient}. This issue raises concerns about the validity of the proof claiming an expected-constant round complexity of one of the main subroutines---namely, the $\pselect$ subroutine, which handles the shortcomings of their message-distribution mechanism. Specifically, in the analysis of $\pselect$ it is claimed that if the leader is chosen from a certain set, the protocol will terminate. However, further examination reveals that there are scenarios in which the protocol may not terminate for certain leaders from that set. As a result, this issue casts doubt on the promised expected-constant round complexity of their concurrent A-BA protocol. (Refer to Appendix~\ref{app:be} for more details on this issue.)

Finally, the concurrent asynchronous (resp., synchronous) BA protocol in \cite{ben-or2003resilient} relies on \emph{multi-valued} asynchronous (resp., synchronous) BA in expected-constant rounds.\footnote{This is true even if one is interested only in \emph{binary} concurrent BA (i.e., when the input vectors consist of bits). Multi-valued BA is needed to agree on the leader's output vector.} Recall that running the binary protocols in~\cite{feldman1997optimal} or~\cite{canetti1993fast} $\omega(1)$ times in parallel would terminate in expected $\omega(1)$ rounds, so they cannot be \naively used for this task. In the synchronous setting, Turpin and Coan~\cite{turpin1984extending} extended binary BA to multi-valued BA for $t<n/3$ with an overhead of just two rounds.
Ben-Or and El-Yaniv claim that this technique can be adapted for asynchronous networks by using Bracha's ``A-Cast'' primitive~\cite{bracha1987asynchronous} for message distribution. However, a closer look reveals that although the Turpin-Coan extension works (with appropriate modifications) in the asynchronous setting for $t<n/5$, it {\bf provably does not work} when $t\ge n/5$, regardless of the specific choice of the underlying binary A-BA protocol and even when the adversary is limited to static corruptions (see Appendix~\ref{app:tc}).

More recently, an optimally resilient multi-valued A-BA protocol with expected-constant round complexity was proposed by Patra~\cite{patra2011error}, but it relies on the \emph{Agreement on a Common Subset} (ACS) protocol in~\cite{ben-or1994asynchronous}; however, as we explain below, this ACS protocol does not achieve expected-constant round complexity without some modifications, which require either expected-constant-round concurrent A-BA (this would be circular), or information-theoretic asynchronous OLE with optimal resiliency. Fortunately, Most\'{e}faoui and Raynal~\cite{mostefaoui2017signature} recently gave a black-box, constant-round reduction from multi-valued to binary A-BA for $t<n/3$,\footnote{They are also concerned with obtaining $O(n^2)$ message complexity. The novelty of their result, even without this more stringent requirement, does not seem to be acknowledged in the paper.} using just one invocation of the underlying binary protocol as in~\cite{turpin1984extending}.

We emphasize that the asynchronous protocol of Ben-Or and El-Yaniv \cite{BE88,ben-or2003resilient} lies, either explicitly or implicitly, at the core of virtually every round-efficient asynchronous MPC construction \cite{ben-or1993asynchronous,ben-or1994asynchronous,hirt2005cryptographic,beerliova-trubiniova2007simple,hirt2008asynchronous,cohen2016asynchronous,coretti2016constant,BZL20,LLMMT20}.
The concerns raised above regarding the result of
\cite{ben-or2003resilient} render this extensive follow-up work unsound.
In this paper, we revisit this seminal result and rectify these issues.

\subsection{Overview of Our Results}
\label{sec:results}

We now present an overview of our results, which are three-fold, including a detailed exposition of our techniques.

\paragraph{Multi-valued and asynchronous oblivious common coin.}

As a starting point, we look at the \emph{binary} asynchronous OCC protocol of Canetti and Rabin~\cite{canetti1993fast}. The idea (following the approach of \cite{feldman1997optimal} in the synchronous setting and \cite{feldman1989asynchronous} in the asynchronous setting) is that each party secret-shares a random vote for each party, using an optimally resilient A-VSS scheme. Each party accepts $t+1$ of the votes cast for him (at least one of which must have come from an honest party), and once it is determined that enough secrets have been fixed (based on several rounds of message exchange, using Bracha's A-Cast primitive~\cite{bracha1987asynchronous} for message distribution), the parties begin reconstructing the accepted votes. The sum or ``tally'' of these votes becomes the value associated with the corresponding party. After computing the values associated with an appropriate set of at least $n-t$ parties, an honest party outputs $0$ if at least one of those tallies is $0$, and outputs~$1$ otherwise. Note that each tally must be uniformly random, and the properties of A-Cast guarantee that no two honest parties can disagree on the value of any given tally (although up to $t$ tallies may be ``missing'' in an honest party's local view); moreover, Canetti and Rabin show using a counting argument  that at least $n-2t>n/3$ of the tallies are known to all honest parties (i.e., common to their local views).\footnote{An enhanced version of this counting argument, using the ``Gather'' protocol of Abraham \etal~\cite{abraham2023reaching}, can be used to increase the size of this intersection to $n-t>2n/3$.} This can be used to show that with probability at least $1/4$ all honest parties output $0$, and similarly for $1$.

In the conference version of Feldman and Micali's paper~\cite{feldman1988optimal}, there is a brief remark suggesting that the synchronous version of the above protocol can be modified to obtain (oblivious) leader election.\footnote{This claim no longer appears in the ICALP~\cite{feldman1989optimal} or journal~\cite{feldman1997optimal} versions of the paper, or in Feldman's thesis~\cite{feldman1988thesis}.} Rather than outputting a bit, parties output the index of the party whose tally is minimum; when the domain of secrets is large enough, with constant probability the same \emph{honest} party is chosen. This approach was fully materialized by Katz and Koo~\cite{katz2009expected} for $t<n/3$ in the information-theoretic setting and $t<n/2$ in the computational setting. We note that what is obtained is not exactly OLE as we have defined it, as the adversary can of course bias the index of the elected party, but nevertheless, this is sufficient for the concurrent BA protocols of Ben-Or and El-Yaniv~\cite{ben-or2003resilient}.

Unfortunately, the approach effective in the synchronous setting cannot be directly applied in the asynchronous setting while maintaining optimal resiliency. The challenge arises from the fact that the adversary can selectively remove $t$ coordinates from the honest parties' local views of the tallies, ensuring that the leader is not chosen from the parties corresponding to those missing coordinates. This poses a significant obstacle as the original concurrent A-BA protocol by Ben-Or and El-Yaniv~\cite{ben-or2003resilient} terminates successfully when the leader is honest and selected from an adversarially chosen subset of $n-2t$ parties. Consequently, when $t<n/3$, the size of this subset is only $t+1$, allowing the adversary to reduce the probability of choosing an appropriate leader to as low as $\frac{1}{n-t}$. The same issue arises in our simplified concurrent A-BA protocol, where we also require the leader to be selected from an adversarially chosen subset of parties with size $n-2t$. Again, when $t<n/3$, this subset can be of size $t+1$, leading to the same challenge. However, when $t<n/4$, the set of potential appropriate leaders becomes larger in both concurrent A-BA protocols, enabling us to circumvent this issue. Therefore, in addition to the inherent value in obtaining a true OCC (where a successful coin toss produces a uniform value), we specifically need such a primitive to obtain optimal resiliency for concurrent A-BA.

In summary, no OLE construction---and therefore no concurrent A-BA protocol in expected-constant rounds---exists with optimal resiliency $t<n/3$ in asynchronous networks. We now describe our solution to this problem, which is based on the following simple combinatorial observation. If we work over a field of size $N\in\Theta(n^2)$, then with \emph{constant} probability, at least one value will be repeated in the global view of tallies, and, moreover, all repeats will occur within the subset of indices known to all honest parties. The intuition for this fact is that when $N\in O(n^2)$ there is, due to the birthday paradox, at least one repeat in any constant fraction of the indices with constant probability, and when $N\in\Omega(n^2)$, there are no repeats at all with constant probability. There is a ``sweet spot'' between these extremes that can be leveraged to extract shared randomness: Honest parties output the value that is repeated in their local views of the tallies and has the minimum index in the vector of tallies. With this modification of Canetti and Rabin's protocol~\cite{canetti1993fast}, we obtain the first $n^2$-valued asynchronous OCC for $t<n/3$ in the information-theoretic setting. We remark that the above combinatorial observation at the heart of our construction may be of independent interest.

It is straightforward to extend our OCC protocol to accommodate arbitrary domains $V$. One approach is to work over a prime field of size at least $\lcm(n^2,|V|)$. In this setting, we can still identify ``repeats'' by considering the tallies modulo $n^2$. Once the repeat with minimum index is determined, we can generate a common random output by reducing the corresponding (original) field element modulo $|V|$. In particular, when $|V|=n$, we get asynchronous OLE with optimal resiliency.

Proving the security of our multi-valued OCC protocol in a \emph{simulation-based} manner is not without its own challenges. The issue is that the simulator must expose a view of the vector of tallies that both adheres to the distribution in the real world and is consistent with the random value chosen by the OCC functionality (in the case of a successful coin toss) in the ideal world. While the simulator can easily determine the exact set of indices known to all honest parties from its internal execution with the real-world adversary, properly sampling the ``repeat pattern'' according to these constraints is a delicate task; furthermore, since the functionality is only parameterized by a (constant) lower bound on the probability of a successful coin toss, the simulator must handle the complementary case carefully in order to avoid skewing the distribution. It is not immediately clear how to perform this inverse sampling efficiently.

A heavy-handed solution is to simply have the functionality sample the tallies itself, and then determine the output based on the location of repeats relative to the subset of indices (supplied by the simulator) that are known to all honest parties. This protocol-like functionality would certainly allow for simulation---and would in fact be sufficient for our purposes---but its guarantees are more difficult to reason about and, more importantly, it cannot be realized by other protocols! Instead, we construct a simulator that, given the exact probabilities of certain events in the protocol, can also efficiently sample from those events, potentially conditioned on the output of a successful coin toss. By selecting the appropriate sampling procedure, the simulator can derive a vector of tallies that preserves the (perfect) indistinguishability of the real and ideal worlds. Equipped with the means to carry out the inverse sampling, we can now realize a more abstract (and natural) OCC functionality, which is ready to be plugged into higher-level protocols.

\paragraph{Simplified concurrent A-BA.}

Ben-Or and El-Yaniv \cite{ben-or2003resilient} devised an expected-constant-round concurrent A-BA protocol. However, in addition to relying on unspecified building blocks, as mentioned above, it suffers from logical issues in its proof. Although our new OCC protocol can instantiate the missing primitive, doubts remain about the expected-constant round complexity due to a lingering issue in the proof. This issue stems from the steps taken to address low message dispersion, and the possibility of resolving it without changing the protocol is unclear. To tackle this, we redesign the message-distribution phase, avoiding the problem in the proof and obtaining stronger guarantees. These guarantees simplify the protocol structure, achieving a level of simplicity comparable to the synchronous solution.

In more detail, we build on Ben-Or and El-Yaniv's work, where parties initiate multiple executions for each A-BA instance that are ``truncated'' after a fixed number of iterations. However, we adopt a different approach to message distribution, allowing us to establish the rest of the protocol based on the simpler structure of their synchronous solution. Similarly to~\cite{ben-or2003resilient}, we set parameters to ensure that each party receives at least one output from each batch of truncated A-BA executions (for each instance) with constant probability. Based on their local results from those truncated A-BAs, parties create suggestions for the final output. Subsequently, the message-distribution phase begins with parties running a binary A-BA to verify a precondition, ensuring the ability to validate each other's suggestions. If they choose to proceed, they A-Cast their suggestions and only accept those they can validate using their local results from the truncated A-BAs. This validation relies on the property that honest parties terminate within two consecutive iterations in each A-BA execution, ensuring the validity of suggestions even when provided by corrupted parties. Once parties receive $n-t$ suggestions, they proceed by A-Casting the set of the first $n-t$ suggestions they receive, along with the identities of their providers. The checked precondition guarantees that everyone can move on to the next step.

Parties then wait until they receive enough suggestions and enough sets, ensuring that at least $n-t$ sets are completely contained within the set of suggestions they received. By employing a counting argument similar to~\cite{feldman1989asynchronous,canetti1993fast}, we can ensure the delivery of suggestions from at least $n-2t>n/3$ parties to all honest participants. Moreover, the validation process applied to the suggestions guarantees that each honest party only accepts valid suggestions, even if they originate from corrupted sources. These robust guarantees in the distribution of suggested outputs establish that the $n-2t$ suggestions commonly received by honest parties are legitimate outputs. This enables us to directly utilize the synchronous protocol in the asynchronous setting.

We employ our new asynchronous OCC protocol to elect a leader for party coordination. Every party adopts the leader's suggestion and runs a multi-valued A-BA to handle the obliviousness of the leader-election mechanism. If the parties successfully agree on a random leader who happens to be chosen from the $n-2t>n/3$ parties whose suggestions are commonly received, which occurs with constant probability, all honest parties initiate the A-BA protocol with the leader's suggestion and output the same value. If the leader is not among those parties (or there is disagreement on the leader), precautions are taken to ensure no malicious value is output. For this purpose, we utilize an ``intrusion-tolerant'' A-BA protocol that guarantees the output to be either a default value $\bot$ or one of the honest parties' inputs.\footnote{Ben-Or and El-Yaniv~\cite{ben-or2003resilient} introduced and used a strengthened property for (A-)BA without naming it, which was later called ``non-intrusion'' validity in~\cite{mostefaoui2017signature}. Non-intrusion validity lies between standard validity and ``strong'' validity \cite{fitzi2003efficient}, as it requires that a value decided by an honest party is either an honest party's input or a special symbol $\bot$ (i.e., the adversary cannot intrude malicious values into the output).} Thus, in case the output is not $\bot$ the parties can safely output it, and otherwise they can simply start over. Unlike the synchronous solution in~\cite{ben-or2003resilient}, which includes a validation step in between the agreement on the leader's suggestion and a final agreement on termination, our concurrent A-BA protocol involves validation in the message-distribution phase, and so consensus on a non-default value can directly be used to agree on termination.

By following the above approach, we sidestep the issues in Ben-Or and El-Yaniv's proof and achieve a significantly simpler protocol structure for the asynchronous setting compared to the one presented in~\cite{ben-or2003resilient}. Somewhat surprisingly, the resulting protocol is conceptually as simple as its synchronous counterpart.

We now elucidate exactly how our findings enable round-efficient asynchronous MPC, in both the computational and information-theoretic settings, while maintaining optimal resiliency. Asynchronous MPC crucially relies on ACS for determining the set of input providers \cite{ben-or1993asynchronous}; this task commonly boils down to concurrently executing $n$ instances of A-BA. Our expected-constant-round concurrent A-BA protocol can be directly plugged into the asynchronous MPC protocols in~\cite{ben-or1993asynchronous,hirt2005cryptographic,beerliova-trubiniova2007simple,hirt2008asynchronous,cohen2016asynchronous,LLMMT20}, preserving their (expected) round complexity. However, despite folklore belief, concurrent A-BA cannot be used in a black-box way in the BKR protocol~\cite{ben-or1994asynchronous} to achieve asynchronous MPC with round complexity linear in the depth of the circuit. The issue (as pointed out in \cite{DBLP:conf/ccs/ZhangD22,AAPS23}) is that the ACS protocol outlined in~\cite{ben-or1994asynchronous} assigns input values for certain A-BA instances based on the outputs of other instances. This problem also affects other works like~\cite{coretti2016constant,BZL20} that rely on~\cite{ben-or1994asynchronous}.

Fortunately, this issue can be readily addressed by modifying the ACS protocol from~\cite{ben-or1993asynchronous} (which is secure for $t<n/3$). Recall that this protocol involves a preprocessing step for distributing input shares before initiating concurrent A-BA, which necessitates $O(\log n)$ rounds. Replacing the $O(\log n)$-round preprocessing step with the constant-round ``Gather'' protocol described in~\cite{abraham2023reaching}—an enhanced version of Canetti and Rabin's counting argument~\cite{canetti1993fast}—results in an ACS protocol with constant round complexity. This modified ACS protocol can be seamlessly integrated into both~\cite{ben-or1993asynchronous} and~\cite{ben-or1994asynchronous}, effectively rendering their round complexity independent of the number of parties. Alternatively, one can leverage the expected-constant-round ACS protocol recently proposed by Abraham \etal~\cite{AAPS23}, or the one by Duan \etal~\cite{duan2023practical} (instantiating all building blocks with information-theoretic realizations; in particular, the assumption of a ``Rabin dealer''~\cite{rabin1983randomized}, used for leader election, can be replaced by our multi-valued OCC).

\paragraph{Composable treatment of expected-constant-round concurrent A-BA.}

We choose to work in Canetti's Universal Composability (UC) framework~\cite{canetti2020universally}, and as such, we prove the security of our protocols in a \emph{simulation-based} manner. (See Section~\ref{sec:uc} for a high-level overview.) The UC framework provides strong composability guarantees when secure protocols are run as a subroutine in higher-level protocols (this is absolutely critical in our context given that expected-constant-round concurrent A-BA is a key building block in many round-efficient cryptographic protocols, as mentioned above), and even in \emph{a priori} unknown or highly adversarial environments (such as asynchronous networks). Moreover, it enables us to provide a modular, bottom-up security analysis. However, obtaining a composable and round-preserving treatment of ``probabilistic-termination'' BA is non-trivial, as pointed out by Cohen \etal~\cite{cohen2019probabilistic,cohen2021round} in the synchronous setting. In the following, we discuss some unique issues in asynchrony and how we address them.

To model eventual message delivery, we follow~\cite{katz2013universally,coretti2016constant,cohen2016asynchronous} and require parties to repeatedly attempt fetching messages from the network. The first $D$ such requests are ignored by the functionality, where $D$ is a value provided by the adversary in unary so that it remains bounded by the adversary's running time (i.e., so that messages cannot be delayed \emph{indefinitely}). It is straightforward then to derive a formal notion of asynchronous rounds in UC, based on this mechanism. We remark that unlike in the synchronous setting~\cite{katz2013universally}, (asynchronous) rounds cannot be used by the environment to distinguish the real and ideal worlds in the asynchronous setting. Thus, as opposed to \cite{cohen2019probabilistic,cohen2021round}, our functionalities are round-unaware. Similarly, we do not need to deal with standard issues in sequential composition, namely, non-simultaneous start/termination (``slack''), since asynchronous protocols are already robust to slack. Indeed, it is entirely possible that some parties receive output from a (secure) asynchronous protocol before other parties have even started the protocol!

On the other hand, the issue of \emph{input incompleteness} is trickier to address. This refers to the problem that in the asynchronous setting, the inputs of up to $t$ honest parties may not be considered in the result of the computation; the remaining $n-t$ parties form a ``core set'' of input providers. Note that in the worst case, the core set is adversarially chosen and includes all corrupted parties. Prior work~\cite{coretti2016constant,cohen2016asynchronous} allowed the adversary to send an explicit core set to the functionality; however, this approach does not always accurately model what happens in the real world, and can cause difficulties in the simulation. Instead, our solution is to allow the adversary to define the core set \emph{implicitly}, by delaying the submission of inputs to the functionality in the same way that it delays the release of outputs from the functionality. Using this novel modeling approach, we obtain updated functionalities for some standard asynchronous primitives that more accurately capture realizable security guarantees.

\subsection{Additional Related Work}
\label{sec:related}

As mentioned earlier, the $t+1$ lower bounds for deterministic BA \cite{fischer1982lower,dolev1983authenticated} were extended to rule out \emph{strict-}constant-round $t$-secure randomized BA for $t=\Theta(n)$ \cite{CMS89,KY86,CPS19}; these bounds show that any such $r$-round BA must fail with probability at least $(c\cdot r)^{-r}$ for a constant~$c$, a result that is matched by the protocol of \cite{GGL22}. Cohen \etal~\cite{CHMOS22} showed that for $t>n/3$, two-round BA is unlikely to reach agreement with constant probability, implying that the \emph{expected} round complexity must be larger; this essentially matches Micali's BA \cite{Micali17} that terminates in three rounds with probability $1/3$. Attiya and Censor-Hillel \cite{AC10} extended the results on worst-case round complexity for $t=\Theta(n)$ from \cite{CMS89,KY86} to the asynchronous setting, showing that any $r$-round A-BA must fail with probability $1/c^r$ for some constant $c$.

In the dishonest-majority setting, expected-constant-round broadcast protocols were initially studied by Garay \etal~\cite{garay2007round}, who established feasibility for $t=n/2+O(1)$ as well as a negative result. A line of work \cite{FN09,CPS20,WXDS20,WXSD20,SLMNPT23} established expected-constant-round broadcast for any constant fraction of corruptions under cryptographic assumptions.

Synchronous and (binary) asynchronous OCC protocols in the information-theoretic setting were discussed earlier. Using the synchronous protocol in~\cite{ben-or2003resilient}, Micali and Rabin~\cite{micali1990collective} showed how to realize a \emph{perfectly unbiased} common coin in expected-constant rounds for $t<n/3$ over secure channels (recall that this task is impossible in asynchronous networks \cite{freitas2022distributed}). In the cryptographic setting, both synchronous and asynchronous OCC protocols with optimal resiliency are known, relying on various computational assumptions; we mention a few here. Beaver and So~\cite{beaver1993global} gave two protocols tolerating $t<n/2$ corruptions in synchronous networks, which are secure under the quadratic residuosity assumption and the hardness of factoring, respectively. Cachin \etal~\cite{cachin2005random} presented two protocols for $t<n/3$ and asynchronous networks, which are secure in the random oracle model based on the RSA and Diffie-Hellman assumptions, respectively. Nielsen~\cite{nielsen2002threshold} showed how to eliminate the random oracle and construct an asynchronous OCC protocol relying on standard assumptions alone (RSA and DDH). Although these constructions are for asynchronous networks, they can be readily extended to work in synchronous networks for $t<n/2$ (e.g., so that they can be used in the computationally secure BA protocol from~\cite{fitzi2003efficient}). We also note that while the resiliency bounds for asynchronous OCC protocols coincide in the information-theoretic and computational settings, working in the latter typically yields expected-constant-round A-BA protocols that are much more efficient in terms of communication complexity (i.e., \emph{interaction}) than the unconditionally secure protocol of Canetti and Rabin~\cite{canetti1993fast}.

Lastly, we discuss solutions to multi-valued A-BA in the computational setting. Cachin \etal~\cite{cachin2001secure} studied a more general version of this problem, in which the validity property is replaced with ``external'' validity, where the output domain can be arbitrarily large but the agreed-upon value must only satisfy an application-specific predicate. They gave a construction for multi-valued ``validated'' A-BA that runs in expected-constant rounds for $t<n/3$, assuming a PKI and a number of threshold cryptographic primitives (including an asynchronous OCC), and used it to obtain an efficient protocol for asynchronous \emph{atomic} broadcast. Recently, Abraham \etal~\cite{abraham2019asymptotically} (and follow-ups, e.g.,~\cite{lu2020dumbo,abraham2023reaching,gao2022efficient}) have improved the communication complexity. The work of Fitzi and Garay~\cite{fitzi2003efficient} also considered \emph{strong} (A-)BA. Here we require ``strong'' validity: the common output must have been one of the honest parties' inputs (note that this is equivalent to standard validity in the binary case). When the size of the input domain is $m>2$, neither a Turpin-Coan-style reduction~\cite{turpin1984extending} nor the obvious approach of running $\log m$ parallel executions of a binary protocol would suffice to realize this stronger notion of agreement; indeed, Fitzi and Garay showed that strong A-BA is possible if and only if $t<n/(m+1)$. This bound holds in both the information-theoretic and computational settings. Their unconditionally secure asynchronous protocol actually involves oblivious coin flipping on the domain, but $m>2$ forces $t<n/4$ and the binary asynchronous OCC protocols in~\cite{feldman1989asynchronous,canetti1993fast} can be extended to multi-valued in this regime, as discussed in Section~\ref{sec:results}.

\subsection{Organization of the Paper}
\label{sec:organization}

The rest of the paper is organized as follows. We start in Section~\ref{sec:prelims} with some preliminaries, including a brief overview of the UC framework. In Section~\ref{sec:model}, we discuss our model in greater depth, formally specifying the asynchronous primitives we use as building blocks and motivating our novel modeling choices along the way. Section~\ref{sec:occ} contains our multi-valued asynchronous OCC construction with optimal resiliency $t<n/3$, and its security proof. Finally, in Section~\ref{sec:concurrent} we present the impactful result  that our OCC construction enables: a simplified and sound protocol for concurrent A-BA in expected-constant rounds. Appendix~\ref{app:be} contains our attack on the asynchronous construction in~\cite{ben-or2003resilient} (along with the required background material), and Appendix~\ref{app:tc} demonstrates that an asynchronous version of the binary to multi-valued BA extension in~\cite{turpin1984extending} works {\bf if and only if} $t<n/5$.
\section{Model and Preliminaries}
\label{sec:prelims}

For $m\in\NN$, we use $[m]$ to denote the set $\{1,\ldots,m\}$. Our protocols often use dynamically growing sets~$S$; when $|S|\ge k$, we denote by $S^{(k)}\subseteq S$ the set containing the first $k$ elements that were added to $S$. We also denote by $\FF$ the field from which the protocols' messages come. Our statements (in particular, the ones about statistical security) assume an (often implicit) security parameter $\kappa$ which is assumed to be linear in $\max\{\log |\FF|, n\}$.

We write our protocols using a series of numbered steps. The interpretation is that a party repeatedly goes through them in sequential order, attempting each instruction. Of course, in the asynchronous setting, the necessary precondition for carrying out an instruction might only be met much later in the protocol, and furthermore messages may be delayed. We discuss the latter below, in Section~\ref{sec:asynchronous}. Regarding the former, we frequently use statements of the form ``Wait until [condition]. Then, [instruction],'' which is to be executed at most once,  or ``If [condition], then [instruction],'' which can be executed multiple times.

\subsection{UC Basics}
\label{sec:uc}

We prove our constructions secure in the UC framework~\cite{canetti2020universally} and we briefly summarize it here. Protocol machines, ideal functionalities, the adversary, and the environment are all modeled as interactive Turing machine (ITM) instances, or ITIs. An execution of protocol $\Pi$ consists of a series of activations of ITIs, starting with the environment $\env$ who provides inputs to and collects outputs from the parties and the adversary $\adv$; parties can also give input to and collect output from sub-parties (e.g., hybrid functionalities), and $\adv$ can communicate with parties via ``backdoor'' messages. Although parties can send messages to one another, $\adv$ is responsible for delivering them, making this model \emph{completely} asynchronous (see Section~\ref{sec:asynchronous} for the way we capture eventual-delivery channels). Protocol instances are delineated using a unique session identifier (SID).

Corruption of parties is modeled by a special \texttt{corrupt} message sent from $\adv$ to the party; upon receipt of this message, the party sends its entire local state to $\adv$, and in all future activations follows the instructions of $\adv$. We emphasize that corruptions can occur at any point during the protocol (i.e., $\adv$ is assumed to be adaptive). Denote by $\exec_{\Pi,\adv,\env}$ the probability distribution ensemble corresponding to the (binary) output of $\env$ at the end of an execution of $\Pi$ with adversary $\adv$.

The ideal-world process for functionality $\func$ is simply defined as an execution of the ideal protocol $\ideal_{\func}$, in which the so-called ``dummy'' parties just forward inputs from $\env$ to $\func$ and forward outputs from $\func$ to $\env$ (in particular, the dummy parties do not communicate with the adversary, but rather the adversary is expected to send backdoor messages directly to $\func$, including corruption requests). Our functionalities implicitly respond to corruption requests in the standard way (e.g., all previous inputs of the newly corrupted party are leaked, and the adversary has the option to replace its output); we refer to~\cite{canetti2020universally} for further details. The adversary in the ideal world is also called the \emph{simulator} and denoted $\sim$; the corresponding ensemble is denoted by $\ideal_{\func,\sim,\env}$.

We are interested in unconditional (aka information-theoretic) security. Thus, we say that a protocol $\Pi$ (statistically) \textit{UC-realizes} an ideal functionality $\func$ if for any computationally unbounded adversary $\adv$, there exists a simulator $\sim$ (which is polynomial in the complexity of $\adv$), such that for any computationally unbounded environment $\env$, we have $\exec_{\Pi,\adv,\env}\approx\ideal_{\func,\sim,\env}$. In case the distribution ensembles are perfectly indistinguishable we say that $\Pi$ UC-realizes $\func$ with \emph{perfect} security.

When $\Pi$ is a $(\mathcal{G}_1,\ldots,\mathcal{G}_n)$-hybrid protocol (i.e., making subroutine calls to $\ideal_{\mathcal{G}_1},\ldots,\ideal_{\mathcal{G}_n}$), we say that $\Pi$ UC-realizes $\func$ in the $(\mathcal{G}_1,\ldots,\mathcal{G}_n)$-hybrid model.

\subsection{Modeling Eventual Message Delivery in UC}
\label{sec:asynchronous}

As mentioned above, the base communication model in UC is completely unprotected. To capture asynchronous networks with eventual message delivery, we work in a hybrid model with access to the multi-use \emph{asynchronous secure message transmission} functionality $\fasmt$, shown in Figure~\ref{fig:func-asmt}, which was introduced in~\cite{coretti2016constant} and is itself based on the (single-use) \emph{eventual-delivery secure channel} functionality $\fedsec$ from~\cite{katz2013universally}.

\begin{figure}[htbp!]
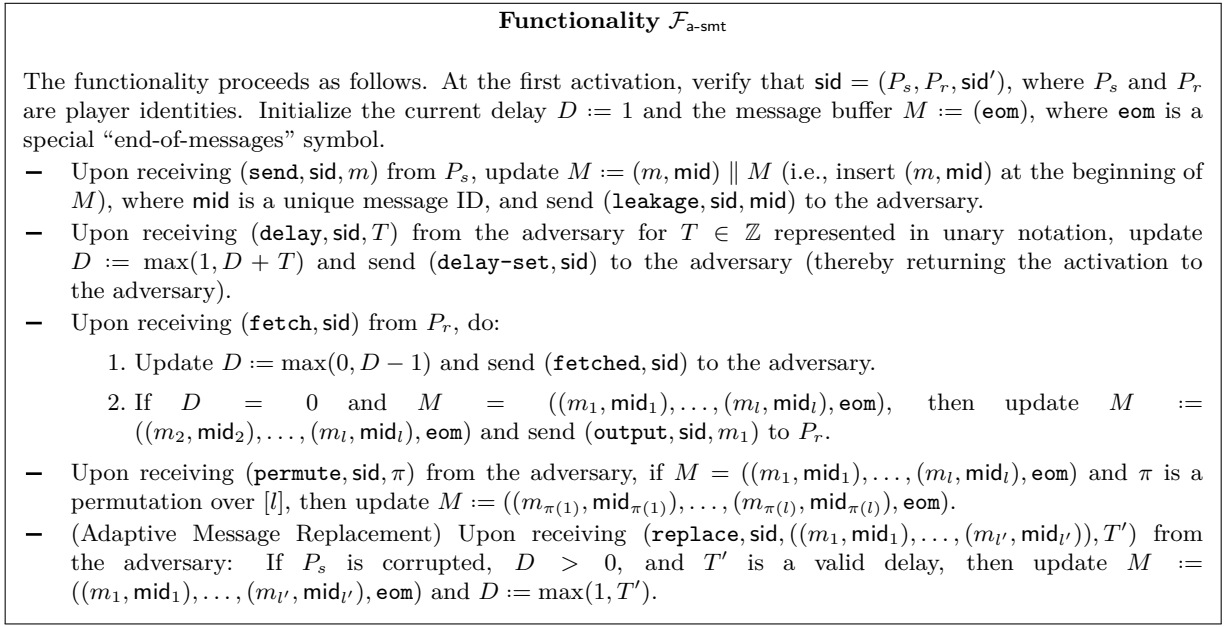

\centering
\funcbox{$\fasmt$}{
    The functionality proceeds as follows. At the first activation, verify that $\sid=(P_s,P_r,\sid')$, where $P_s$ and $P_r$ are player identities. Initialize the current delay $D\assign1$ and the message buffer $M\assign(\eom)$, where $\eom$ is a special ``end-of-messages'' symbol.
    \begin{tiret}
        \item Upon receiving $(\sendcmd,\sid,m)$ from $P_s$, update $M\assign (m,\MID)\parallel M$ (i.e., insert $(m,\MID)$ at the beginning of $M$), where $\MID$ is a unique message ID, and send $(\leakagecmd,\sid,\MID)$ to the adversary.
        \item Upon receiving $(\delaycmd,\sid,T)$ from the adversary for $T\in\ZZ$ represented in unary notation, update $D\assign \max(1,D+T)$ and send $(\delaysetcmd,\sid)$ to the adversary (thereby returning the activation to the adversary).
        \item Upon receiving $(\fetchcmd,\sid)$ from $P_r$, do:
        \begin{enumerate}
            \item Update $D\assign\max(0,D-1)$ and send $(\fetchedcmd,\sid)$ to the adversary.
            \item If $D=0$ and $M=((m_1,\MID_1),\ldots,(m_l,\MID_l),\eom)$, then update $M\assign((m_2,\MID_2),\ldots,(m_l,\MID_l),\eom)$ and send $(\outputcmd,\sid,m_1)$ to $P_r$.
        \end{enumerate}
        \item Upon receiving $(\permutecmd,\sid,\pi)$ from the adversary, if $M=((m_1,\MID_1),\ldots,(m_l,\MID_l),\eom)$ and $\pi$ is a permutation over $[l]$, then update $M\assign((m_{\pi(1)},\MID_{\pi(1)}),\ldots,(m_{\pi(l)},\MID_{\pi(l)}),\eom)$.
        \item (Adaptive Message Replacement) Upon receiving $(\replacecmd,\sid,((m_1,\MID_1),\ldots,(m_{l'},\MID_{l'})),T')$ from the adversary: If $P_s$ is corrupted, $D>0$, and $T'$ is a valid delay, then update $M\assign((m_1,\MID_1),\ldots,(m_{l'},\MID_{l'}),\eom)$ and $D\assign \max(1,T')$.
    \end{tiret}
    \smallskip
}
\caption{The asynchronous secure message transmission functionality.}
\label{fig:func-asmt}
\end{figure}

The functionality $\fasmt$ models a \emph{secure} eventual-delivery channel between a sender $P_s$ and receiver $P_r$.\footnote{Recall that while (concurrent) A-BA is not a private task, secure channels are needed to construct an OCC.} To reflect the adversary's ability to delay the message by an arbitrary finite duration (even when $P_s$ and $P_r$ are not corrupted), the functionality operates in a ``pull'' mode, managing a message buffer $M$ and a counter $D$ that represents the current message delay. This counter is decremented every time $P_r$ tries to fetch a message, which is ultimately sent once the counter hits $0$. The adversary can at any time provide an additional integer delay $T$, and if it wishes to immediately release the messages, it needs only to submit a large negative value. It is important to note that $T$ must be encoded in unary; this ensures that the delay, while arbitrary, remains bounded by the adversary's computational resources or running time.\footnote{We refer to~\cite{canetti2020universally} for a formal definition of running time in the UC framework.} Also, note that $\fasmt$ guarantees eventual message delivery assuming that the environment gives sufficient resources to the protocol, \ie activates $P_r$ sufficiently many times. We additionally allow the adversary to permute the messages in the buffer; see~\cite{coretti2016constant} for a discussion.

Since $\fasmt$ (and our other asynchronous functionalities) may not provide output immediately, we must clarify what we mean when instructing a party to fetch the output from multiple instances of a hybrid (e.g., one per party), as repeatedly attempting to fetch from the first instance before continuing to the second instance might not make any progress. Accordingly, we require that parties fetch the output in a ``round-robin fashion'' across activations (i.e., after attempting to fetch from each instance, the party moves on to try the remaining steps in the protocol). It is implicit that the party stops making fetch requests to any given instance after it receives an output from that instance (of course, this behavior can be easily implemented using flags).

Finally, we note that while the UC framework has no notion of time, and we are in the asynchronous setting (with eventual message delivery) where parties may proceed at different rates, one can still formally define a notion of asynchronous rounds along the lines of~\cite{coretti2016constant}, which will be referred to in our statements. Informally, rounds are delimited by alternations between sending messages to and reading messages from the network. More precisely, each party maintains a local round counter, which is incremented every time the party sends an input to a hybrid after fetching a (causally related) output from a hybrid. For example---and this accounts for virtually all of the times that a round advances during an execution of our protocols---a successful fetch request may directly trigger a send in the next activation, in which case the party increments its round counter. Note that we allow for multiple sends to be ``batched'' within the same round (i.e., only the first send will trigger a round advancement). We refrain from putting forth a formal treatment of causality; notwithstanding, causal relationships between messages will be immediately apparent from our protocol descriptions.

Denote by $R_{\adv,\env,z}$ the random variable (defined over the random choices of all relevant machines) representing the maximum number of rounds that an honest party takes in an execution with a \emph{fixed} adversary $\adv$, environment $\env$, and input $z$ to $\env$. We now define the (expected) round complexity of the protocol by considering the maximum expected value of this random variable, over all admissible adversaries, environments, and inputs: $\max_{\adv,\env,z}\{\mathbb{E}[R_{\adv,\env,z}]\}$.
\section{Ideal Functionalities for a Few Standard Primitives}
\label{sec:model}

In this section, we present ideal functionalities for asynchronous primitives used in our constructions. This seemingly simple task requires careful consideration of certain aspects, as the adversary's ability to delay messages in the network has an upstream impact on the achievable security guarantees of distributed tasks. In particular, the adversary can obstruct the output release procedure and impede honest parties' participation. To model \emph{delayed output release}, we extend the mechanism discussed in Section~\ref{sec:asynchronous} (using per-party delay counters). However, to model \emph{delayed participation}, we introduce a novel and more natural approach that addresses limitations of prior work and more closely captures the effects of asynchrony. We start by discussing in Section~\ref{sec:delayed_participation} our design choices for the latter, and proceed to describe a few ideal functionalities in Section~\ref{sec:functionalities}.

\subsection{Modeling Delayed Participation}
\label{sec:delayed_participation}

In the asynchronous setting, honest parties cannot distinguish whether uncooperative parties are corrupted and intentionally withholding messages, or if they are honest parties whose messages have been delayed. Consequently, when the number of corruptions is upper-bounded by $t$, waiting for the participation of the last $t$ parties can result in an indefinite wait. In ideal functionalities, this translates to expecting participation and/or input only from a subset of adversarially chosen parties, known as the ``core set,'' with a size of $n-t$. Observe that the value of $t$ is closely related to the behavior of the functionality. In this work, we specifically consider optimal resiliency, where $t=\lceil \frac{n}{3} \rceil -1$. However, our functionalities can be easily adapted to accommodate other resiliency bounds.

In our modeling approach, the adversary implicitly determines the core set by strategically delaying participation (or input submission) to the functionality. If we instruct the ideal functionality to proceed once $n-t$ parties have participated, the adversary can precisely determine the core set by manipulating the order of participation (or input submissions). To accommodate arbitrary but finite delays for input submissions, we employ a technique similar to the one we use for the output-release mechanism. That is, in addition to a per-party \emph{output delay counter}, there is an \emph{input delay counter} (updatable by the adversary) which is decremented every time the party pings the functionality; once the counter reaches zero, the party is allowed to participate.

This approach contrasts with the work of Cohen~\cite{cohen2016asynchronous} and Coretti \etal~\cite{coretti2016constant}, wherein the adversary (simulator) {\bf explicitly} sends a core set to the functionality. Obtaining the core set from the adversary all at once does not accurately mimic real-world executions and requires careful consideration to ensure the implementation works in all scenarios. For instance, a challenging case to model is when the simulator sets the core set, but the environment never activates some parties in that core set. If not handled properly, this situation can lead to either the ideal functionality stalling indefinitely while the real-world execution proceeds, or allowing for core sets of smaller sizes, neither of which is acceptable.

There are other aspects of modeling the core set that can potentially cause issues. If the simulator is required to set the core set early on, it may encounter issues during the simulation because in some real-world protocols, such as the asynchronous MPC protocol of Ben-Or \etal~\cite{ben-or1994asynchronous}, the core set is not fixed in the early stages of the execution. Similarly, if the functionality allows late submission of the core set, then when using the functionality as a hybrid, the adversary can potentially stall the functionality unless appropriate preventive mechanisms are in place. In contrast, implicitly defining the core set by delaying parties' participation aligns more closely with real-world executions, reducing the probability of errors. Additionally, it can potentially simplify the simulation process, as the simulator can gradually define the core set as the protocol progresses.

\subsection{Ideal Functionalities for a Few Standard Primitives}
\label{sec:functionalities}

We now cast a few standard asynchronous primitives as UC functionalities, following our novel modeling approach. We also present security statements showing how classical protocols can be used to realize these primitives.

\paragraph{Asynchronous Broadcast (A-Cast).}

The first essential primitive used in both our OCC and concurrent A-BA protocols, which also finds numerous applications in other asynchronous protocols, is {\em Bracha's Asynchronous Broadcast} (A-Cast)~\cite{bracha1987asynchronous}. A-Cast enables a distinguished sender to distribute its input, such that if an honest party outputs a value, then all honest parties must (eventually) output the same value. Moreover, if the sender is honest, then all honest parties must (eventually) output the sender's input. While these are essentially the agreement and validity properties required from regular (synchronous) broadcast, we stress that honest parties may not terminate when the sender is corrupted. We formulate A-Cast as the ideal functionality $\facast$, shown in Figure~\ref{fig:func-acast}.

\begin{figure}[htbp!]
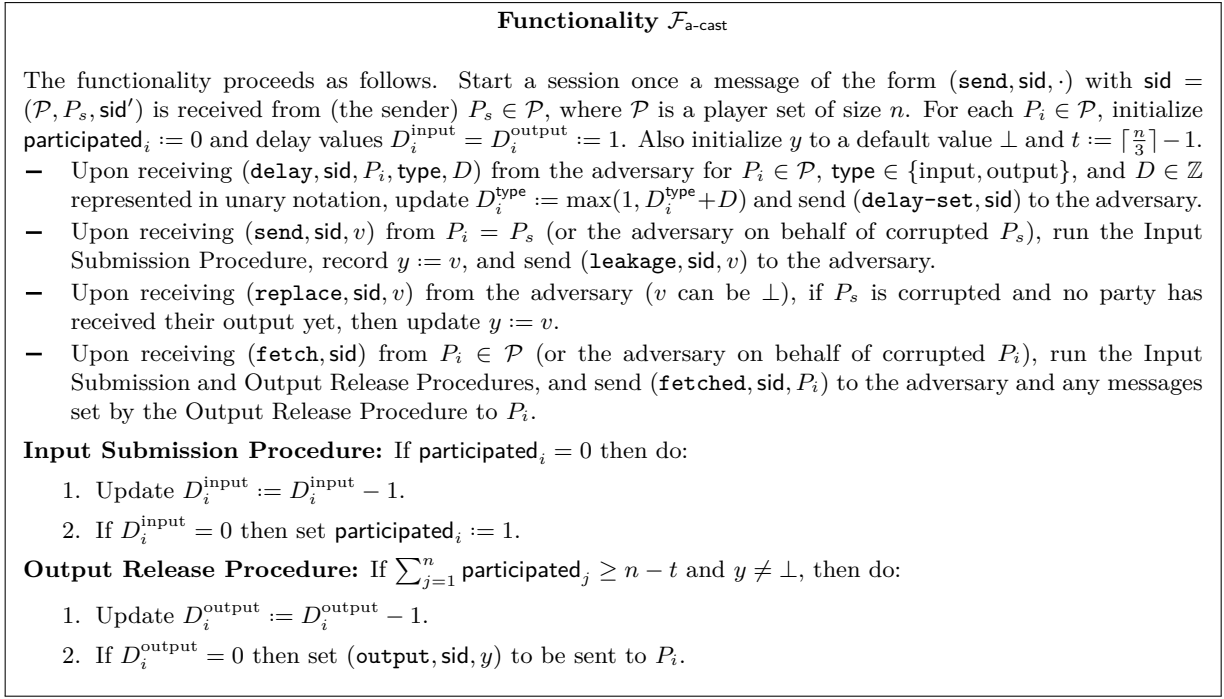

\centering
\funcbox{$\facast$}{
    The functionality proceeds as follows. Start a session once a message of the form $(\sendcmd,\sid,\cdot)$ with $\sid=(\P,P_s,\sid')$ is received from (the sender) $P_s\in \P$, where $\P$ is a player set of size $n$. For each $P_i\in\P$, initialize $\participated_i\assign 0$ and delay values $D^{\inputlbl}_i=D^{\outputlbl}_i\assign 1$. Also initialize $y$ to a default value $\bot$ and $t\assign \lceil \frac{n}{3} \rceil - 1$.
    \begin{tiret}
        \item Upon receiving $(\delaycmd,\sid,P_i,\type,D)$ from the adversary for $P_i\in\P$, $\type\in\{\inputlbl,\outputlbl\}$, and $D\in\ZZ$ represented in unary notation, update $D^{\type}_i\assign \max(1,D^{\type}_i+D)$ and send $(\delaysetcmd,\sid)$ to the adversary.
        \item Upon receiving $(\sendcmd,\sid,v)$ from $P_i=P_s$ (or the adversary on behalf of corrupted $P_s$), run the Input Submission Procedure, record $y\assign v$, and send $(\leakagecmd,\sid,v)$ to the adversary.
        \item Upon receiving $(\replacecmd, \sid,v)$ from the adversary ($v$ can be $\bot)$, if $P_s$ is corrupted and no party has received their output yet, then update $y\assign v$.
        \item Upon receiving $(\fetchcmd, \sid)$ from $P_i\in\P$ (or the adversary on behalf of corrupted $P_i$), run the Input Submission and Output Release Procedures, and send $(\fetchedcmd,\sid,P_i)$ to the adversary and any messages set by the Output Release Procedure to $P_i$.
    \end{tiret}
    \smallskip
    {\bf Input Submission Procedure:}
    If $\participated_i=0$ then do:
    \begin{enumerate}
        \item Update $D^{\inputlbl}_i\assign D^{\inputlbl}_i-1$.
        \item If $D^{\inputlbl}_i=0$ then set $\participated_i\assign1$.
    \end{enumerate}
    {\bf Output Release Procedure:}
    If $\sum_{j=1}^{n}\participated_j\geq n-t$ and $y\neq \bot$, then do:
    \begin{enumerate}
        \item Update $D^{\outputlbl}_i\assign D^{\outputlbl}_i-1$.
        \item If $D^{\outputlbl}_i=0$ then set $(\outputcmd,\sid,y)$ to be sent to $P_i$.
    \end{enumerate}
    \smallskip
}
\caption{The A-Cast functionality.}
\label{fig:func-acast}
\end{figure}

Note that although A-Cast is a single-sender primitive, assuming it can proceed to the output generation phase without sufficient participation from other parties is too idealized. A realizable functionality should only proceed to the output generation phase when $n-t$ parties (which may include the sender) have participated. Non-sender parties demonstrate their participation by issuing fetch requests to the functionality. An important technicality here is that the participation of parties before the sender initiates the session should not contribute to the count. Therefore, $\facast$ starts a session only once the input from the sender is received. Consequently, any efforts for participation before that point will not be taken into account. An implication of this design choice when using $\facast$ as a hybrid is that parties other than the sender will not know when the session has started.

Bracha's asynchronous broadcast protocol~\cite{bracha1987asynchronous} can be used to UC-realize $\facast$ with perfect security. We formally state this result in the following proposition, whose proof is straightforward.

\begin{proposition}
\label{thm:acast}
    $\facast$ can be UC-realized with perfect security in the $\fasmt$-hybrid model, in constant rounds and against an adaptive and malicious $t$-adversary, provided $t< \frac{n}{3}$.
\end{proposition}

\paragraph{Asynchronous Verifiable Secret Sharing (A-VSS).}

Another crucial primitive we require, mainly for our OCC protocol, is \emph{Asynchronous Verifiable Secret Sharing} (A-VSS). A-VSS allows a dealer to secret-share a value among all parties, ensuring that no unauthorized subset of colluding parties can learn any information about the secret. However, any authorized subset of parties should be able to efficiently reconstruct the secret using their shares. The term ``verifiable'' reflects that the dealer cannot cheat, for example by causing the reconstruction to fail or by inducing inconsistent output values from honest parties. In particular, whenever the sharing phase succeeds, any authorized subset of parties should be able to efficiently complete the reconstruction phase, and all honest parties doing so must recover the same secret. Moreover, if the dealer is honest, the sharing phase must always succeed, and everyone should recover the value originally shared by the dealer. In our context, we consider only the threshold access structure, where a subset of parties can recover the secret if and only if it contains at least $n-t$ of the parties. However, we leak the secret to the adversary as soon as $t+1$ parties have participated (at least one of whom must be honest). The formulation of A-VSS as an ideal functionality, $\favss$, is shown in Figure~\ref{fig:func-avss}. We denote this as the ``plain'' A-VSS functionality, as it does not allow homomorphic computation over shares; nonetheless, this basic functionality suffices for obtaining a multi-valued asynchronous OCC.

\begin{figure}[ht!]
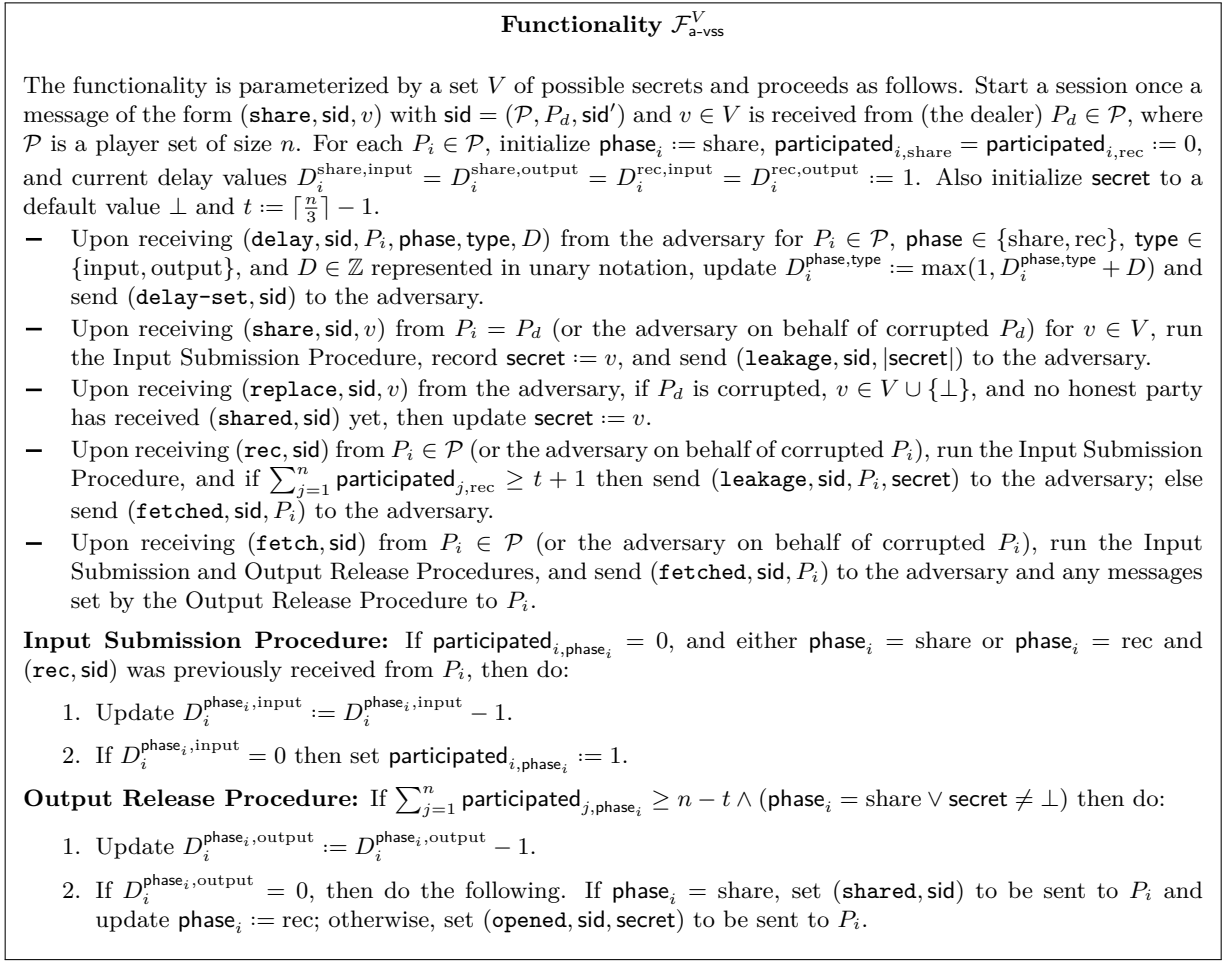

\centering

\funcbox{$\favss^V$}{
    The functionality is parameterized by a set $V$ of possible secrets and proceeds as follows. Start a session once a message of the form $(\sharecmd,\sid,v)$ with $\sid=(\P,P_d,\sid')$ and $v\in V$ is received from (the dealer) $P_d\in \P$, where $\P$ is a player set of size $n$. For each $P_i\in\P$, initialize $\phase_i\assign \sharelbl$, $\participated_{i,\sharelbl}=\participated_{i,\reclbl}\assign 0$, and current delay values $D^{\sharelbl,\inputlbl}_i=D^{\sharelbl,\outputlbl}_i=D^{\reclbl,\inputlbl}_i=D^{\reclbl,\outputlbl}_i\assign 1$. Also initialize $\secret$ to a default value $\bot$ and $t\assign \lceil \frac{n}{3} \rceil - 1$.
    \begin{tiret}
        \item Upon receiving $(\delaycmd,\sid,P_i,\phase,\type,D)$ from the adversary for $P_i\in\P$, $\phase\in\{\sharelbl,\reclbl\}$, $\type\in \{\inputlbl,\outputlbl\}$, and $D\in\ZZ$ represented in unary notation, update $D^{\phase,\type}_i\assign \max(1,D^{\phase,\type}_i+D)$ and send $(\delaysetcmd,\sid)$ to the adversary.
        \item Upon receiving $(\sharecmd,\sid,v)$ from $P_i=P_d$ (or the adversary on behalf of corrupted $P_d$) for $v\in V$, run the Input Submission Procedure, record $\secret\assign v$, and send $(\leakagecmd,\sid,|\secret|)$ to the adversary.
        \item Upon receiving $(\replacecmd,\sid, v)$ from the adversary, if $P_d$ is corrupted, $v\in V\cup \{\bot\}$, and no honest party has received $(\sharedcmd,\sid)$ yet, then update $\secret\assign v$.
        \item Upon receiving $(\reccmd,\sid)$ from $P_i\in \P$ (or the adversary on behalf of corrupted $P_i$), run the Input Submission Procedure, and if $\sum_{j=1}^{n}\participated_{j,\reclbl}\geq t+1$ then send $(\leakagecmd,\sid,P_i,\secret)$ to the adversary; else send $(\fetchedcmd,\sid,P_i)$ to the adversary.
        \item Upon receiving $(\fetchcmd, \sid)$ from $P_i\in\P$ (or the adversary on behalf of corrupted $P_i$), run the Input Submission and Output Release Procedures, and send $(\fetchedcmd,\sid,P_i)$ to the adversary and any messages set by the Output Release Procedure to $P_i$.
    \end{tiret}
    \smallskip
    {\bf Input Submission Procedure:} If $\participated_{i,\phase_i}=0$, and either $\phase_i=\sharelbl$ or $\phase_i=\reclbl$ and $(\reccmd,\sid)$ was previously received from $P_i$, then do:
    \begin{enumerate}
        \item Update $D^{\phase_i,\inputlbl}_i\assign D^{\phase_i,\inputlbl}_i-1$.
        \item If $D^{\phase_i,\inputlbl}_i=0$ then set $\participated_{i,\phase_i}\assign 1$.
    \end{enumerate}
    {\bf Output Release Procedure:} If $\sum_{j=1}^{n}\participated_{j,\phase_i}\geq n-t \land( \phase_i=\sharelbl \lor \secret\neq \bot)$ then do:
    \begin{enumerate}
        \item Update $D^{\phase_i,\outputlbl}_i\assign D^{\phase_i,\outputlbl}_i-1$.
        \item If $D^{\phase_i,\outputlbl}_i=0$, then do the following. If $\phase_i=\sharelbl$, set $(\sharedcmd,\sid)$ to be sent to $P_i$ and update $\phase_i\assign \reclbl$; otherwise, set $(\openedcmd,\sid,\secret)$ to be sent to $P_i$.
    \end{enumerate}
    \smallskip
}
\caption{The plain A-VSS functionality.}
\label{fig:func-avss}
\end{figure}

A-VSS, being a single-sender primitive, also requires the participation of at least $n-t$ parties to be realizable. We adopt a similar approach to $\facast$ in modeling this requirement. Parties other than the dealer demonstrate their participation in the sharing phase by issuing fetch requests, while participation in the reconstruction phase additionally requires sending a (dummy) input message. This ensures that the minimum participation threshold is met for the A-VSS protocol to proceed with each phase. Also, it is important to note that $\favss$ only initiates a session once it receives the first input (share) from the dealer. Any participation efforts made before that point are not taken into account.

The A-VSS protocol given by Canetti and Rabin in~\cite{canetti1993fast} can be used to UC-realize $\favss$ with statistical security for $t<n/3$. This result is formally stated in the following proposition. It is worth noting that perfectly secure A-VSS is impossible for $t\ge n/4$~\cite{ben-or1994asynchronous,DBLP:journals/dc/AbrahamDS22}.

\begin{proposition}
\label{thm:avss}
    For any finite field $\FF$ (with $|\FF|>n$), $\favss^{\FF}$ can be UC-realized with statistical security in the $\fasmt$-hybrid model, in constant rounds and against an adaptive and malicious $t$-adversary, provided $t<\frac{n}{3}$.
\end{proposition}

\paragraph{Asynchronous Byzantine Agreement (A-BA).}

We use A-BA for both binary and multi-valued domains $V$ in our revised concurrent A-BA protocol. In this primitive, each party $P_i$ has an input $v_i \in V$. The goal is for all honest parties to output the same value, such that if $n-2t$ input values are the same, that value (which for well-definedness we require is unique) is chosen as the output; otherwise, the adversary determines the output. We initially consider {\em corruption-unfair} A-BA, where the adversary learns the input of each party the moment it is provided. Corruption fairness, as introduced in~\cite{hirt2010adaptively} and later coined in~\cite{cohen2021completeness}, essentially ensures that the (adaptive) adversary cannot corrupt a party and subsequently influence the input value of that party based on its original input. It is worth noting that corruption-fair A-BA can easily be defined by avoiding leaking honest parties' inputs before the output is generated. We formulate A-BA as the ideal functionality $\faba$, shown in Figure~\ref{fig:func-aba}.

\begin{figure}[htbp!]
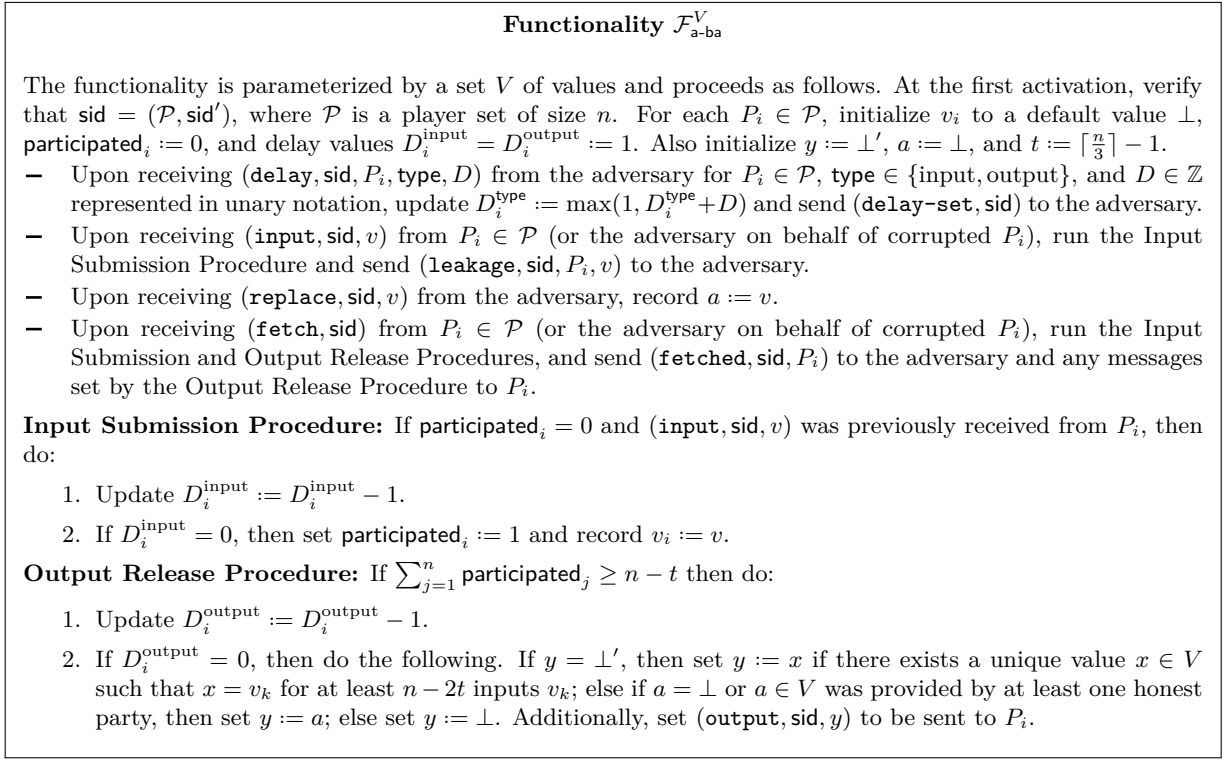

\centering
\funcbox{$\faba^V$}{
    The functionality is parameterized by a set $V$ of values and proceeds as follows. At the first activation, verify that $\sid=(\P,\sid')$, where $\P$ is a player set of size $n$. For each $P_i\in\P$, initialize $v_i$ to a default value $\bot$, $\participated_i\assign 0$, and delay values $D^{\inputlbl}_i=D^{\outputlbl}_i\assign 1$. Also initialize $y\assign\bot'$, $a\assign\bot$, and $t\assign \lceil \frac{n}{3} \rceil - 1$.
    \begin{tiret}
        \item Upon receiving $(\delaycmd,\sid,P_i,\type,D)$ from the adversary for $P_i\in\P$, $\type\in\{\inputlbl,\outputlbl\}$, and $D\in\ZZ$ represented in unary notation, update $D^{\type}_i\assign \max(1,D^{\type}_i+D)$ and send $(\delaysetcmd,\sid)$ to the adversary.
        \item Upon receiving $(\inputcmd,\sid,v)$ from $P_i\in \P$ (or the adversary on behalf of corrupted $P_i$), run the Input Submission Procedure and send $(\leakagecmd,\sid,P_i,v)$ to the adversary.
        \item Upon receiving $(\replacecmd, \sid,v)$ from the adversary, record $a\assign v$.
        \item Upon receiving $(\fetchcmd, \sid)$ from $P_i\in\P$ (or the adversary on behalf of corrupted $P_i$), run the Input Submission and Output Release Procedures, and send $(\fetchedcmd,\sid,P_i)$ to the adversary and any messages set by the Output Release Procedure to $P_i$.
    \end{tiret}
    \smallskip
    {\bf Input Submission Procedure:}
    If $\participated_i=0$ and $(\inputcmd,\sid,v)$ was previously received from $P_i$, then do:
    \begin{enumerate}
        \item Update $D^{\inputlbl}_i\assign D^{\inputlbl}_i-1$.
        \item If $D^{\inputlbl}_i=0$, then set $\participated_i\assign 1$ and record $v_i\assign v$.
    \end{enumerate}
    {\bf Output Release Procedure:}
    If $\sum_{j=1}^{n}\participated_j\geq n-t$ then do:
    \begin{enumerate}
        \item Update $D^{\outputlbl}_i\assign D^{\outputlbl}_i-1$.
        \item If $D^{\outputlbl}_i=0$, then do the following. If $y= \bot'$, then set $y\assign x$ if there exists a unique value $x\in V$ such that $x=v_k$ for at least $n-2t$ inputs $v_k$; else if $a=\bot$ or $a\in V$ was provided by at least one honest party, then set $y\assign a$; else set $y\assign \bot$. Additionally, set $(\outputcmd,\sid,y)$ to be sent to $P_i$.
    \end{enumerate}
    \smallskip
}
\caption{The (intrusion-tolerant) A-BA functionality.}
\label{fig:func-aba}
\end{figure}

The functionality $\faba$ encompasses an additional property known as ``non-intrusion'' validity (see Section~\ref{sec:results}). In a nutshell, this property guarantees that no malicious value can be present in the output. In other words, the output must be either an honest party's input or a default value~$\bot$. This stronger notion of A-BA is vital for the security of our concurrent A-BA protocol.

Note that when the input domain $V$ already contains a ``null'' value, $\faba$ can simply output it as the default value. Slightly abusing notation, we can therefore run multi-valued A-BA on some domain $V'\cup\{\bot\}$, with the assurance that a non-default (i.e., \emph{meaningful}) output value must have originated with an honest party. (We stress that, on the other hand, an output of $\bot$ would guarantee nothing in that scenario, even though $\bot$ is a possible input value.) Looking ahead, our concurrent A-BA protocol will make use of this behavior, by setting $V'$ to be the set of all output vectors and allowing an honest party who did not receive a vector to input $\bot$.

The expected-constant-round binary A-BA protocol of Canetti and Rabin~\cite{canetti1993fast} can be used to UC-realize $\faba^V$ with statistical security for binary domains ($|V|=2$) and $t<n/3$. However, our concurrent A-BA protocol (and the one in~\cite{ben-or2003resilient}) require {\bf multi-valued} A-BA, where $|V|$ is not constant (in fact, exponential), in expected-constant rounds. Ben-Or and El-Yaniv~\cite{ben-or2003resilient} claim that the constant-round reduction of multi-valued to binary BA proposed by Turpin and Coan~\cite{turpin1984extending}, which works in the synchronous setting for $t<n/3$, can be extended to work in the asynchronous setting by using A-Cast for message distribution. However, we demonstrate that the asynchronous version of this reduction works if and only if $t < n/6$, even when using A-Cast and considering a static adversary. (See Appendix~\ref{app:tc} for details.) Some additional modifications can improve this bound to $t < n/5$, but achieving optimal resiliency is not straightforward.

More recently, Most{\'{e}}faoui and Raynal~\cite{mostefaoui2017signature} presented a constant-round transformation from binary to multi-valued A-BA that works for $t < n/3$. Furthermore, the resulting protocol satisfies the non-intrusion validity property mentioned above. By applying this transformation to the binary A-BA protocol of Canetti and Rabin~\cite{canetti1993fast}, we can UC-realize $\faba^V$ for arbitrary $V$ with statistical security:

\begin{proposition}
\label{thm:aba}
    For any domain $V$, $\faba^V$ can be UC-realized with statistical security in the $\fasmt$-hybrid model, in expected-constant rounds and against an adaptive and malicious $t$-adversary, provided $t<\frac{n}{3}$.
\end{proposition} 
\section{Asynchronous Oblivious Common Coin}
\label{sec:occ}

The oblivious common coin (OCC) is a crucial primitive, particularly in the asynchronous setting. As discussed in the Introduction, achieving even binary A-BA against a single fail-stop corruption becomes impossible without the use of randomization~\cite{fischer1985impossibility}, emphasizing the importance of randomization in the asynchronous setting. The OCC primitive allows parties to collectively agree on a random value with a constant probability. This property opens up possibilities for various higher-level asynchronous protocols, as an OCC can potentially serve as the building block that provides the necessary randomization. To illustrate its significance, in the case of binary values, OCC enables binary A-BA in expected-constant rounds~\cite{rabin1983randomized}, and when extended to domains with size equal to the number of participants, it can be used for oblivious leader election (OLE) and sets the stage for expected-constant-round concurrent BA in both synchronous and asynchronous networks~\cite{ben-or2003resilient}. Refer to Section~\ref{sec:concurrent_protocol} for a formal definition of asynchronous OLE (A-OLE) and its application toward expected-constant-round concurrent A-BA.

As highlighted in the Introduction, none of the existing OCC proposals is simultaneously information-theoretic, asynchronous, multi-valued, and optimally resilient. Furthermore, no straightforward adaptation of the existing schemes yields an OCC with all of these properties. In Section~\ref{sec:OCC_functionality}, we first formulate (multi-valued) OCC as an ideal functionality in the asynchronous setting with eventual delivery. Subsequently, in Section~\ref{sec:OCC_protocol}, we propose our own OCC protocol that aims to satisfy all of these properties. In Section~\ref{sec:OCC_proof}, we prove that our protocol runs in constant time and UC-realizes the (multi-valued) OCC functionality with perfect security in the asynchronous setting, even against an adaptive and malicious adversary who corrupts fewer than $1/3$ of the parties.

At a high level, our protocol is based on the binary OCC of Feldman \cite{feldman1989asynchronous} and Canetti and Rabin \cite{canetti1993fast}, and incorporates a novel combinatorial technique derived from our observation stated in Lemma~\ref{lem:repetitions} below. By leveraging this lemma, we unveil interesting and powerful properties of the local views formed during the protocol's execution, leading to enhanced extraction capabilities. In fact, instead of extracting a single bit, by choosing appropriate parameters we can extract random values from any arbitrary domain still with a constant probability.

\subsection{A-OCC Ideal Functionality}
\label{sec:OCC_functionality}

An OCC is parameterized by a set $V$ and some constant probability $p>0$.  Each party starts with an empty input $\lambda$ and outputs a value from $V$, where with probability at least $p$ all parties output the same uniformly random value $x\in V$ and with probability $1-p$ the adversary chooses each party's output.\footnote{It is important to note that the term ``oblivious'' in this context refers to the fact that parties do not learn whether an agreement on a random coin value has been achieved or not, while the adversary does.}

The above goal can be translated to a UC functionality as follows: Initially, the ideal functionality samples a ``fairness bit'' $b\leftarrow \mathrm{Bernoulli}(p)$ and a random value $y\rsample V$. Then, if $b=1$ or no meaningful input is received from the adversary, it outputs $y$ to every party. However, if $b=0$ and meaningful input is received from the adversary, it assigns each party the value provided by the adversary; note that in this case, we allow the adversary to assign outputs gradually. The functionality also informs the adversary about the fairness bit and the random value once at least $t+1$ parties have participated (at least one of whom must be honest). Asynchronous aspects, including delayed output release and participation, are handled as in Section~\ref{sec:model}. The resulting functionality is shown in Figure~\ref{fig:func-aocc}.

\begin{figure}[ht!]
\centering
\funcbox{$\faocc^{V,p}$}{
    The functionality is parameterized by a set $V$ of possible outcomes and a fairness probability $p$, and it proceeds as follows. At the first activation, verify that $\sid=(\P,\sid')$, where $\P$ is a player set of size $n$. For each $P_i\in\P$, initialize $y_i$ to a default value $\bot$, $\participated_i\assign 0$, and delay values $D^{\inputlbl}_i=D^{\outputlbl}_i\assign 1$. Also initialize $a$, $y$, and $b$ to $\bot$, and $t\assign \lceil \frac{n}{3} \rceil - 1$.
    \begin{tiret}
        \item Upon receiving $(\delaycmd,\sid,P_i,\type,D)$ from the adversary for $P_i\in\P$, $\type\in\{\inputlbl,\outputlbl\}$, and $D\in\ZZ$ represented in unary notation, update $D^{\type}_i\assign \max(1,D^{\type}_i+D)$ and send $(\delaysetcmd,\sid)$ to the adversary.
        \item Upon receiving $(\inputcmd,\sid)$ from $P_i\in \P$ (or the adversary on behalf of corrupted $P_i$), run the Input Submission Procedure, and send $(\leakagecmd,\sid,P_i)$ to the adversary. Moreover, if $\sum_{j=1}^n\participated_j\ge t+1$, then additionally send $(\revealcmd,\sid,b,y)$ to the adversary.
        \item Upon receiving $(\replacecmd, \sid,v)$ from the adversary, record $a\assign v$.
        \item Upon receiving $(\fetchcmd, \sid)$ from $P_i\in\P$ (or the adversary on behalf of corrupted $P_i$), run the Input Submission and Output Release Procedures, and send $(\fetchedcmd,\sid,P_i)$ to the adversary and any messages set by the Output Release Procedure to $P_i$.
    \end{tiret}
    \smallskip
    {\bf Input Submission Procedure:} If $\participated_i=0$ and $(\inputcmd,\sid)$ was previously received from $P_i$, then do:
    \begin{enumerate}
        \item Update $D^{\inputlbl}_i\assign D^{\inputlbl}_i-1$.
        \item If $D^{\inputlbl}_i=0$ then set $\participated_i\assign1$.
        \item If $b = \bot$ then sample a ``fairness bit'' $b\leftarrow \mathrm{Bernoulli}(p)$ and a random value $y\rsample V$.
    \end{enumerate}
    {\bf Output Release Procedure:} If $\sum_{j=1}^{n}\participated_j\geq n-t$ then do:
    \begin{enumerate}
        \item Update $D^{\outputlbl}_i\assign D^{\outputlbl}_i-1$.
        \item If $D^{\outputlbl}_i=0$, then do the following. If $y_i=\bot$, then if $b=1$ or $a$ cannot be parsed as $(a_1,\ldots,a_n)\in V^n$, set $y_i\assign y$; otherwise, set $y_i\assign a_i$. Additionally, set $(\outputcmd,\sid,y_i)$ to be sent to $P_i$.
    \end{enumerate}
    \smallskip
}
\caption{The asynchronous OCC functionality.}
\label{fig:func-aocc}
\end{figure}

\subsection{The A-OCC Protocol}
\label{sec:OCC_protocol}

We proceed to present our asynchronous and multi-valued OCC protocol. We begin by discussing all the essential building blocks employed in our protocol. Subsequently, we provide a high-level overview of the protocol, highlighting its key ideas. This is followed by a detailed description.

\paragraph{Building blocks.}

The basic building blocks of our A-OCC protocol are A-VSS and A-Cast. A-VSS enables parties to contribute by privately providing their local randomness and only revealing this randomness when the contributions to the output are determined. Thus, A-VSS ensures the secrecy and verifiability of the shared secrets in an asynchronous setting. The A-VSS primitive is formally modeled as the ideal functionality $\favss$, described in Section~\ref{sec:functionalities}. On the other hand, A-Cast facilitates communication among parties by providing stronger guarantees than simple message distribution. This is especially crucial in asynchronous settings where challenges such as low message dispersion can occur. A-Cast helps in overcoming these challenges and ensures reliable message dissemination among the parties. We use the ideal functionality $\facast$, described in Section~\ref{sec:functionalities}, to model this primitive.

\paragraph{The protocol.}

As previously mentioned, our multi-valued protocol is built upon existing binary A-OCC constructions~\cite{feldman1989asynchronous,canetti1993fast} and introduces a novel combinatorial technique for extracting values from arbitrary domains. In both the binary protocol and our proposed protocol, each party secret-shares $n$ random elements from a field. It can be observed that at some point during the protocol execution, a vector of length $n$ consisting of random elements from the same field is established (with up to $t$ missing values due to asynchrony). For each coordinate of this vector, random elements shared by $t+1$ parties are utilized to prevent the adversary, controlling up to $t$ parties, from biasing any specific coordinate. Subsequently, each party starts reconstructing secrets shared by other parties to form the same vector locally. In the asynchronous setting, due to the low dispersion of messages, not all coordinates can be reconstructed by honest parties. This can result in different parties reconstructing different subsets of coordinates. However, by using mechanisms to improve message dispersion, as originally demonstrated by Feldman~\cite{feldman1989asynchronous}, it has been proven that when $t<n/3$, while the local vectors of honest parties may have up to $t$ missing coordinates, they have an overlap of size at least $n-2t>n/3$.\footnote{Feldman calculated the size of the overlap, denoted as $x$, based on the number of participants $n$ and the maximum number of corruptions $t$. The general relation is $x \geq n-t-\frac{t^2}{n-2t}$, which yields $x > n/3$ and $x > 5n/8$ when $t < n/3$ and $t < n/4$, respectively. This argument was later used in~\cite{canetti1993fast} to achieve optimal resiliency.} This is significant because without such mechanisms, and allowing for $t$ missing components when $n=3t+1$, the overlap in the local vectors of just four parties could be empty.

In fact, by incorporating the ``Gather'' protocol in~\cite{abraham2023reaching}, we can achieve an overlap of size at least $n-t>2n/3$, at the cost of an extra round of message exchange. While this stronger guarantee is not crucial to our feasibility result, we use it to improve the agreement probability of our OCC, and consequently the concrete round efficiency of our concurrent A-BA protocol.

Traditionally, existing protocols extract a single bit from the local views of the random vector by instructing parties to take all existing coordinates modulo $n$, outputting $0$ if any coordinate is $0$ and outputting $1$ otherwise. In contrast, our protocol represents a significant improvement by going beyond the extraction of a single bit from the local views of the random vector. This enhanced randomness extraction is through a combinatorial observation regarding vectors of random values, as formulated in Lemma~\ref{lem:repetitions}. This observation allows the parties to agree non-interactively on certain coordinates of the random vector with a constant probability, while also ensuring that these agreed-upon coordinates lie within their overlap section. The minimum such coordinate is then used to select a common output value from the vector.

Another important observation regarding existing binary A-OCC protocols is the lack of a proper termination mechanism. This poses significant challenges as network delays can cause parties to operate out of sync. In such cases, some parties may receive the output before completing their role in the execution, and if they stop, others may not be able to generate the output at all. This directly affects the simulator's ability to accurately simulate the protocol, especially in managing input and output delays in the ideal functionality. This is mainly because, unlike the protocol, the ideal functionality ensures that once a party receives the output, sufficient participation has occurred, and any other party, regardless of others' participation, can fetch the output if activated sufficiently many times. One potential solution could be invoking A-BA on the output at the end; however, this would create a circular dependency since A-OCC itself is used to achieve A-BA. Instead, we choose to adopt a simpler approach inspired by Bracha's termination mechanism. This approach resolves the participation issue without causing deadlocks and ensures agreement if all parties initiate the procedure with the same value. The formal description of our multi-valued asynchronous OCC protocol $\paocc$ appears in Figure~\ref{fig:prot-aocc}.

\begin{figure}[htbp]
\centering
\protbox{$\paocc^{V}$}{
The protocol is parameterized by a set $V$ of possible outcomes. At the first activation, verify that $\sid=(\P,\sid')$ for a player set $\P$ of size $n$. Let $t\assign\lceil\frac{n}{3}\rceil-1$ and $m\assign\lcm(n^2,|V|)$, and let $\FF$ be the smallest prime field with size at least $m$. Party $P_i\in\P$ proceeds as follows. Initialize sets $C_i$, $G_i$, $R_i$, and $S_i$ to $\emptyset$, flag $\finished$ to $0$, and variables $Z_i$ and $y_i$ to $\bot$. Also initialize $C_j'\assign \bot$, $b_{j,k}\assign 0$, $r_{k,j}\assign \bot$, $v_j\assign \bot$, and $z'_j\assign \bot$ for all $j,k\in[n]$.
\begin{tiret}
    \item Upon receiving input $(\inputcmd,\sid)$ from the environment, do the following. For each $j\in[n]$, choose $x_{i,j}\rsample[m]$ and send $(\sharecmd,\sid_i^{\sharelbl,j},x_{i,j})$ to an instance of $\favss^{\FF}$ with SID $\sid_i^{\sharelbl,j}\assign (\P,P_i,\sid',\texttt{share},j)$.
    \item Upon receiving input $(\fetchcmd,\sid)$ from the environment, if $\finished=1$ then output $(\outputcmd,\sid,y_i)$ to the environment; else do:
    \begin{enumerate}
        \item For each $j,k\in[n]$, fetch the share from the instance of $\favss^{\FF}$ with SID $\sid_j^{\sharelbl,k}$. Upon receiving back $(\sharedcmd,\sid_j^{\sharelbl,k})$ set $b_{j,k}\assign 1$, and if $b_{j,l}=1$ for all $l\in[n]$ then update $C_i\assign C_i\cup\{P_j\}$.
        \item Wait until $|C_i|\ge t+1$. Then, send $(\sendcmd,\sid_i^\attachlbl,C_i^{(t+1)})$ to an instance of $\facast$ with SID $\sid_i^\attachlbl\assign (\P,P_i,\sid',\texttt{attach})$.
        \item For each $j\in [n]$, fetch the output from the instance of $\facast$ with SID $\sid_j^\attachlbl$. Upon receiving back $(\outputcmd,\sid_j^\attachlbl,Y)$ where $|Y|=t+1$, record $C_j'\assign Y$ and wait until $C_j'\subseteq C_i$; then update $G_i\assign G_i\cup\{P_j\}$.
        \item Wait until $|G_i|\ge n-t$. Then, send $(\sendcmd,\sid_i^\setlbl,G_i^{(n-t)})$ to an instance of $\facast$ with SID $\sid_i^\setlbl\assign (\P,P_i,\sid',\setcmd)$.
        \item For each $j\in[n]$, fetch the output from the instance of $\facast$ with SID $\sid_j^\setlbl$. Upon receiving back $(\outputcmd,\sid_j^\setlbl,G_j')$ where $|G_j'|=n-t$, wait until $G_j'\subseteq G_i$; then update $R_i\assign R_i\cup\{P_j\}$.
        \item Wait until $|R_i|\ge n-t$. Then, send $(\sendcmd,\sid_i^\readylbl,R_i^{(n-t)})$ to an instance of $\facast$ with SID $\sid_i^\readylbl\assign (\P,P_i,\sid',\texttt{ready})$.
        \item For each $j\in[n]$, fetch the output from the instance of $\facast$ with SID $\sid_j^\readylbl$. Upon receiving back $(\outputcmd,\sid_j^\readylbl,R_j')$ where $|R_j'|=n-t$, wait until $R_j'\subseteq R_i$; then update $S_i\assign S_i\cup\{P_j\}$.
        \item Wait until $|S_i|\ge n-t$. Then, record $Z_i\assign G_i$.
        \item If $|S_i|\ge n-t$, then for each $P_k\in C_j'$ such that $P_j\in G_i$, send $(\reccmd,\sid_k^{\sharelbl,j})$ to the instance of $\favss^{\FF}$ with SID $\sid_k^{\sharelbl,j}$.\footnote{Note that reconstruction inputs can be sent outside of the initial batch, if a party gets added to $G_i$ later, but only one reconstruction input is ever sent per instance of $\favss^{\FF}$.}
        \item For each $P_k\in C_j'$ such that $P_j\in G_i$, fetch the secret from the instance of $\favss^{\FF}$ with SID $\sid_k^{\sharelbl,j}$. Upon receiving back $(\openedcmd,\sid_k^{\sharelbl,j},s)$ record $r_{k,j}\assign s$, and if $r_{l,j}\not=\bot$ for all $P_l\in C_j'$ then set $v_j\assign \sum_{P_l\in C_j'}r_{l,j} \bmod m$.
        \item Wait until $Z_i\not=\bot$ and $v_j\not=\bot$ for all $P_j\in Z_i$. Then, let $S\assign\{j\in[n] \mid v_j\equiv v_k\mod{n^2}\text{ for some }k\not=j\}$ and do the following. Let $z_i\assign v_{\min\{S\}}\bmod{|V|}$ (if $S=\emptyset$ then let $z_i$ be an arbitrary value from $V$) and send $(\sendcmd,\sid_i^{\termlbl},z_i)$ to an instance of $\facast$ with SID $\sid_i^{\termlbl}\assign (\P,P_i,\sid',\texttt{term})$.
        \item For each $j\in [n]$, fetch the output from the instance of $\facast$ with SID $\sid_j^{\termlbl}$. Upon receiving back $(\outputcmd,\sid_j^{\termlbl},x)$, record $z'_j\assign x$.
        \item Wait until $z'_j\not=\bot$ for at least $n-t$ indices $j\in[n]$. Then, set $\finished\assign1$ and $y_i$ to the mode (most repeated value) among all $z'_j\not=\bot$.
    \end{enumerate}
\end{tiret}
\smallskip
}
\caption{The multi-valued asynchronous OCC protocol.}
\label{fig:prot-aocc}
\end{figure}

\subsection{Security Proof}
\label{sec:OCC_proof}

Having described our A-OCC protocol, we proceed to present and prove the formal security statement that demonstrates how the protocol UC-realizes $\faocc$. However, we first prove a combinatorial observation regarding vectors of random values that facilitates the security proof. We formulate this observation separately in the following lemma, as it may be of independent interest.

\begin{lemma}
\label{lem:repetitions}
    Let $V$ be a vector of $n$ values chosen independently and uniformly at random from a set $S$ of size $N\in\Theta(n^2)$, and let $\alpha$ be a constant satisfying $0<\alpha\le1$. Then for any subset of indices $I\subseteq[n]$ such that $|I|\ge\alpha n$, with constant probability $p$ there is at least one repeated value in $V$; moreover, all of the repeated values are constrained to the indices in $I$.
\end{lemma}

\begin{proof}
It suffices to consider the case $|I|=\lceil\alpha n\rceil$. We first define the following events. Denote by $V_J$ the restriction of $V$ to the indices in $J\subseteq[n]$.
\begin{align*}
    &E_{V}\colon&&\text{At least two components of $V$ have the same value} \\
    &E_{V_{I}}\colon&&\text{At least two components of $V_{I}$ have the same value} \\
    &E_{V_{I^\complement}}\colon&&\text{At least two components of $V_{I^\complement}$ have the same value} \\
    &E_{(V_{I},V_{I^\complement})}\colon&&\text{At least one component of $V_{I}$ has the same value as a component of $V_{I^\complement}$}
\end{align*}
We now compute
\begin{align*}
    p&=\Pr[E_{V_I}\cap E^\complement_{V_{I^\complement}}\cap E^\complement_{(V_{I},V_{I^\complement})}] \\
    &= \Pr[E_{V_I}\mathbin| E^\complement_{V_{I^\complement}}\cap E^\complement_{(V_{I},V_{I^\complement})}]\Pr[E^\complement_{V_{I^\complement}}\cap E^\complement_{(V_{I},V_{I^\complement})}] \\
    & \geq \Pr[E_{V_I}\mathbin| E^\complement_{V_{I^\complement}}\cap E^\complement_{(V_{I},V_{I^\complement})}]\Pr[E^\complement_{V}] && \left(E^\complement_{V}\subseteq E^\complement_{V_{I^\complement}}\cap E^\complement_{(V_{I},V_{I^\complement})}\right) \\
    & = \left(1- \Pr[E^\complement_{V_I}\mathbin| E^\complement_{V_{I^\complement}}\cap E^\complement_{(V_{I},V_{I^\complement})}] \right)\Pr[E^\complement_{V}] \\
    & = \left(1- \frac{\Pr[E^\complement_{V_I}\mathbin\cap E^\complement_{V_{I^\complement}}\cap E^\complement_{(V_{I},V_{I^\complement})}]}{\Pr[E^\complement_{(V_{I},V_{I^\complement})}\mathbin|E^\complement_{V_{I^\complement}}]\Pr[E^\complement_{V_{I^\complement}}]} \right)\Pr[E^\complement_{V}] \\
    & = \left(1- \frac{\Pr[E^\complement_{V}]}{\Pr[E^\complement_{(V_{I},V_{I^\complement})}\mathbin |E^\complement_{V_{I^\complement}}]\Pr[E^\complement_{V_{I^\complement}}]} \right)\Pr[E^\complement_{V}] \\
    & = \left(1 - \frac{\prod_{i=0}^{n -1}\left(1-\frac{i}{N}\right)}{\left(1-\frac{n-\lceil \alpha n\rceil}{N} \right)^{\lceil \alpha n\rceil}\prod_{i=0}^{n -\lceil \alpha n \rceil -1}\left(1-\frac{i}{N}\right)}\right)\prod_{i=0}^{n -1}\left(1-\frac{i}{N}\right) \\
    &\approx \left(1- \frac{ \prod_{i=0}^{n -1}\left(e^{-\frac{i}{N}}\right)}{\left( e^{-\frac{n-\lceil \alpha n \rceil}{N}} \right)^{\lceil \alpha n \rceil}\prod_{i=0}^{n -\lceil \alpha n \rceil -1}\left(e^{-\frac{i}{N}}\right)}\right)\prod_{i=0}^{n -1}\left(e^{-\frac{i}{N}}\right) && \left(\text{since }\frac{n}{N}\ll 1\right)\\
    & = \left(1- \frac{ e^{-\frac{1}{N}\sum_{i=0}^{n -1}i}}{e^{-\frac{n \lceil \alpha n \rceil -\lceil \alpha n \rceil^2}{N}}\left(e^{-\frac{1}{N}\sum_{i=0}^{n -\lceil \alpha n \rceil -1}i} \right)}\right)e^{-\frac{1}{N}\sum_{i=0}^{n -1}i} \\
    & = \left(1- \frac{ e^{-\frac{n^2-n}{2N}}}{e^{-\frac{n \lceil \alpha n \rceil -\lceil \alpha n \rceil^2}{N}}\left(e^{-\frac{n^2 + \lceil \alpha n \rceil^2 -2n\lceil \alpha n \rceil -n + \lceil\alpha n \rceil}{2N}} \right)}\right)e^{-\frac{n^2-n}{2N}} \\
    & = \left(1-e^{-\frac{\lceil\alpha n\rceil^2-\lceil\alpha n\rceil}{2N}}\right)e^{-\frac{n^2-n}{2N}} \\
    & \in \Theta(1). && \left(\textrm{since}~N\in\Theta(n^2)\right)
    \qedhere
\end{align*}
\end{proof}

With the groundwork laid out, we now prove the security of protocol $\paocc$:

\begin{theorem}
\label{thm:occ}
    There exists a probability $p\in \Theta(1)$ such that for any integer domain $V$, protocol $\paocc^V$ UC-realizes $\faocc^{V,p}$ with perfect security in the $(\facast,\allowbreak\favss^\FF)$-hybrid model where $\FF$ is the smallest prime field of size at least $\lcm(|V|,\allowbreak n^2)$, in constant rounds and in the presence of an adaptive and malicious t-adversary, provided $t < \frac{n}{3}$.
\end{theorem}

\begin{proof}
Let $\adv$ be an adversary in the real world. We construct a simulator $\sim$ in the ideal world, such that no environment $\env$ can distinguish whether it is interacting with $\paocc^V$ and $\adv$, or with $\faocc^{V,p}$ and $\sim$, where $p$ is to be determined later.

The simulator internally runs a copy of $\adv$, and plays the roles of $\facast$, $\favss^\FF$, and the parties in a simulated execution of the protocol. All inputs from $\env$ are forwarded to $\adv$, and all outputs from $\adv$ are forwarded to $\env$. Moreover, whenever $\adv$ corrupts a party in the simulation, $\sim$ corrupts the same party in the ideal world by interacting with $\faocc^{V,p}$, and if the corruption was direct (i.e., not via either of the aiding functionalities), then $\sim$ sends $\adv$ the party's state and thereafter follows $\adv$'s instructions for that party. Moreover, $\sim$ manages delays in the ideal functionality based on the delays that $\adv$ sets on the hybrids.

The simulated execution commences when $\sim$ receives the messages $(\leakagecmd, \sid, P_i)$ and $(\revealcmd, \sid, b, y)$ from $\faocc$. During the internal execution, $\sim$ faithfully simulates the behavior of $\facast$. As for $\favss^\FF$, $\sim$ can simulate the sharing phase without explicitly assigning any random secrets to honest parties. However, the actual value of each secret must be determined by $\sim$ before $n-t$ parties participate in the reconstruction phase of the corresponding $\favss^\FF$ instance. It is crucial for $\sim$ to select these secrets in such a way that the output matches that of $\faocc$, while ensuring indistinguishability between the real- and ideal-world views.

We first argue that every tally $v_j$ computed by an honest party $P_i$ is uniformly distributed in $[m]$, and moreover that another honest party $P_h$ cannot compute a \emph{different} tally $v_j'$ for party $P_j$ (although up to $t$ tallies may be ``missing'' from $P_h$'s view). Recall that $P_i$ computes $v_j$ by summing the reconstructed secrets ``attached'' to $P_j$ (i.e., the votes cast by the parties in $C_j'$ for $P_j$). An honest party $P_k\in C_j'$ will have cast a uniform vote $x_{k,j}\in[m]$, but even for a corrupted party a well-defined field element will be reconstructed as its vote, since $P_i$ validated that $C_j'\subseteq C_i$ and $C_i$ contains parties whose secrets were successfully shared to $P_i$. Moreover, the votes of corrupted parties in $C_j'$ are independent of those cast by honest parties in $C_j'$, because the secrets attached to $P_j$ are fixed by the time the first honest party $P_h$ ``accepts'' $P_j$ (i.e., adds $P_j$ to $G_h$)---due to the guarantees of $\favss^{\FF}$---while no honest party starts reconstructing them until it considers $P_j$ as accepted. Since $|C_j'|=t+1$, it must contain at least one honest party, and therefore $P_i$ will compute a uniform tally $v_j\in[m]$. Furthermore, since $P_j$ A-Casts the set $C_j'$, all honest parties will agree on the parties contributing to $P_j$'s tally, and since any honest party $P_h$ computing $v_j$ has $C_j'\subseteq C_h$, by the guarantees of $\favss^{\FF}$ they will all reconstruct the secrets attached to $P_j$ to the same values. Therefore, we can indeed speak of a \emph{fixed} vector of uniform tallies, where each honest party eventually holds a local view of at least $n-t$ coordinates.

By employing an enhanced version of the counting argument in~\cite{feldman1989asynchronous,canetti1993fast} in the same way as~\cite{abraham2023reaching}, we can demonstrate that once the first honest party $P_i$ reaches the stage where $|S_i|\ge n-t$, a subset $A$ of parties, with size $n-t>2n/3$, becomes fixed such that each honest party $P_h$ will (eventually) have $A\subseteq Z_h$. That is, over $2/3$ of the tallies will be visible to all honest parties by the time they are able to generate output. In more detail, assume that the honest party $P_i$  reaches the stage where $s\assign|S_i|\ge n-t$. We can construct an $s \times n$ table $T$, where $T_{j,k}=1$ if and only if $P_j \in S_i$ and $P_k \in R_j'$ (where $R_j'$ was received from $P_j$'s A-Cast).\footnote{Technically the table has $n$ rows, each corresponding to an A-Casted $R$-set, and $P_i$ has a local view of $s$ rows that it has validated against $R_i$.} Each row of the table contains (exactly) $n-t$ entries equal to one, indicating that there are a total of $s(n-t)$ one-entries in $T$. We claim that at least one column must contain at least $t+1$ one-entries. Indeed, suppose to the contrary that each column contains at most $t$ one-entries. Since $n-t\le s$ and $t<n/3$, we have
\[
\frac{4}{9}n^2<(n-t)^2\le s(n-t)\le nt<\frac{1}{3}n^2,
\]
which is a contradiction.

Therefore, we can conclude that some column $k$ in ($P_i$'s view of) the table contains at least $t+1$ one-entries. Now consider another honest party $P_h$ who reaches the stage $|S_h|\ge n-t$ and thereby records $Z_h$. Since $(n-t)+(t+1)>n$, the \kth column in $P_h$'s view of the table must necessarily contain a one-entry. From the specification of the protocol, we can deduce that $P_k\in R_h$, and hence $G_k'\subseteq Z_h$ (where $G_k'$ was received from $P_k$'s A-Cast). In other words, there is a set $A\assign G_k'$ of size $n-t$ such that every honest party $P_h$ will eventually have $A\subseteq Z_h$.

Once the first honest party $P_i$ reaches the condition $|S_i|\ge n-t$, the simulator $\sim$, being aware of message delivery, can easily determine the set $A$. Basically, $\sim$ can form the table $T$ as described above, find a column $k$ with at least $t+1$ ones, and then set $A\assign G_k'$ (or rather the indices of the parties in $G_k'$). Utilizing the set $A$, and the values $b$ and $y$ received from $\faocc$, the simulator is able to sample a vector of tallies that follows the same distribution as in the real-world execution and results in the same output as $\faocc$ when $b=1$. Let us now outline the process for this inverse sampling. We first define the following events, given the set $A$ and a random vector $(w_1,\ldots,w_n)$ where each element $w_i\rsample [m]$.

\begin{itemize}
    \item $E_1$: There is at least one repeat in the vector $(w'_1,\ldots,w'_n)$, where $w'_i = w_i \bmod{n^2}$ for all $i\in [n]$, and all the repeats lie inside $A$.
    \item $E_{1,y}$ for $y\in V$: The vector $(w_1,\ldots,w_n)$ is in $E_1$, and the smallest index $j$ of the repeated elements in $(w'_1,\ldots,w'_n)$, where $w'_i = w_i \bmod{n^2}$ for all $i\in [n]$, contains an element $w_j$ congruent to $y$ modulo $|V|$.
    \item $E_2$: There are no repeats at all in the vector $(w'_1,\ldots,w'_n)$, where $w'_i = w_i \bmod{n^2}$ for all $i\in [n]$.
    \item $E_3$: There is at least one repeat in the vector $(w'_1,\ldots,w'_n)$, where $w'_i = w_i \bmod{n^2}$ for all $i\in [n]$, which is not fully contained inside $A$.
\end{itemize}

The simulator $\sim$ initially computes the precise probabilities of the events $E_1$, $E_2$, and $E_3$.\footnote{These probabilities only need to be computable by the simulator to give a constructive description of $\sim$. Note that an existential argument, which is sufficient to complete the emulation, would not require $\sim$ to calculate those probabilities at all. Nonetheless, we opt for a constructive approach to better illustrate the concept.} Given that $m = \lcm(n^2, |V|)$, sampling vectors uniformly from $[m]^n$ and taking each coordinate modulo $n^2$ is equivalent to uniform sampling from $[n^2]^n$. Therefore, we can compute the aforementioned probabilities by considering vectors exclusively from $[n^2]^n$. Let $l=|A|$. For $1\le k\le l$, denote by $N_k$ the number of ways to choose the elements inside $A$ from any fixed subset of $k$ residues modulo $n^2$ such that \emph{exactly} $k$ distinct residues are chosen (equivalently, the number of strings of length $l$ over an alphabet of size $k$ that use each character at least once). We have the following recursive expression for $N_k$.
\[
N_k=k^l-\sum_{i=1}^{k-1}\binom{k}{i}N_i
\]
Note that $N_1,\ldots,N_l$ are efficiently computable. Moreover, for any $1\le k\le l$, the number of ways to choose the elements in $(w'_1,\ldots,w'_n)$ from $[n^2]$ such that $A$ contains exactly $k$ distinct elements and all repeats (if any) lie inside $A$ is
\[
\binom{n^2}{k}N_k\prod_{j=0}^{n-l-1}(n^2-k-j).
\]
Noting that $k=l$ exactly corresponds to the case that there are no repeats at all, we now compute $p_i\assign\Pr[E_i]$ for $i\in [3]$:
\begin{align*}
    p_1&=\frac{1}{(n^2)^n}\sum_{k=1}^{l-1}\left(\binom{n^2}{k}N_k\prod_{j=0}^{n-l-1}(n^2-k-j)\right)\\
    p_2&=\frac{1}{(n^2)^n}\prod_{j=0}^{n-1}(n^2-j)\\
    p_3&=1-p_1-p_2
\end{align*}

Now we present a sampling strategy utilized by $\sim$ to ensure that the vector of tallies is sampled in a manner that maintains both the observed distribution in the real-world execution and the consistency of outputs with the ideal functionality. This strategy is effective under the condition that the probability of outputting a random common coin in the ideal functionality, denoted as $p$, does not exceed the probability defined above as $p_1$, which is the probability with which the protocol guarantees a random common coin. Looking ahead, $\sim$ is required to perform a task with a probability denoted as $q_1$, which becomes negative if $p > p_1$, thereby rendering the entire process unsound. For any $p \le p_1$, the simulator does the following:
\begin{itemize}
    \item If $\faocc$ outputs $b=1$, uniformly sample from $E_{1,y}$ where $y$ is the value received from $\faocc$.
    \item If $b=0$ is received from $\faocc$:
    \begin{itemize}
        \item With probability $q_1=\frac{p_1-p}{1-p}$, uniformly sample from $E_1$.
        \item With probability $q_2=\frac{p_2}{1-p}$, uniformly sample from $E_2$.
        \item With probability $q_3=\frac{p_3}{1-p}$, uniformly sample from $E_3$.
    \end{itemize}
\end{itemize}

Before delving into the details of the sampling procedures, it is important to understand why this inverse sampling approach effectively eliminates any distinguishability. The simulator samples from $E_{1,y}\subseteq E_1$ when $\faocc$ outputs $b=1$, and from $E_1$ with probability $q_1$ when $\faocc$ outputs $b=0$. Since $\faocc$ samples $b\leftarrow \mathrm{Bernoulli}(p)$, the simulator samples a vector from $E_1$ with probability $p+(1-p)q_1=p_1$, which exactly matches the probability in the real-world execution. Moreover, when $b=1$, the simulator ensures that the element in the sample at the smallest index of a repeated residue is congruent modulo $|V|$ to the value $y$ received from $\faocc$, which guarantees that the simulated view is consistent with the output of $\faocc$. It is important to note that enforcing this constraint does not skew the probability since $\faocc$ samples $y$ uniformly at random from $V$ and $m$ is a multiple of $|V|$. In cases when $b=0$, the simulator influences $\faocc$ to output values that match those in the simulated execution. Additionally, $\sim$ samples from $E_2$ and $E_3$ with probabilities $q_2$ and $q_3$, respectively. Since $\faocc$ outputs $b=0$ with probability $1-p$, the simulator samples from $E_2$ and $E_3$ with respective probabilities $(1-p)q_2=p_2$ and $(1-p)q_3=p_3$, which exactly match the corresponding real-world probabilities. Therefore, this inverse sampling strategy ensures that there are no distinguishability opportunities for $\env$.

Now we explain how each sampling procedure works. The general approach is to first determine the ``repeat pattern,'' which is encoded as a vector $(r_1,\ldots,r_n)\in[n^2]^n$ of residues modulo $n^2$, consistently with the event being sampled from, and then fill in each coordinate of the vector $(w_1,\ldots,w_n)\in[m]^n$ by sampling from the corresponding congruence class (with one caveat, discussed below). This is of course equivalent to sampling the elements in $(w_1,\ldots,w_n)$ such that their residues modulo $n^2$ have the appropriate repeat pattern, but it greatly simplifies the presentation and also enables us to cast the procedure for sampling from $E_{1,y}$ as a variation on the procedure for sampling from $E_1$.

\begin{itemize}
    \item Procedure 1 (sampling from $E_1$ or $E_{1,y}$):
    \begin{enumerate}
        \item Initialize $(w_1,\ldots,w_n)$ and $(r_1,\ldots,r_n)$ to $(\bot,\ldots,\bot)$.
        \item Randomly choose two distinct indices $g$ and $h$ from $A$.\footnote{Note that this is well-defined since $|A|=n-t>2$ when $1\le t<n/3$.} Sample $r_g=r_h\rsample[n^2]$.
        \item For each $i\in A\setminus\{g,h\}$, sample $r_i\rsample[n^2]$.
        \item Fill in the remaining coordinates of $(r_1,\ldots,r_n)$ with random \emph{distinct} residues that have not yet been chosen. That is, for each $j\in [n]\setminus A$, sample $r_j\rsample[n^2]\setminus\cup_{k=1}^n\{r_k\}$.
        \item[\bf 4.5.] \textbf{(Only if sampling from $E_{1,y}$)} Let $S\assign\{j\in[n]\mid r_j=r_k\text{ for some }k\not=j\}$, and sample $x\rsample\{v\in[m]\mid v\equiv y\mod{|V|}\}$. Set $w_{\min\{S\}}\assign x$. Then, randomly choose a permutation $\pi$ over $[n^2]$ such that $\pi(r_{\min\{S\}})=x\bmod{n^2}$, and update $r_i\assign\pi(r_i)$ for each $i\in[n]$.
        \item For each $i\in[n]$ such that $w_i=\bot$, sample $w_i\rsample\{v\in[m]\mid v\equiv r_i\mod{n^2}\}$.
    \end{enumerate}
    \item Procedure 2 (sampling from $E_2$):
    \begin{enumerate}
        \item Initialize $(w_1,\ldots,w_n)$ to $(\bot,\ldots,\bot)$.
        \item Randomly choose $n$ distinct residues $r_1,\ldots,r_n$ from $[n^2]$.
        \item For each $i\in[n]$, sample $w_i\rsample\{v\in[m]\mid v\equiv r_i\mod{n^2}\}$.
    \end{enumerate}
    \item Procedure 3 (sampling from $E_3$):
    \begin{enumerate}
        \item Initialize $(w_1,\ldots,w_n)$ and $(r_1,\ldots,r_n)$ to $(\bot,\ldots,\bot)$.
        \item Randomly choose an index $g$ from $[n]\setminus A$ and an index $h$ from $[n]\setminus \{g\}$. Sample $r_g=r_h\rsample[n^2]$.
        \item For each $i\in [n]\setminus \{g,h\}$, sample $r_i\rsample[n^2]$.
        \item For each $i\in[n]$, sample $w_i\rsample\{v\in[m]\mid v\equiv r_i\mod{n^2}\}$.
    \end{enumerate}
\end{itemize}

The simulator uses Procedure 1 to sample from either $E_1$ or $E_{1,y}$. In the former case, there is no need to ensure that the sample is consistent with any particular output value. Thus, $\sim$ can randomly ``plant'' a repeat inside $A$, sample the rest of the residues appearing within $A$ totally at random, and lastly round out the repeat pattern by selecting unused residues for the indices outside of $A$. This process guarantees that we sample from $E_1$ with the uniform distribution.

The latter case, in which $\sim$ must sample from $E_1$ while ensuring that the element at the smallest index of a repeated residue is congruent to $y$ modulo $|V|$, requires a more complex approach. It is not enough to simply follow Procedure 1 and then readjust the repeated residue with minimum index (as well as the elements in the vector $(w_1,\ldots,w_n)$ from the corresponding congruence class) at the end. The reason is that adjusting values after sampling can potentially skew the distribution of the repeat pattern. If we adjust only the values of existing repeats, there is a chance that the adjusted value may match another repeated value, altering the distribution by decreasing the number of distinct repeated values and increasing the number of times a value is repeated. It is even possible that the adjusted value will collide with a value outside of the designated set $A$! On the other hand, introducing new repeats would increase the number of distinct repeated values and once again deviate from the desired distribution. Addressing these challenges requires careful consideration and decision-making on the part of $\sim$, as determining when to add new repeats and when to readjust existing ones is a delicate task.

Fortunately, it is possible for $\sim$ to carefully re-sample the repeat pattern, preserving the overall \emph{structure} but ensuring consistency with $y$. Since taking the repeated value with minimum index is a global criterion, this two-pass strategy is essentially the best we can hope for. In more detail, $\sim$ starts by determining the (preliminary) repeat pattern as if it were sampling from $E_1$; in this way, $\sim$ obtains a correct distribution of repeats naturally induced by random values. However, $\sim$ now samples a random element $x\in[m]$ that is congruent to $y$ modulo $|V|$, which is used to reassign the congruence class (modulo $n^2$) associated with the smallest index of a repeated value in the old repeat pattern (note that this new congruence class is not uniquely determined by $y$). The set of distinct residues appearing in the repeat pattern is then re-sampled around this constraint; in other words, we apply a random permutation over residues with the property that the appropriate residue is mapped to $x\bmod{n^2}$. The simulator now uses the modified repeat pattern to sample the elements in $(w_1,\ldots,w_n)$, except that the element at the smallest index of a repeated residue should be fixed to $x$. By following this procedure, $\sim$ is guaranteed to sample uniformly from $E_{1,y}$, for any $y\in V$.

Procedure 2 is straightforward: $\sim$ forms the repeat pattern by choosing random residues, ensuring that no value is used more than once. This simple strategy guarantees that the sample is drawn from $E_2$ and maintains a uniform distribution.

Finally, in Procedure 3, the simulator first randomly plants a repeat involving an index outside of $A$ (it does not matter whether the second index is inside $A$ or not). Then, the remaining residues can be chosen randomly without any restrictions. This process guarantees a uniform sample from $E_3$.

That completes the inverse sampling process. With the knowledge of $A$, $b$, and $y$, the simulator is able to sample a vector $(w_1, \ldots, w_n)$ that is indistinguishable from the observations in the real-world execution. Furthermore, the sampled vector is consistent with the output $y$ of the ideal functionality when the fairness bit $b=1$.

After sampling the vector of tallies $(w_1, \ldots, w_n)$, the simulator retains the flexibility to determine the secrets shared by honest parties. This is possible because none of the honest parties engage in the reconstruction phase of any $\favss$ instance until the first honest party $P_i$ satisfies the condition $|S_i|\ge n-t$. Consequently, $\sim$ has sufficient time to make decisions regarding the shared secrets without the need to leak them to the adversary. At this point, $\sim$ can initiate the opening of secrets through the reconstruction phase of $\favss$ in such a way that if an honest party recovers the \jth tally $v_j$, it will correspond to $v_j=w_j$. This can be accomplished by leveraging the fact that each tally consists of secrets from $t+1$ parties, which must include at least some honest parties. By appropriately adjusting the secrets of the honest parties, $\sim$ can ensure that the tally sums up to the desired value. It is crucial to note that, according to the inverse sampling process described above, the vector $(w_1,\ldots,w_n)$ maintains the same distribution as in the real-world execution. Therefore, sampling random secrets while conditioning them on the desired tallies does not introduce any distinguishability in the views.

\paragraph{Correctness.} The above approach allows for generating an output in the simulated execution matched with the output of the ideal functionality and results in a view indistinguishable from the real-world execution if the inverse sampling process works correctly. This is the case because when $\faocc$ samples $b=1$ and $y$, the simulator samples the vector of tallies from $E_{1,y}$ which guarantees the same output $y$ for all the honest parties. On the other hand, when $b=0$ is sampled by $\faocc$, $\sim$ is allowed to influence the output of $\faocc$. The simulator, who observes the simulated execution, can seize this opportunity and exert influence to ensure that the output of $\faocc$ matches the output of the simulated execution. As mentioned above, the inverse sampling process relies on the fact that $p \le p_1$. By a direct utilization of Lemma~\ref{lem:repetitions}, we can show that there is a constant lower bound for $p_1$ for any set $A$ with size $|A|\ge \frac{2}{3}n$. Therefore, there exists some constant $p$ such that $p \le p_1$, which makes the simulation go through correctly.

\paragraph{Termination.} To ensure the termination of the protocol and maintain the indistinguishability of the real and ideal worlds, we need to demonstrate that the protocol will eventually terminate.

For each honest party $P_j$, the secret-sharing of all of its votes $x_{j,1},\ldots,x_{j,n}$ will eventually succeed for any honest party $P_i$. It is certainly the case then that every honest party $P_i$ will eventually pass the threshold $|C_i|\ge t+1$ and A-Cast $C_i^{(t+1)}$.

Whenever an honest party $P_h$ adds $P_j$ to $C_j$, any other honest party $P_i$ will eventually add $P_j$ to $C_i$. This is because $\favss$ guarantees that if the sharing of a secret succeeds for one honest party, then it will eventually succeed for all honest parties. Since every honest party $P_h$ eventually A-Casts $C_h^{(t+1)}$, every honest party $P_i$ will eventually receive at least $n-t$ A-Casts $C_j'$ satisfying $|C_j'|=t+1$ and $C_j'\subseteq C_i$. Thus, every honest party $P_i$ will eventually pass the threshold $|G_i|\ge n-t$ and A-Cast $G_i^{(n-t)}$.

Whenever an honest party $P_h$ adds $P_j$ to $G_h$, any other honest party $P_i$ will eventually add $P_j$ to $G_i$. This is because $P_h$ received some $C_j'$ from $P_j$'s A-Cast, where $|C_j'|=t+1$ and $C_j'\subseteq C_h$, so eventually $P_i$ will receive the same $C_j'$, and following the analysis above it will eventually hold that $C_j'\subseteq C_i$ as well. Since every honest party $P_h$ eventually A-Casts $G_h^{(n-t)}$, every honest party $P_i$ will eventually receive at least $n-t$ A-Casts $G_j'$ satisfying $|G_j'|=n-t$ and $G_j'\subseteq G_i$. Thus, every honest party $P_i$ will eventually pass the threshold $|R_i|\ge n-t$ and A-Cast $R_i^{(n-t)}$.

Whenever an honest party $P_h$ adds $P_j$ to $R_h$, any other honest party $P_i$ will eventually add $P_j$ to $R_i$. This is because $P_h$ received some $G_j'$ from $P_j$'s A-Cast, where $|G_j'|=n-t$ and $G_j'\subseteq G_h$, so eventually $P_i$ will receive the same $G_j'$, and following the analysis above it will eventually hold that $G_j'\subseteq G_i$ as well. Since every honest party $P_h$ eventually A-Casts $R_h^{(n-t)}$, every honest party $P_i$ will eventually receive at least $n-t$ A-Casts $R_j'$ satisfying $|R_j'|=n-t$ and $R_j'\subseteq R_i$. Thus, every honest party $P_i$ will eventually pass the threshold $|S_i|\ge n-t$ and record $Z_i$.

For each honest party $P_i$, every $P_j\in Z_i$ will eventually be added to $G_h$ by any other honest party $P_h$. Thus, all honest parties will eventually start the reconstruction of all the secrets attached to every $P_j\in Z_i$, and $P_i$ will eventually be able to compute each tally $v_j$ that it needs to determine its output $z_i$. In actuality, this reasoning does not take into account the termination of parties; this is resolved via the termination procedure at the end of the protocol. An honest party $P_i$ who is able to compute $z_i$ on its own does not terminate immediately, but rather A-Casts $z_i$ and continues to participate in the protocol (in particular, in the reconstruction of secrets). Eventually $P_i$ will receive at least $n-t$ A-Casts $z_j'$, at which point it can safely terminate with the knowledge than any other honest party $P_h$ will eventually receive those same A-Casts (even if $P_h$ is unable to compute $z_h$). By taking the mode of these $n-t$ values, the majority of which must have been suggested by honest parties, it is guaranteed that all honest parties will output the same value in the case of a successful coin toss.

In summary, the protocol $\paocc$ ensures termination by guaranteeing that each honest party eventually satisfies the required thresholds and proceeds to the necessary phases of $\favss$. This guarantees eventual output generation without any deadlocks, preserving the perfect indistinguishability of the real and ideal worlds.

We now claim that the overall round complexity of $\paocc$ is constant. The protocol consists of a constant number of steps, each involving constant-round operations. Note that by ensuring that the underlying building blocks $\facast$ and $\favss$ have constant-round realizations, we are able to preserve the overall constant round complexity of protocol $\paocc$ even after replacing $\facast$ and $\favss$ with actual protocols realizing them.

We can conclude that the executions of $\paocc^{V}$ with $\adv$ and $\faocc^{V,p}$ with $\sim$ are perfectly indistinguishable to any environment $\env$. This ensures that $\paocc^{V}$ UC-realizes $\faocc^{V,p}$ with perfect security.
\end{proof}
\section{Concurrent A-BA in Expected-Constant Rounds}
\label{sec:concurrent}

Concurrent A-BA plays a significant role in existing asynchronous MPC protocols, including those described in \cite{ben-or1993asynchronous,ben-or1994asynchronous,hirt2005cryptographic,beerliova-trubiniova2007simple,hirt2008asynchronous,cohen2016asynchronous,coretti2016constant,BZL20,LLMMT20}. Its primary purpose is to address the lack of coordination among the parties, which is amplified in the asynchronous setting. Specifically, it enables parties to reach agreement on a common subset of input providers (the ``core set'' \cite{ben-or1993asynchronous}). This capability is essential in various applications within asynchronous networks, where the different delays experienced by each party can lead to discrepancies in their respective views.

In this section, we dive into the important problem of achieving concurrent A-BA in an expected-constant number of rounds. As discussed in the Introduction, Ben-Or and El-Yaniv~\cite{ben-or2003resilient} highlighted the potential issue of running multiple executions of a probabilistic-termination protocol in parallel, which could lead to an increase in the expected number of rounds required for \emph{all} executions to terminate. The concurrent A-BA protocol proposed in~\cite{ben-or2003resilient} relies on A-OLE and multi-valued A-BA, which can be instantiated using our A-OCC protocol from Section~\ref{sec:occ} and the extended A-BA protocol from~\cite{canetti1993fast,mostefaoui2017signature}, respectively. However, during our analysis, we discovered certain issues in their analysis that cast doubt on the expected-constant round complexity of one of their main building blocks and, consequently, their concurrent A-BA protocol. For a more comprehensive presentation of these issues, see Appendix~\ref{app:be}. It is unclear how to address these issues without modifying the protocol itself.

To rectify these concerns, we modify the underlying message distribution mechanism and incorporate an additional layer of message validation. These changes not only resolve the identified issues, but also significantly simplify the protocol. It is worth emphasizing that our revised concurrent A-BA protocol achieves a level of simplicity that is comparable to the synchronous version proposed in~\cite{ben-or2003resilient}. This accomplishment is significant because when designing an asynchronous counterpart to a synchronous protocol, achieving a level of simplicity on par with the synchronous version is often considered the ideal outcome. In the following, we present an ideal functionality for concurrent A-BA, describe our protocol, an ideal functionality for concurrent A-BA, and its required building blocks, and provide a security proof.

\subsection{Concurrent A-BA Ideal Functionality}
\label{sec:concurrent_functionality}

Concurrent A-BA, as the name suggests, refers to a primitive that enables parties to solve $N$ instances of A-BA concurrently. We are primarily interested in the case $N=n$, corresponding to the emulation of $n$ ideal A-BA primitives, commonly used in asynchronous MPC protocols to form the core set and overcome low message dispersion. However, in our study, we consider a more general version that allows for a broader range of values for $N$. In this setting, each party $P_i$ initiates the concurrent A-BA by providing $N$ values, namely $v_{i,1}, \ldots, v_{i,N}$. Subsequently, all parties receive the same set of $N$ output values, denoted as $y_1, \ldots, y_N$. Each individual output value $y_j$ is computed based on the input values $v_{1,j}, \ldots, v_{n,j}$, following the prescribed procedure outlined in the standard A-BA primitive. Specifically, if $n-2t$ input values are identical, that common value (which we require to be unique) is selected as the output; otherwise, the output is determined by the adversary. We capture the task of concurrent A-BA using the ideal functionality $\fcaba$, shown in Figure~\ref{fig:func-caba}. Since we are able to achieve non-intrusion validity (i.e., for each instance, the corresponding output must be either $\bot$ or the corresponding input of an honest party), we present an \emph{intrusion-tolerant} concurrent A-BA functionality.

\begin{figure}[ht!]
\centering
\funcbox{$\fcaba^{V,N}$}{
    The functionality is parameterized by a set $V$ of values and a positive integer $N$ representing the number of concurrent instances, and proceeds as follows. At the first activation, verify that $\sid=(\P,\sid')$, where $\P$ is a player set of size $n$. For each $P_i\in\P$, initialize $v_{i,1},\ldots,v_{i,N}$ to a default value $\bot$, $\participated_i\assign 0$, and delay values $D^{\inputlbl}_i=D^{\outputlbl}_i\assign 1$. Also initialize $y_1,\ldots,y_N$ to $\bot'$ and $a_1,\ldots,a_N$ to $\bot$, and $t\assign \lceil \frac{n}{3} \rceil - 1$.
    \begin{tiret}
        \item Upon receiving $(\delaycmd,\sid,P_i,\type,D)$ from the adversary for $P_i\in\P$, $\type\in\{\inputlbl,\outputlbl\}$, and $D\in\ZZ$ represented in unary notation, update $D^{\type}_i\assign \max(1,D^{\type}_i+D)$ and send $(\delaysetcmd,\sid)$ to the adversary.
        \item Upon receiving $(\inputcmd,\sid,(v_1,\ldots,v_N))$ from $P_i\in \P$ (or the adversary on behalf of corrupted $P_i$), run the Input Submission Procedure and send $(\leakagecmd,\sid,P_i,(v_1,\ldots,v_N))$ to the adversary.
        \item Upon receiving $(\replacecmd, \sid,(v_1,\ldots,v_N))$ from the adversary, record $(a_1,\ldots,a_N)\assign (v_1,\ldots,v_N)$.
        \item Upon receiving $(\fetchcmd, \sid)$ from $P_i\in\P$ (or the adversary on behalf of corrupted $P_i$), run the Input Submission and Output Release Procedures, and send $(\fetchedcmd,\sid,P_i)$ to the adversary and any messages set by the Output Release Procedure to $P_i$.
    \end{tiret}
    \smallskip
    {\bf Input Submission Procedure:}
    If $\participated_i=0$ and $(\inputcmd,\sid,(v_1,\ldots,v_N))$ was previously received from $P_i$, then do:
    \begin{enumerate}
        \item Update $D^{\inputlbl}_i\assign D^{\inputlbl}_i-1$.
        \item If $D^{\inputlbl}_i=0$, then set $\participated_i\assign 1$ and record $(v_{i,1},\ldots,v_{i,N})\assign (v_1,\ldots,v_N)$.
    \end{enumerate}
    {\bf Output Release Procedure:}
    If $\sum_{j=1}^{n}\participated_j\geq n-t$ then do:
    \begin{enumerate}
        \item Update $D^{\outputlbl}_i\assign D^{\outputlbl}_i-1$.
        \item If $D^{\outputlbl}_i=0$, then do the following. For each $j\in [N]$ do: if $y_j= \bot'$, then set $y_j\assign x$ if there exists a unique value $x\in V$ such that $x=v_{k,j}$ for at least $n-2t$ inputs $v_{k,j}$; else if $a_j=\bot$ or $a_j\in V$ was provided by at least one honest party, then set $y_j\assign a_j$; else set $y_j\assign \bot$. Additionally, set $(\outputcmd,\sid,(y_1,\ldots,y_N))$ to be sent to $P_i$.
    \end{enumerate}
    \smallskip
}
\caption{The (intrusion-tolerant) concurrent A-BA functionality.}
\label{fig:func-caba}
\end{figure}

\subsection{The Simplified Concurrent A-BA Protocol}
\label{sec:concurrent_protocol}

\paragraph{Building blocks.}

Our concurrent A-BA protocol relies on A-Cast as the fundamental communication primitive due to its enhanced guarantees compared to basic message distribution mechanisms. The ideal functionality $\facast$, which models the A-Cast primitive, was described in Section~\ref{sec:functionalities}.

As another crucial building block, our protocol incorporates asynchronous oblivious leader election (A-OLE) as a coordination mechanism among the parties. A-OLE enables parties to randomly elect a leader from among themselves. The term ``oblivious'' indicates that parties are unaware of whether or not agreement on a random leader has been achieved. In our concurrent A-BA protocol, similar to the approach described in~\cite{ben-or2003resilient}, there comes a point where all parties suggest outputs, and A-OLE assists them in reaching agreement on the output by adopting the suggestion of the elected leader. To capture the task of A-OLE, we parameterize the A-OCC functionality $\faocc$, given in Section~\ref{sec:OCC_functionality}, by a domain with size equal to the number of parties. This yields an ideal functionality for A-OLE, denoted as $\faole$, which is defined as $\faocc^{[n],p}$ for appropriate $p\in\Theta(1)$. Recall from Theorem~\ref{thm:occ} that $\faole$ can be realized using protocol $\paocc^{[n]}$.

Our concurrent A-BA protocol leverages both binary and multi-valued A-BA. Binary A-BA is employed to achieve agreement on a critical decision within the protocol, namely whether to continue a particular iteration. On the other hand, the use of multi-valued A-BA addresses the inherent obliviousness of the A-OLE primitive by providing a means for parties to reach an agreement on the output, and it doubly serves as an agreement on termination. The ideal functionality for A-BA was described in Section~\ref{sec:functionalities}.

As in~\cite{ben-or2003resilient}, another essential component in our concurrent A-BA protocol is {\em truncated} executions of an A-BA protocol limited to a predefined number of iterations, consequently implying a fixed number of rounds. In the spirit of \cite{cohen2019probabilistic}, we model those executions with the ideal functionality $\ftaba$, defined in Figure~\ref{fig:func-taba}. $\ftaba$ is parameterized with $V$, $p$, and $\itr$, where $V$ denotes the domain, $p$ represents the termination probability in each iteration, and $\itr$ indicates the maximum number of iterations in the execution.

This ideal functionality also encapsulates a ``1-shift'' property for termination, meaning that all honest parties produce the output within two consecutive iterations. We note that the adversary has the discretion to determine which parties discover the output first. Unlike a traditional A-BA functionality that outputs a single value, $\ftaba$ produces a vector of values that includes the output after each iteration of the execution. Modeling the 1-shift property and providing outputs for all iterations is crucial since our concurrent A-BA protocol relies on those properties of truncated executions of A-BA.

\begin{figure}[htbp!]
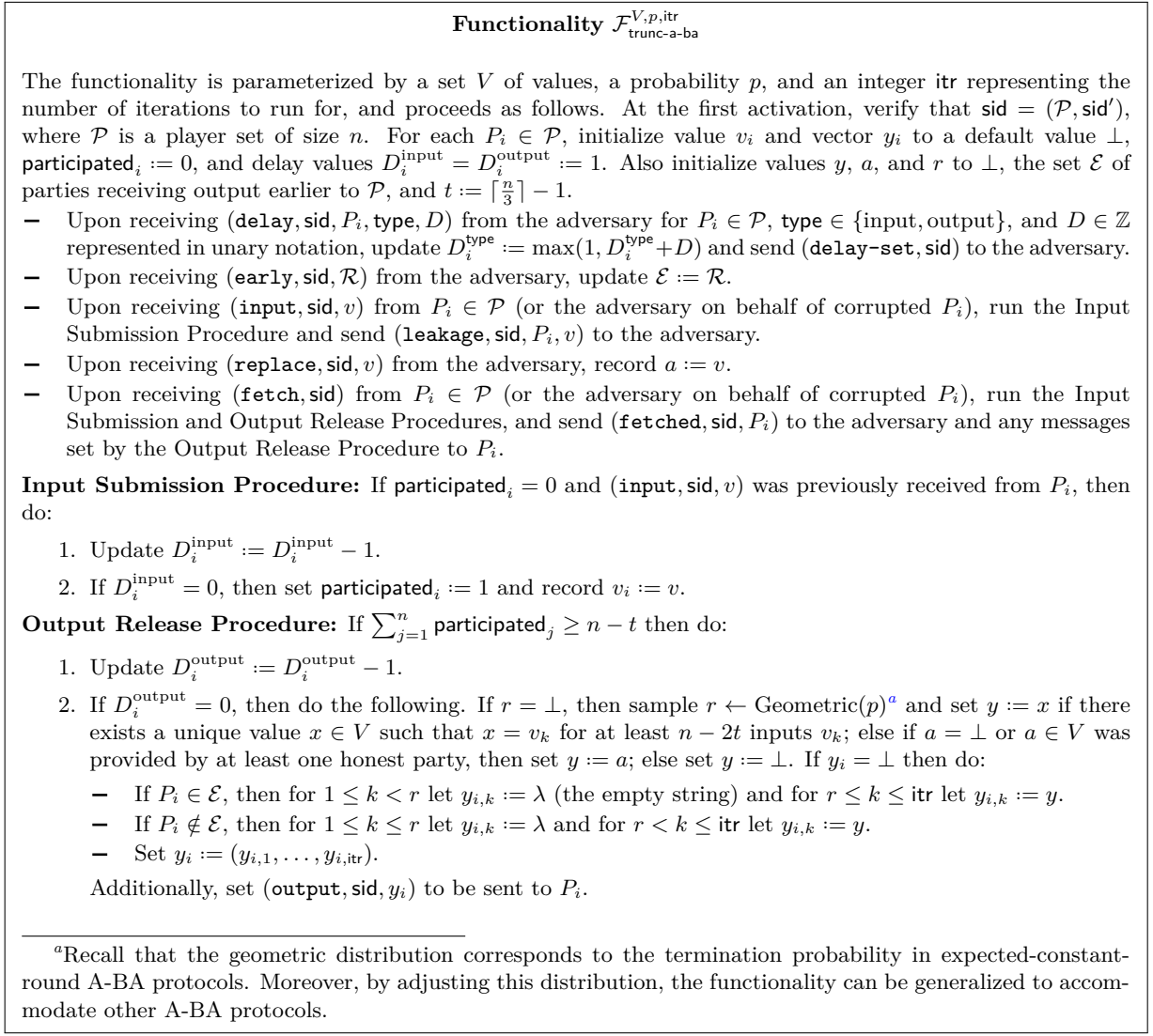

\centering
\funcbox{$\ftaba^{V,p,\itr}$}{
    The functionality is parameterized by a set $V$ of values, a probability $p$, and an integer $\itr$ representing the number of iterations to run for, and proceeds as follows. At the first activation, verify that $\sid=(\P,\sid')$, where $\P$ is a player set of size $n$. For each $P_i\in\P$, initialize value $v_i$ and vector $y_i$ to a default value $\bot$, $\participated_i\assign 0$, and delay values $D^{\inputlbl}_i=D^{\outputlbl}_i\assign 1$. Also initialize values $y$, $a$, and $r$ to $\bot$, the set $\calE$ of parties receiving output earlier to $\P$, and $t\assign \lceil \frac{n}{3} \rceil - 1$.
    \begin{tiret}
        \item Upon receiving $(\delaycmd,\sid,P_i,\type,D)$ from the adversary for $P_i\in\P$, $\type\in\{\inputlbl,\outputlbl\}$, and $D\in\ZZ$ represented in unary notation, update $D^{\type}_i\assign\max(1,D^{\type}_i+D)$ and send $(\delaysetcmd,\sid)$ to the adversary.
        \item Upon receiving $(\earlycmd,\sid,\mathcal{R})$ from the adversary, update $\calE \assign \mathcal{R}$.
        \item Upon receiving $(\inputcmd,\sid,v)$ from $P_i\in \P$ (or the adversary on behalf of corrupted $P_i$), run the Input Submission Procedure and send $(\leakagecmd,\sid,P_i,v)$ to the adversary.
        \item Upon receiving $(\replacecmd, \sid,v)$ from the adversary, record $a\assign v$.
        \item Upon receiving $(\fetchcmd, \sid)$ from $P_i\in\P$ (or the adversary on behalf of corrupted $P_i$), run the Input Submission and Output Release Procedures, and send $(\fetchedcmd,\sid,P_i)$ to the adversary and any messages set by the Output Release Procedure to $P_i$.
    \end{tiret}
    \smallskip
    {\bf Input Submission Procedure:}
    If $\participated_i=0$ and $(\inputcmd,\sid,v)$ was previously received from $P_i$, then do:
    \begin{enumerate}
        \item Update $D^{\inputlbl}_i\assign D^{\inputlbl}_i-1$.
        \item If $D^{\inputlbl}_i=0$, then set $\participated_i\assign 1$ and record $v_i\assign v$.
    \end{enumerate}
    {\bf Output Release Procedure:}
    If $\sum_{j=1}^{n}\participated_j\geq n-t$ then do:
    \begin{enumerate}
        \item Update $D^{\outputlbl}_i\assign D^{\outputlbl}_i-1$.
        \item If $D^{\outputlbl}_i=0$, then do the following. If $r=\bot$, then sample $r\leftarrow \mathrm{Geometric}(p)$\footnote{Recall that the geometric distribution corresponds to the termination probability in expected-constant-round A-BA protocols. Moreover, by adjusting this distribution, the functionality can be generalized to accommodate other A-BA protocols.} and set $y\assign x$ if there exists a unique value $x\in V$ such that $x=v_k$ for at least $n-2t$ inputs $v_k$; else if $a=\bot$ or $a\in V$ was provided by at least one honest party, then set $y\assign a$; else set $y\assign \bot$. If $y_i=\bot$ then do:
        \begin{tiret}
            \item If $P_i\in\calE$, then for $1\le k<r$ let $y_{i,k}\assign \lambda$ (the empty string) and for $r\leq k\leq \itr$ let $y_{i,k}\assign y$.
            \item If $P_i\notin\calE$, then for $1\le k\leq r$ let $y_{i,k}\assign \lambda$ and for $r< k\leq \itr$ let $y_{i,k}\assign y$.
            \item Set $y_i\assign (y_{i,1},\ldots,y_{i,\itr})$.
        \end{tiret}
        Additionally, set $(\outputcmd,\sid,y_i)$ to be sent to $P_i$.
    \end{enumerate}
    \smallskip
}
\caption{The (intrusion-tolerant) truncated A-BA functionality with the 1-shift property.}
\label{fig:func-taba}
\end{figure}

It is worth mentioning that $\ftaba^{V,p,\itr}$ can be implemented by executing any intrusion-tolerant A-BA protocol with the 1-shift property and a termination probability of $p$ in each iteration, precisely for $\itr$ iterations, and concatenating the output of all iterations to get the final output (with $\lambda$ representing iterations without output). Canetti and Rabin's binary A-BA protocol~\cite{canetti1993fast} possesses the desired properties of intrusion tolerance and terminating with a constant probability in each iteration. However, using Bracha's technique for termination~\cite{bracha1987asynchronous} in Canetti and Rabin's binary A-BA protocol does not admit the 1-shift property. The main reason is that parties may take the shortcut and use Bracha termination messages to generate the output in their very early iterations. Fortunately, Bracha's termination procedure is unnecessary in the truncated execution of Canetti and Rabin's binary A-BA protocol. This is primarily due to the fact that all parties will naturally terminate after a fixed number of iterations. With this adjustment, Canetti and Rabin's binary A-BA protocol also successfully attains the 1-shift property. Moreover, Most\'{e}faoui and Raynal's multi-valued A-BA protocol~\cite{mostefaoui2017signature} offers intrusion tolerance, the 1-shift property, and terminating with a constant probability in each iteration if the underlying binary A-BA used in their construction also exhibits these characteristics. Thus, we can formulate the following proposition about realizing $\ftaba^{V,p,\itr}$.

\begin{proposition}
\label{thm:taba}
    For some constant probability $p$, any domain $V$, and any integer $\itr$, $\ftaba^{V,p,\itr}$ can be UC-realized with statistical security in the $\fasmt$-hybrid model, in constant rounds and in the presence of an adaptive and malicious $t$-adversary, provided $t<n/3$.
\end{proposition}

The above statement guarantees statistical security since the A-BA protocols we consider rely on the A-VSS primitive, which can only be realized with statistical security when $t < n/3$. In fact, by working in the $\favss$-hybrid model, we can achieve perfectly secure $\ftaba$. We remark that the intrusion tolerance of $\ftaba$ is not a requirement in our concurrent A-BA protocol; however, employing a non-intrusion-tolerant version of truncated A-BA will naturally lead to a non-intrusion-tolerant concurrent A-BA protocol.

\paragraph{The new concurrent A-BA protocol.}

Our protocol builds on the core ideas presented in~\cite{ben-or2003resilient}. In addition to instantiating the missing OLE building block using our OCC protocol from Section~\ref{sec:occ}, we address the issue in the analysis by redesigning the message distribution phase. Our revised message distribution mechanism not only resolves the issue in the proof but also provides stronger guarantees, which in turn simplifies the final protocol design. In fact, apart from the message distribution phase, the overall structure of our protocol closely resembles the synchronous version of the protocol described in~\cite{ben-or2003resilient}, as illustrated in Figure~\ref{fig:graphic}.

\begin{figure}[ht!]
\centering
\includegraphics{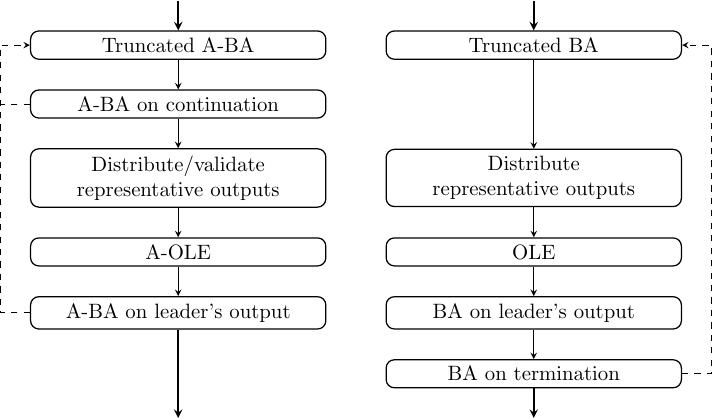}
\caption{Side-by-side comparison of our simplified concurrent A-BA protocol (left) and Ben-Or and El-Yaniv's concurrent (synchronous) BA protocol (right).}
\label{fig:graphic}
\end{figure}

Before diving into the high-level description of our protocol, we highlight the choices we made in the message distribution mechanism and discuss some alternative approaches that fail to meet our requirements. In the eventual-delivery model, at least $n-t$ parties will receive each other's messages if they wait for a sufficient duration, as messages from honest parties will eventually be delivered to one another, but determining the exact waiting time required is not straightforward. One possible approach is to instruct parties to A-Cast the identities of the parties from which they have received messages. After constructing a graph with parties as vertices and adding edges between parties that have reported message receipts from each other, we can look for a clique of size $n-t$; however, finding a maximum clique is known to be NP-complete and also difficult to approximate~\cite{feige1991approximating}.

To overcome this challenge, one possible approach is to investigate alternative structures that offer weaker guarantees regarding message dispersion but can be efficiently identified~\cite{ben-or1993asynchronous,ben-or1994asynchronous,canetti1996studies,patra2011error}. However, it is important to note that this approach does not guarantee the precise message dispersion required for our specific application. In our analysis, it is crucial that \emph{all} honest parties receive messages from a linear fraction of other parties. Even finding a clique of size $n-t$ does not guarantee this level of message dispersion, rendering this approach unsuitable for our protocol. Thus, we adopt another approach that has been used in prior works~\cite{feldman1989asynchronous,canetti1993fast}, and extend it by incorporating a precondition check before the process and introducing a validation layer during the execution. These additions enhance the guarantees, making them more suitable and effective for our specific purpose. We proceed to explain our protocol.

Similar to the (synchronous) protocol described in~\cite{ben-or2003resilient}, in our concurrent A-BA protocol, each party initiates for every A-BA instance a batch of $m$ executions of the A-BA protocol (over the same inputs) for a fixed number of iterations, denoted as $\itr$. If the A-BA protocol has a termination probability of at least a constant value $p$ in each iteration, which is the case for most existing A-BA protocols, suitable values of $m$ and $\itr$ can be determined so that each party obtains at least one output value for each batch. Each party then selects an output for each instance and forms a suggestion for the final output. The next paragraph, which explains the mechanism for distributing suggestions among parties, is our main modification to the protocol. Then, for the remaining part, we can use a similar structure as the synchronous version of the protocol.

Firstly, parties initiate a binary A-BA protocol to determine a specific condition that allows for choosing and validating suggested outputs later in the execution. Based on the outcome of this binary A-BA, parties decide whether to continue or start over. In the case of continuation, parties perform A-Cast operations to distribute their suggestions and wait to receive suggestions from other parties. Each party only accepts an A-Cast message containing a suggested output if the value is consistent with the outputs obtained from its own truncated A-BA executions. This validation step is crucial as it ensures that even corrupted parties provide acceptable (correct) suggestions. The validation is based on the 1-shift property of A-BA that ensures all honest parties terminate within two consecutive iterations. Parties wait to accept A-Casts of suggested outputs from at least $n-t$ parties and then A-Cast the set of all these $n-t$ suggestions along with the identities of the corresponding senders.\footnote{In fact, it suffices to A-Cast the identities alone. We primarily use this observation to simplify the description of the protocol, though it also reduces the communication complexity.} They continue accepting A-Casts of suggestions and sets until they receive at least $n-t$ sets that are fully contained within their accepted suggestions. At this point, a sufficient number of messages have been exchanged, and using a counting argument similar to the one in~\cite{feldman1989asynchronous,canetti1993fast}, it can be deduced that at least $n-2t>n/3$ parties have their suggested outputs received by all honest parties.

After the message distribution phase described above, our protocol proceeds similarly to the synchronous protocol presented in~\cite{ben-or2003resilient}. Specifically, parties execute OLE to elect a random leader to adopt its suggested output. Subsequently, a multi-valued A-BA is performed on the adopted output to address the oblivious nature of OLE. Depending on whether agreement on a non-default output is reached, this results in the termination or restarting of the protocol. The intrusion-tolerance property of multi-valued A-BA is crucial to make sure that the common output is not provided by a corrupted party alone. It is worth noting that the favorable scenario occurs when the leader is among the $n-2t$ parties whose suggested outputs have been accepted by all honest parties. If the leader is elected randomly, this event happens with a probability of $1/3$. In this case, all parties adopt the same output, resulting in termination.

It is worth noting that the set of $n-2t$ parties whose suggested outputs have been accepted by all honest parties may only contain a single honest party. This single honest party can only be elected with a probability of $O(1/n)$ by the OLE. However, due to the validation step in the message-distribution phase, there is no longer a need to ensure that an honest leader is elected, as the suggested values from corrupted parties are considered valid outputs. Refer to Figure~\ref{fig:prot-caba} for a formal description of protocol $\pcaba$.

\begin{figure}[ht!]
\centering
\protbox{$\pcaba^{V,N,m,p,\itr}$}{
    The protocol is parameterized by a set $V$ of possible inputs (for each instance), positive integers $N$ and $m$ representing respectively the number of instances and the number of batched executions to run per instance, a termination probability $p$, and the number $\itr$ of iterations after which to truncate the executions. At the first activation, verify that $\sid=(\mathcal{P},\sid')$ for a player set $\mathcal{P}$ of size $n$. Let $t\assign\lceil\frac{n}{3}\rceil-1$. Party $P_i\in\mathcal{P}$ proceeds as follows. Initialize sets $Q_{i,1}$, $Q_{i,2}$, and $W_i$ to $\emptyset$, flag $\finished$ to 0, and counter $\varphi$ to $1$. Also initialize $S_{i,j}^r\assign \emptyset$, $b_{j,k}\assign0$, $\Vec{C}_h'\assign\bot$, and $y_{i,j}\assign\bot$ for all $h\in[n]$, $j\in[N]$, $k\in[m]$, and $\itr\le r\le\itr+3$.
    \begin{tiret}
        \item Upon receiving input $(\inputcmd,\sid,(v_{i,1},\ldots,v_{i,N}))$ from the environment where $v_{i,1},\ldots,v_{i,N}\in V$, do the following. For each $j\in[N]$ and $k\in[m]$, send $(\inputcmd,\sid_{\varphi}^{j,k},v_{i,j})$ to an instance of $\ftaba^{V,p,\itr+3}$ with SID $\sid_{\varphi}^{j,k}\assign (\sid,\varphi,j,k)$.
        \item Upon receiving input $(\fetchcmd,\sid)$ from the environment, if $\finished=1$ then output $(\outputcmd,\sid,(y_{i,1},\ldots,y_{i,N}))$ to the environment; else do:
        \begin{enumerate}
            \item For each $j\in[N]$ and $k\in[m]$, fetch the output from the instance of $\ftaba^{V,p,\itr+3}$ with SID $\sid_{\varphi}^{j,k}$. Upon receiving back $(\outputcmd,\sid_{\varphi}^{j,k},(z_{i,1},\ldots,z_{i,\itr+3}))$, update $S_{i,j}^r\assign S_{i,j}^r\cup\{z_{i,r}\}\setminus\{\lambda\}$ for $\itr\le r\le\itr+3$, and also set $b_{j,k}\assign 1$.
            \item Wait until $b_{j,k}=1$ for all $j\in[N]$ and $k\in[m]$. Then, let $\continue_i\assign 1$ if for all $j\in[N]$ it holds that $S_{i,j}^\itr\not=\emptyset$, and $\continue_i\assign 0$ otherwise; send $(\inputcmd,\sid_{\varphi}^{\contlbl},\continue_i)$ to an instance of $\faba^{\{0,1\}}$ with SID $\sid_{\varphi}^{\contlbl}\assign (\sid,\varphi,\texttt{cont})$.
            \item Fetch the output from the instance of $\faba^{\{0,1\}}$ with SID $\sid_{\varphi}^{\contlbl}$. Upon receiving back $(\outputcmd,\sid_{\varphi}^{\contlbl},\continue)$, do the following. If $\continue=1$, wait until for all $j\in[N]$ it holds that $S_{i,j}^{\itr+1}\not=\emptyset$; then for each $j\in[N]$ choose $c_{i,j}\in S_{i,j}^{\itr+1}$, and send $(\sendcmd,\sid_{i,\varphi}^{\vectorlbl},(c_{i,1},\ldots,c_{i,N}))$ to an instance of $\facast$ with SID $\sid_{i,\varphi}^{\vectorlbl}\assign (\mathcal{P},P_i,\sid',\varphi,\texttt{vector})$. Otherwise, start over (i.e., increment $\varphi$, re-initialize all other local variables, and run truncated A-BA $mN$ times again).
            \item For each $h\in[n]$, fetch the output from the instance of $\facast$ with SID $\sid_{h,\varphi}^{\vectorlbl}$. Upon receiving back $(\outputcmd,\sid_{h,\varphi}^{\vectorlbl},(c'_{h,1},\ldots,c'_{h,N}))$, record $\Vec{C}_h'\assign(c'_{h,1},\ldots,c'_{h,N})$ and for each $b\in\{1,2\}$ wait until it holds that $c'_{h,j}\in S_{i,j}^{\itr+1+b}$ for all $j\in[N]$; then update $Q_{i,b}\assign Q_{i,b}\cup\{P_h\}$.
            \item Wait until $|Q_{i,1}|\ge n-t$. Then, send $(\sendcmd,\sid_{i,\varphi}^{\setlbl},Q_{i,1}^{(n-t)})$ to an instance of $\facast$ with SID $\sid_{i,\varphi}^{\setlbl}\assign (\mathcal{P},P_i,\sid',\varphi,\setcmd)$.
            \item For each $h\in[n]$, fetch the output from the instance of $\facast$ with SID $\sid_{h,\varphi}^{\setlbl}$. Upon receiving back $(\outputcmd,\sid_{h,\varphi}^{\setlbl},Q_{h,1}')$ where $|Q_{h,1}'|=n-t$, wait until $Q_{h,1}'\subseteq Q_{i,2}$; then update $W_i\assign W_i\cup \{h\}$.
            \item Wait until $|W_i|\geq n-t$. Then, send $(\inputcmd,\sid_{\varphi}^{\electlbl})$ to an instance of $\faole$ with SID $\sid_{\varphi}^{\electlbl}\assign (\sid,\varphi,\texttt{elect})$.
            \item Fetch the output from the instance of $\faole$ with SID $\sid_{\varphi}^{\electlbl}$. Upon receiving back $(\outputcmd,\sid_{\varphi}^{\electlbl},l_i)$, wait until $|W_i|\ge n-t$; then send $(\inputcmd,\sid_{\varphi}^{\vectorlbl},\Vec{C}'_{l_i})$ to an instance of $\faba^{V^N\cup\{\bot\}}$ with SID $\sid_{\varphi}^{\vectorlbl}\assign (\sid,\varphi,\texttt{vector})$ if it holds that $P_{l_i}\in Q_{i,2}$, and send $(\inputcmd,\sid_{\varphi}^{\vectorlbl},\bot)$ otherwise.
            \item Fetch the output from the instance of $\faba^{V^N\cup\{\bot\}}$ with SID $\sid_{\varphi}^{\vectorlbl}$. Upon receiving back $(\outputcmd,\sid_{\varphi}^{\vectorlbl},\Vec{Y})$, if $\Vec{Y}\not=\bot$ then record $(y_{i,1},\ldots,y_{i,N})\assign \Vec{Y}$ and set $\finished\assign 1$; else start over.
        \end{enumerate}
    \end{tiret}
    \smallskip
}
\caption{The concurrent A-BA protocol.}
\label{fig:prot-caba}
\end{figure}

\subsection{Security Proof}
\label{sec:concurrent_proof}

We now consider the security of our concurrent A-BA protocol. Before stating the theorem, it is worth noting that the specific parameters of the hybrid model, which combine the different ideal functionalities, are not explicitly specified in the theorem statement. However, they can be determined from the protocol's parameters and are integral to the overall security guarantees of the protocol. Now, let us state the theorem formally:

\begin{theorem}
\label{thm:caba}
    For any domain $V$, integer $N$, constant $0<p<1$, and constant integer $\itr>1$, setting $m\assign\log_{\frac{1}{1-p}} N$, the protocol $\pcaba^{V,N,m,p,\itr}$ UC-realizes $\fcaba^{V,N}$ with statistical security in the $(\facast,\faba,\ftaba,\faole)$-hybrid model, in expected-constant rounds and in the presence of an adaptive and malicious $t$-adversary, provided $t<n/3$.
\end{theorem}

\begin{proof}
Let $\adv$ be an adversary in the real world. We construct a simulator $\sim$ in the ideal world, such that no environment $\env$ can distinguish whether it is interacting with $\pcaba^{V,N,m,p,\itr}$ and $\adv$, or with $\fcaba^{V,N}$ and $\sim$. Here, $V$ and $N$ are given parameters of the problem, and $p$ is the parameter for the $\ftaba$ hybrid, which is determined based on the specific protocol used to realize it. The parameters $m$ and $\itr$ need to be appropriately adjusted based on the other parameters to ensure an expected-constant round complexity for $\pcaba^{V,N,m,p,\itr}$. As discussed below, setting $m\assign\log_{\frac{1}{1-p}} N$ for any constant integer $\itr > 1$ suffices.

The simulator internally runs a copy of $\adv$, and plays the roles of all the hybrids and the parties in a simulated execution of the protocol. All inputs from $\env$ are forwarded to $\adv$, and all outputs from $\adv$ are forwarded to $\env$. Moreover, whenever $\adv$ corrupts a party in the simulation, $\sim$ corrupts the same party in the ideal world by interacting with $\fcaba^{V,N}$, and if the corruption was direct (i.e., not via either of the aiding functionalities), then $\sim$ sends $\adv$ the party's state and thereafter follows $\adv$'s instructions for that party. Moreover, $\sim$ manages delays in the ideal functionality based on the delays that $\adv$ sets on the hybrids.

The simulated execution begins when $\sim$ receives messages of the form $(\leakagecmd,\sid,\cdot,\cdot)$ from $\fcaba^{V,N}$. These messages contain the values that party $P_i$ receives as input from the environment. $\sim$ can derive these input values from the leakage message $(\leakagecmd,\sid,P_i,(v_1,\ldots,v_N))$ by setting $v_{i,j}=v_j$ for all $j\in [N]$. The set $\{v_{1,j},v_{2,j},\ldots,v_{n,j}\}$ represents the inputs for the \jth instance of the (concurrent) A-BA that the parties are supposed to solve.

Upon receiving these leakage messages, $\sim$ can initiate sessions of $\ftaba$ with the actual inputs. For each instance $j\in [N]$, $\sim$ invokes $m$ sessions of $\ftaba$ with the inputs $v_{1,j},v_{2,j},\ldots,v_{n,j}$. After this point, $\sim$ continues the simulation by honestly playing the roles of other ideal functionalities and honest parties involved in the protocol. Since there are no private states that $\sim$ needs to guess or fake in order to provide a consistent view, the rest of the simulation proceeds straightforwardly.

Now, this internal execution generates a transcript that is perfectly indistinguishable from a real-world execution. Therefore, it remains to show that the protocol execution indeed results in the same output as $\fcaba^{V,N}$. To demonstrate this, we first establish that with probability $1$ the execution will eventually terminate. Additionally, we show that if at least $n-2t$ parties have the same input value for any instance (and there is only one such value), then everyone will receive that value as the output; otherwise, the output will be either one of the inputs provided by some honest parties for that instance or a special value~$\bot$.

\paragraph{Termination.}
After initializing some variables, the parties collectively initiate $m$ sessions of $\ftaba$ for each of the $N$ instances. Since the maximum number of corrupted parties is $t$ and $\ftaba$ only requires the participation of $n-t$ parties, all sessions of $\ftaba$ will eventually terminate upon receiving a sufficient number of activations from the environment (i.e., every honest party $P_i$ will eventually set $b_{j,k}\assign 1$ for all $j\in[N]$ and $k\in[m]$).

If, for all $N$ instances, the integer $r$ sampled in any of the $m$ sessions of $\ftaba$ for that instance is at most $\itr-1$ (so that an honest party $P_i$ must eventually receive $z_{i,\itr}\not=\lambda$ from that session), then every honest party $P_i$ eventually sets $\continue_i\assign 1$. Here, $r$ is sampled from the geometric distribution $\mathrm{Geometric}(p)$ (that corresponds to expected-constant-round A-BA as in~\cite{feldman1989asynchronous,canetti1993fast}). We define $E$ as the event where for all instances at least one of the $m$ sessions of $\ftaba$ has $r \leq \itr-1$. The probability of event $E$ occurring can be calculated as follows, using similar logic as presented in~\cite{ben-or2003resilient}.

Define $A_{r_0}$ as the event where $\ftaba$ does not sample $r=r_0$ given $r\ge r_0$. Then
\[
\Pr[A_{r_0}]=1-p.
\]
Therefore, for the event $B_{r_0}$, in which at least one of the $m$ $\ftaba$'s for an instance samples $r=r_0$ given $r\ge r_0$, we have
\[
\Pr[B_{r_0}]=1-\Pr[A_{r_0}]^m=1-\left(1-p\right)^m.
\]
This implies that for the event $C_{r_0}$ in which for some of the $N$ instances, none of the $m$ $\ftaba$'s sample $r=r_0$ given $r\ge r_0$, the following holds:
\[
\Pr[C_{r_0}]=1-\Pr[B_{r_0}]^N=1-\left(1-\left(1-p\right)^m\right)^N
\]
Now we can calculate the probability of $E^\complement$, which represents the event where for some of the $N$ instances, none of the $m$ $\ftaba$'s sample $r \le \itr-1$:
\[
\Pr[E^\complement]=\prod_{r_0=1}^{\itr-1}\Pr[C_{r_0}]=\left( 1-\left(1-\left(1-p\right)^m\right)^N\right)^{\itr-1}
\]
Therefore, the complementary probability is
\[
\Pr[E]=1-\Pr[E^\complement]=1-\left( 1-\left( 1- \left( 1-p \right)^{m}\right)^{N}\right)^{\itr -1}.
\]
By setting $m\assign\log_{\frac{1}{1-p}} N$ it holds that
\[
\Pr[E]=1-\left( 1-\left( 1- \frac{1}{N}\right)^{N}\right)^{\itr -1}.
\]
Then considering $\left( 1- \frac{1}{N}\right)^{N}\approx\frac{1}{e}$ we have
\[
\Pr[E]\approx 1-\left( 1-\frac{1}{e}\right)^{\itr -1}.
\]

So, for a constant value of $\itr>1$, the probability $\Pr[E]$ is a non-zero constant. This implies that, with appropriate parameter choices, every honest party $P_i$ sets $\continue_i\assign 1$ with constant probability. In other words, the A-BA on continuation outputs $\continue=1$ with constant probability.

In any case, since every honest party $P_i$ eventually sets $\continue_i$ and $\faba$ only requires participation from $n-t$ parties, agreement on $\continue$ will eventually be reached. When $\continue=0$, all honest parties will eventually proceed to the next iteration.

On the other hand, when $\continue=1$, at least one honest party $P_g$ must have provided $\continue_g=1$, meaning that it observed $S_{g,j}^{\itr}\not=\emptyset$ for all $j\in[N]$. Note that whenever an honest party $P_g$ adds $z_{g,r}\not=\lambda$ to $S_{g,j}^r$ for some $j\in[N]$ and $\itr\le r\le \itr+2$, any honest party $P_i$ will eventually add $z_{i,r+1}=z_{g,r}$ to $S_{i,j}^{r+1}$, since $\ftaba$ models the 1-shift property. Thus, in case of continuation, every honest party $P_i$ will eventually have $S_{i,j}^{\itr+1}\not=\emptyset$ for all $j\in[N]$ and A-Cast $(c_{i,1},\ldots,c_{i,N})$.

Since every honest party $P_g$ eventually A-Casts a vector $(c_{g,1},\ldots,c_{g,N})$ satisfying $c_{g,j}\in S_{g,j}^{\itr+1}$ for all $j\in[N]$, following the analysis above every honest party $P_i$ will eventually receive at least $n-t$ A-Casts $\Vec{C}_h'=(c_{h,1}',\ldots,c_{h,N}')$ satisfying $c_{h,j}'\in S_{i,j}^{\itr+2}$ for all $j\in[N]$. Thus, every honest party $P_i$ will eventually pass the threshold $|Q_{i,1}|\ge n-t$ and A-Cast $Q_{i,1}^{(n-t)}$.

Whenever an honest party $P_g$ adds $P_h$ to $Q_{g,1}$, any other honest party $P_i$ will eventually add $P_h$ to $Q_{i,2}$. This is because $P_g$ received some $\Vec{C}_h'=(c_{h,1}',\ldots,c_{h,N}')$ from $P_h$'s A-Cast, where $c_{h,j}'\in S_{g,j}^{\itr+2}$ for all $j\in[N]$, so eventually $P_i$ will receive the same $\Vec{C}_h'$, and following the analysis above it will eventually hold that $c_{h,j}'\in S_{i,j}^{\itr+3}$ for all $j\in[N]$. Since every honest party $P_g$ eventually A-Casts $Q_{g,1}^{(n-t)}$, every honest party $P_i$ will eventually receive at least $n-t$ A-Casts $Q_{h,1}'$ satisfying $|Q_{h,1}'|=n-t$ and $Q_{h,1}'\subseteq Q_{i,2}$. Thus, every honest party $P_i$ will eventually pass the threshold $|W_i|\ge n-t$ and participate in $\faole$.

Before proceeding, it is important to note that by using a similar counting argument as in~\cite{feldman1989asynchronous,canetti1993fast}, we can guarantee that messages (suggestions) from at least $n-2t$ parties will eventually be received by all honest parties (by the time of participation in $\faole$). In more detail, consider the first honest party $P_i$ who reaches the stage where $s\assign|W_i|\ge n-t$. We can construct an $s \times n$ table $T$, where $T_{h,l}=1$ if and only if $P_h\in W_i$ and $P_l\in Q_{h,1}'$ (where $Q_{h,1}'$ was received from $P_h$'s A-Cast). Each row of the table contains (exactly) $n-t$ entries equal to one, indicating that there are a total of $s(n-t)$ one-entries in $T$. Let $q$ be the minimum number of columns in $T$ that contain at least $t+1$ one-entries. By considering the worst-case distribution of ones, we have the inequality
\[
qs+(n-q)t\ge s(n-t).
\]
Noting that it suffices to consider the case $s=n-t$ and rearranging terms, we have
\begin{align*}
    q& \ge \frac{(n-t)^2-nt}{n-2t} \\
    & = \frac{(n-t)(n-2t)-t^2}{n-2t} \\
    & = n-t-\frac{t^2}{n-2t} \\
    & \ge  n-2t. && (n-2t>t)
\end{align*}

Therefore, we can conclude that at least $q\ge n-2t$ columns in ($P_i$'s local view of) the table contain at least $t+1$ one-entries. Fix such a column $l$, and consider another honest party $P_g$ who reaches the stage $|W_g|\ge n-t$. Since $(n-t)+(t+1)>n$, the \iith{l} column in $P_g$'s view of the table must necessarily contain a one-entry. This means that $P_g$ received the A-Cast of $Q_{h,1}'$ from some party $P_h$, where $P_l\in Q_{h,1}'\subseteq Q_{g,2}$. In other words, a set of at least $n-2t$ parties will eventually be contained in the $Q_{g,2}$ set of every honest party $P_g$ (by the time that $|W_g|\ge n-t$).

Since all honest parties eventually participate in $\faole$, which only requires participation from $n-t$ parties, every honest party $P_i$ will eventually receive the index of a party $P_{l_i}$. With constant probability, there is agreement on a random leader $P_l$. Conditioned on agreeing on a random leader, with probability at least $\frac{n-2t}{n}>\frac{1}{3}$ the leader $P_l$ is among the parties whose vectors will eventually be delivered to (and validated by) all honest parties. We stress that it is crucial here that a linear number of such parties are fixed \emph{before} any honest party participates in $\faole$, and consequently before the identity of the leader is leaked to the adversary. Since every honest party $P_i$ waits until $|W_i|\ge n-t$, the condition $P_l\in Q_{i,2}$ will hold, and hence all honest parties will input (the same) $\Vec{C}_l'\not=\bot$ to the multi-valued A-BA on the leader's vector. In this case, every honest party $P_i$ will eventually receive $\Vec{Y}=\Vec{C}_l'$ from the A-BA and terminate.

In any case, every honest party $P_i$ eventually has $|W_i|\ge n-t$ and participates in the multi-valued A-BA, so agreement on some $\Vec{Y}$ will eventually be reached. If $\Vec{Y}=\bot$, all honest parties will eventually proceed to the next iteration; otherwise, they will all output $\Vec{Y}$.

Therefore, the execution of the protocol $\pcaba$ does not encounter any deadlocks. Furthermore, by choosing appropriate parameters, specifically $m\assign\log_{\frac{1}{1-p}} N$, the protocol can achieve termination with a constant probability in each iteration. This implies that $\pcaba$ will eventually terminate with probability $1$, and the expected number of iterations is constant. Since the number of rounds in each iteration is also constant, the overall round complexity of the protocol is constant in expectation.

\paragraph{Correctness.}

Now we demonstrate that once the protocol terminates, all honest parties output the same value as the ideal functionality $\fcaba^{V,N}$. This part of the proof requires the intrusion-tolerance property of A-BA, which guarantees that the output is either $\bot$ or a value provided by an honest party.

Since we assumed termination, there was agreement on $\Vec{Y}=(y_1,\ldots,y_N)$. Since $\faba$ models the intrusion-tolerance property, this implies that at least one honest party $P_i$ provided $\Vec{C}_{l_i}'=\Vec{Y}$, where $l_i$ was received from $\faole$ and $P_{l_i}\in Q_{i,2}$, meaning that party $P_{l_i}$ A-Casted a vector $\Vec{C}_{l_i}'=(y_1,\ldots,y_N)$ satisfying $y_j\in S_{i,j}^{\itr+3}$ for all $j\in[N]$. This implies that $y_j$ is a valid and non-empty (not $\lambda$) output from some session of $\ftaba$ for the \jth instance, for each $j\in[N]$. According to the description of $\ftaba$, if there is a unique input value that was provided by at least $n-2t$ parties for an instance, everyone will receive that value as the output for the instance; otherwise, the output will be either one of the inputs provided by some honest parties for that instance or a special value $\bot$. It is clear from inspection that $\fcaba^{V,N}$ generates its output in the same manner.

Therefore, with probability $1$, no environment can distinguish between the execution of $\pcaba^{V,N,m,p,\itr}$ with $\adv$ and $\fcaba^{V,N}$ with $\sim$, assuming expected-PPT machines. However, in order to apply UC composition guarantees, we need to introduce a stopping point (proportional to the security parameter) for our expected-PPT machines and make them strict PPT. In doing so, the distinguishing probability of the environment would become a non-zero value which can be arbitrarily small, depending on the duration for which we allow the machines to run. This ensures that the \emph{strict} PPT protocol $\pcaba^{V,N,m,p,\itr}$ UC-realizes $\fcaba^{V,N}$ with statistical security.
\end{proof}
\section*{Acknowledgements}
Our original motivation for this project was to provide a simulation-based treatment of concurrent A-BA protocols, such as Ben-Or and El-Yaniv's~\cite{ben-or2003resilient}, but the search for building blocks, in particular of an optimally resilient asynchronous OCC protocol became a bit of a ``detective story,'' as many references pointed to an unpublished manuscript by Feldman~\cite{feldman1989asynchronous}, which was nowhere to be found. We thank Michael Ben-Or for providing it to us, which corroborated its in-existence.

Ran Cohen's research is supported in part by NSF grant no.\ 2055568. Juan Garay's research is supported in part by NSF grants no.\ 2001082 and 2055694. Vassilis Zikas's research is supported in part by NSF grant no.\ 2055599 and by Sunday Group.
The authors were also supported by the Algorand Centres of Excellence programme managed by Algorand Foundation. Any opinions, findings, and conclusions or recommendations expressed in this material are those of the author(s) and do not necessarily reflect the views of Algorand Foundation.

{\small
\newcommand{\etalchar}[1]{$^{#1}$}

}

\appendix

\section{Attack on Ben-Or and El-Yaniv's Select Protocol}
\label{app:be}

In this section, we begin by presenting a static adversary that disproves an important claim regarding the expected-constant round complexity for one of the main subroutines in Ben-Or and El-Yaniv's concurrent A-BA protocol~\cite{ben-or2003resilient}. Consequently, the assured expected-constant round complexity of their concurrent A-BA protocol is no longer valid. Furthermore, we demonstrate that when considering adaptive adversaries, attempting to argue even for weaker claims becomes futile. The attacks we present here consider a simple case with only four parties, one of which can be corrupted (respecting the bound $t<n/3$). To provide essential context for our attacks, we present some of the protocols nearly verbatim from~\cite{ben-or2003resilient} in Figures~\ref{fig:prot-be-acast+},~\ref{fig:prot-be-spread}, and~\ref{fig:prot-be-select}.

\begin{figure}[htbp!]
\centering
\protbox{$\pacastp$}{
    {\bf Common inputs:} $n$, $t\le (n-1)/3$; $k$ (the identity of the transmitter)

    {\bf Local input for processor $p_i$:} none

    {\bf Additional local input for processor $p_k$:} $V$

    {\bf Local output for processor $p_i$:} $Value_i$, $Relay_i$

    \medskip

    Epoch 1: ($p_k$ only)
    \begin{itemize}
        \item A-Cast $V$.
    \end{itemize}
    Epoch 2: (Every processor $p_i$)
    \begin{itemize}
        \item Upon accepting a value $v$ from a transmitter $p$, A-Cast $[``commit",p,v]$.
    \end{itemize}
    Epoch 3: (Every processor $p_i$)
    \begin{itemize}
        \item Upon accepting $2t+1$ $[``commit",p',v']$ A-Casts from a set $S$ of processors (transmitters), $Value_i\leftarrow v'$, $Relay_i\leftarrow S$, and return $Value_i, Relay_i$.
    \end{itemize}
}
\caption{The A-Cast\textsuperscript{+} protocol from~\cite[Fig.\ 3]{ben-or2003resilient}.}
\label{fig:prot-be-acast+}
\end{figure}

\begin{figure}[htbp!]
\centering
\protbox{$\pspread$}{
    {\bf Common inputs:} $n$, $t\le (n-1)/3$

    {\bf Local input for processor $p_i$:} $V_i$

    {\bf Local output for processor $p_i$:} $\langle X_1,X_2,\ldots,X_n \rangle$

    \medskip

    Epoch 1:
    \begin{itemize}
        \item A-Cast\textsuperscript{+} $V_i$.
    \end{itemize}
    Epoch 2:
    \begin{itemize}
        \item Wait until $2t+1$ values $V_{i_1},\ldots,V_{i_{2t+1}}$ accepted (from $p_{i_1},\ldots,p_{i_{2t+1}}$).
        \item A-Cast $\langle (p_{i_1},V_{i_1}),\ldots,(p_{i_{2t+1}},V_{i_{2t+1}}) \rangle$.

        Participate in the A-Cast of a vector $u=\langle (p_{j_1},U_{j_1}),\ldots,(p_{j_{2t+1}},U_{j_{2t+1}}) \rangle$ only if for every $1\le k\le 2t+1$, $U_{j_k}$ accepted from $p_{j_k}$.
    \end{itemize}
    Epoch 3:
    \begin{itemize}
        \item Wait until $t+1$ vectors $u_{l_1},\ldots,u_{l_{2t+1}}$ [sic] accepted (from $p_{l_1},\ldots,p_{l_{t+1}}$) such that for every $1\le k \le t+1$ if $u_{l_k}=\langle (p_{k_1},U_{k_1}),\ldots,(p_{k_{2t+1}},U_{k_{2t+1}}) \rangle$, and for every $1\le r \le 2t+1$, $U_{k_r}$ accepted from $p_{k_r}$'s A-Cast\textsuperscript{+}.
        \item For every $1\le l \le n$:
        \begin{itemize}
            \item If a value $V$ was accepted from $p_l$ via A-Cast\textsuperscript{+} then $X_l\leftarrow (V,Relay_{i,l})$ (where $(V,Relay_{i,l})$ is the output of A-Cast\textsuperscript{+}).
            \item Else $X_l\assign \emptyset$.
        \end{itemize}
        \item Return $\langle X_1,X_2,\ldots,X_n \rangle$.
    \end{itemize}
}
\caption{The Spread protocol from~\cite[Fig.\ 4]{ben-or2003resilient}.}
\label{fig:prot-be-spread}
\end{figure}

\begin{figure}[htbp!]
\centering
\protbox{$\pselect$}{
    {\bf Common inputs:} $n$, $t\le (n-1)/3$

    {\bf Local input for processor $p_i$:} $V_i$ - an arbitrary value; $Pred_i(x)$ - a predicate

    {\bf Local output for processor $p_i$:} $D_i$

    \medskip

    Epoch 1:
    \begin{itemize}
        \item $U_i\leftarrow V_i$ and run Spread on input $U_i$ and let $View_i$ denote the local output.
    \end{itemize}
    Epoch 2:
    \begin{itemize}
        \item Run A-OLE. Let $k$ be the output denoting the identity of the leader, $p_k$.
    \end{itemize}
    Epoch 3:
    \begin{itemize}
        \item $U\leftarrow View_i[k]$ (i.e., the value ``related'' to $p_k$ upon termination of Spread).
        \item If $U$ was accepted during the run of Spread and $\models Pred_i(U)$, then $R_i\leftarrow U$ and $U_i\leftarrow [``deflected", U, Relay_i]$, where $Relay_i$ is obtained by Spread (via A-Cast\textsuperscript{+}) for the accepted broadcast of the leader $p_k$.
        \item Else $R_i\leftarrow \emptyset$.
        \item Run A-BA on input $R_i$. Let $D_i$ denote the local output.
        \item If $D_i\neq \emptyset$ then return $D_i$.
    \end{itemize}
    Epoch 4:
    \begin{itemize}
        \item A-Cast $View_i$ ($View_i$ defined in Epoch 1).

        Participate in admissible A-Casts only.
    \end{itemize}
    Epoch 5:
    \begin{itemize}
        \item Wait until $2t+1$ views, $View_{i_1},\ldots,View_{i_{2t+1}}$ accepted.
        \item A-Cast $U_i$.
    \end{itemize}
    Epoch 6:
    \begin{itemize}
        \item Wait until $2t+1$ values accepted.
        \item If a $[``deflected", U', Relay]$-message was accepted and $p_k\in Relay$ ($p_k$ is the leader from Epoch 2) and there are $t+1$ processors in $Relay$ which have a non-null value related to $p_k$ in their accepted Views with the same value $U'$ (if several such messages were accepted with different $U'$ values take an arbitrary one), then $R_i\leftarrow U'$.
        \item Else $R_i\assign \emptyset$.
    \end{itemize}
    Epoch 7:
    \begin{itemize}
        \item Run A-BA on input $R_i$ and let $D_i$ denote the local output.
        \item If $D_i\neq \emptyset$ then return $D_i$.
        \item Else $U_i\leftarrow V_i$ and go to Epoch 1.
    \end{itemize}
}
\caption{The Select protocol from~\cite[Fig.\ 5]{ben-or2003resilient}.}
\label{fig:prot-be-select}
\end{figure}

Next, we highlight that certain aspects of the attacks do not necessitate any corruption and can be executed solely by imposing appropriate delays. These aspects are common to both the static and adaptive attacks; thus, they form the initial focus of our analysis before we dive into the separate subsections for the static and adaptive attacks. Claims~\ref{claim:be1} and~\ref{claim:be2} formulate capabilities of adversaries only imposing delays maliciously.

\begin{claim}
\label{claim:be1}
    In any invocation of the $\pacastp$ protocol, the adversary can arbitrarily assign different sets of participants of size $n-t$ to the $Relay$ sets of honest parties' outputs by imposing appropriate delays. No corruption is necessary for this manipulation.
\end{claim}

\begin{proof}
To assign a set $S$ of participants to the $Relay$ set of an honest party $p_i$, the adversary can allow all the honest parties to accept the A-Cast initiated in Epoch 1 of $\pacastp$. However, the adversary selectively delivers $[``commit",\cdot,\cdot]$ messages from parties in set $S$ to $p_i$ before delivering any other $[``commit",\cdot,\cdot]$ messages.
\end{proof}

\begin{claim}
\label{claim:be2}
    Let $P=\{p_1,p_2,p_3,p_4\}$ be the set of participants, $v_i$ for $i\in[4]$ be the input of $p_i$, and $Relay_{i,j}$ for $i,j\in[4]$ be arbitrary subsets of $P$ with size $3$. The adversary can manipulate the output of the $\pspread$ protocol as follows solely through message scheduling. No corruption is necessary.
    \begin{align*}
        p_1\colon & \langle (v_1,Relay_{1,1}),(v_2,Relay_{1,2}),(v_3,Relay_{1,3}),(v_4,Relay_{1,4}) \rangle \\
        p_2\colon & \langle (v_1,Relay_{2,1}),(v_2,Relay_{2,2}),(v_3,Relay_{2,3}),(v_4,Relay_{2,4}) \rangle \\
        p_3\colon & \langle (v_1,Relay_{3,1}),(v_3,Relay_{3,3}),(v_4,Relay_{3,4}) \rangle \\
        p_4\colon & \langle (v_2,Relay_{4,2}),(v_3,Relay_{4,3}),(v_4,Relay_{4,4}) \rangle
    \end{align*}
\end{claim}

\begin{proof}
According to the description of $\pspread$, for any $i,j\in[4]$, the \jth component of the output of $p_i$ is either $\bot$ or is the output of $\pacastp$ accepted from $p_j$ (from Epoch 1). Since in each party's output all the available values corresponding to other parties are correct, the adversary does not need to make any changes to them. Moreover, based on Claim~\ref{claim:be1}, the adversary can set the $Relay$ set in the output of $\pacastp$ to any arbitrary subset of participants as long as it has the correct size, which is 3. Since all the $Relay_{i,j}$ sets are of size 3, the adversary will not face any challenges setting them.
Now we only need to demonstrate how the adversary can include and exclude inputs from different parties in each party's output to achieve the desired event. The adversary can accomplish this by carefully designing the order in which the parties accept the $\pacastp$ messages from Epoch 1. The following steps outline the message delivery order:
\begin{itemize}
    \item The adversary delivers the $\pacastp$ messages to $p_1$ in the order: $p_1, p_3, p_4, p_2$.
    \item The adversary delivers the $\pacastp$ messages to $p_2$ in the order: $p_2, p_3, p_4, p_1$.
    \item The adversary delivers the $\pacastp$ messages to $p_3$ in the order: $p_1, p_3, p_4$, and delays the $\pacastp$ message from $p_2$ until $p_3$ terminates.
    \item The adversary delivers the $\pacastp$ messages to $p_4$ in the order: $p_2, p_3, p_4$, and delays the $\pacastp$ message from $p_1$ until $p_4$ terminates.
\end{itemize}
By using this scheduling and delivering all the A-Cast\textsuperscript{+} messages initiated in Epoch 2, the following outcomes occur:
\begin{itemize}
    \item Parties $p_1$ and $p_3$ accept the vector $\langle(p_1,v_1),(p_3,v_3),(p_4,v_4) \rangle$ twice (once from $p_1$ and once from $p_3$ by the participation of $p_1, p_2, p_3$).
    \item Parties $p_2$ and $p_4$ accept the vector $\langle(p_2,v_2),(p_3,v_3),(p_4,v_4) \rangle$ twice (once from $p_2$ and once from $p_4$ by the participation of $p_1, p_2, p_4$).
\end{itemize}
Therefore, all parties terminate with the given output as stated in the claim.
\end{proof}

\subsection{Static Attack}

$\pselect$ may involve multiple iterations, where parties either agree on a value and terminate, or restart the protocol. In the proof of the $\pselect$ protocol, an event $E$ is defined, which is claimed to occur with a constant probability. It is argued that when $E$ happens in an iteration, all parties terminate. However, we demonstrate that this claim is not valid, resulting in uncertainty regarding the expected-constant round complexity guarantee of the $\pselect$ protocol.

Assuming the preconditions for the $\pselect$ protocol hold (i.e., for every non-faulty parties $p_i$ and $p_j$, $\models Pred_j(V_i)$ and $V_i \neq \emptyset$), we will demonstrate that the argument regarding the round complexity is not valid. Let the static adversary $\adv$ corrupt $p_1$ before the protocol begins. According to Claim~\ref{claim:be2}, $\adv$ can make parties terminate $\pspread$ in Epoch 1 of $\pselect$ with the following outputs where $Relay_{1,2}=Relay_{4,2}=\{p_1,p_2,p_3\}$ (no corruption is required):
\begin{align*}
    p_1\colon & \langle (v_1,Relay_{1,1}),(v_2,Relay_{1,2}),(v_3,Relay_{1,3}),(v_4,Relay_{1,4}) \rangle \\
    p_2\colon & \langle (v_1,Relay_{2,1}),(v_2,Relay_{2,2}),(v_3,Relay_{2,3}),(v_4,Relay_{2,4}) \rangle \\
    p_3\colon & \langle (v_1,Relay_{3,1}),(v_3,Relay_{3,3}),(v_4,Relay_{3,4}) \rangle \\
    p_4\colon & \langle (v_2,Relay_{4,2}),(v_3,Relay_{4,3}),(v_4,Relay_{4,4}) \rangle \\
\end{align*}

\begin{remark}
    Based on the properties of the $\pspread$ protocol, after the termination of the first honest party, there exist two subsets of participants, $P_1$ and $P_2$, satisfying the following conditions: $P_1$ only contains honest parties, $|P_1|=t+1$, $|P_2|=2t+1$, and all the values from parties in $P_2$ are included in the output vector of all the parties in $P_1$.
\end{remark}

Now let us assume that $\adv$ allows $p_4$ to be the first to terminate $\pspread$ in Epoch 1 of the $\pselect$ protocol. In this case, $P_1$ and $P_2$ could be $\{p_2,p_4\}$ and $\{p_2,p_3,p_4\}$, respectively.

In Epoch 2 of the $\pselect$ protocol, a leader $p_k$ is randomly chosen. The event $E$ is defined as $p_k\in P_2\cap (P-T)$ and it is claimed that if $E$ occurs, the protocol terminates with the correct output. Since $p_k$ is chosen randomly and $|P_2\cap (P-T)|=t+1$, the probability of the event $E$ is at least $\frac{1}{3}$. Therefore, this claim suggests an expected-constant number of iterations (or asynchronous rounds).

However, we claim that when $p_k=p_2$, which satisfies $E$ (since $p_2$ is honest and is in $P_2$), the protocol will start over instead of terminating. In the following, we explain the adversary's behavior in Epoch 3 to Epoch 7 of the $\pselect$ protocol and provide a proof for our claim.

In Epoch 3 of $\pselect$, $\adv$ can delay $p_2$ and assign a malicious value $v'_2\neq v_2$ to $R_1$. However, $\adv$ allows $p_1$ to follow the protocol for the rest of its actions. As a result, the following state is observed:
\begin{itemize}
    \item $p_1$: $R_1=v'_2$ and $U_1=(v_2,\{p_1,p_2,p_3\})$
    \item $p_2$: Delayed
    \item $p_3$: $R_3=\bot$ and $U_3=v_3$
    \item $p_4$: $R_4=v_2$ and $U_4=(v_2,\{p_1,p_2,p_3\})$
\end{itemize}

Consequently, the parties initiate the A-BA protocol with the inputs $(v'_2, \text{delayed}, \bot, v_2)$. However, since there is no majority in the input values, we have no guarantee on the output of this A-BA protocol. In fact, as formulated in the A-BA ideal functionality, the adversary $\adv$ can manipulate the protocol execution to make it output $\bot$ and prevent the termination of $\pselect$. Therefore, it is necessary to proceed to the next epoch, as the protocol cannot terminate at this stage.

In Epoch 4 of $\pselect$, $\adv$ allows $p_1$ to follow the protocol. It should be noted that all the views in this epoch are admissible to everyone, as the components of all the views are either $\bot$ or the correct values. Therefore, all the views can be received by the honest parties. The only control that $\adv$ has is over the order in which the views are received by the parties in the next epoch.

In Epoch 5 of $\pselect$, $\adv$ allows $p_1$ to follow the protocol and strategically manages delays to achieve the following:
\begin{itemize}
    \item Delivers $view_1, view_3, view_4$ to $p_1$ as its first $2t+1$ views.
    \item Delays the acceptance of any view for $p_2$ until the A-BA in Epoch 7 of $\pselect$ is terminated.
    \item Delivers $view_1, view_3, view_4$ to $p_3$ as its first $2t+1$ views.
    \item Delivers $view_1, view_3, view_4$ to $p_4$ as its first $2t+1$ views.
\end{itemize}
Then, parties A-Cast values as follows based on the output they received from the $\pspread$ protocol in Epoch 1 of $\pselect$:
\begin{itemize}
    \item $p_1$: $U_1 = (v_2, \{p_1, p_2, p_3\})$
    \item $p_2$: Still waiting from the previous epoch
    \item $p_3$: $U_3 = v_3$
    \item $p_4$: $U_4 = (v_2, \{p_1, p_2, p_3\})$
\end{itemize}
In Epoch 6 of $\pselect$, $\adv$ lets $p_1$ follow the protocol, keeps $p_2$ waiting, and allows the remaining parties to receive each other's A-Cast messages. As a result, the following situation arises:
\begin{itemize}
\item $p_1, p_3, p_4$: Receive $(v_2, \{p_1, p_2, p_3\})$ from $p_1, p_4$, and $v_3$ from $p_3$.
\item $p_2$: Still waiting from previous epochs
\end{itemize}
However, the ``if'' condition in Epoch 6 is not satisfied because only $p_1$ among $\{p_1, p_2, p_3\}$ has an accepted view that includes a value associated with $p_2$. Therefore, the following assignments are made:
\begin{itemize}
\item $p_1, p_3, p_4$: Assign $R_1 = R_3 = R_4 = \bot$.
\item $p_2$: Still waiting from previous epochs
\end{itemize}
In Epoch 7 of $\pselect$, $\adv$ allows $p_1$ to follow the protocol. Consequently, parties initiate the A-BA protocol with the following inputs:
\begin{itemize}
    \item $p_1, p_3, p_4$: $\bot$
    \item $p_2$: Still waiting from previous epochs
\end{itemize}
As a result, parties complete the A-BA protocol and obtain the output $\bot$. Since the output is $\bot$, they have to start over, completing our attack.

\subsection{Adaptive Attack}

In the static attack, we demonstrated that the claim ``for all possibilities for the set $P_2$, if the leader is an honest party in $P_2$, then the protocol terminates with the correct value'' is not valid. We provided specific choices of the leader within certain possibilities for $P_2$ where the claim does not hold. This shows that the claim is incorrect in those cases.

In the adaptive attack, we further illustrate the difficulty in finding an argument for the weaker claim ``there exists a possibility for $P_2$ such that if the leader is an honest party in $P_2$, then the protocol terminates with the correct values.'' We use the same setup with four parties, namely $p_1$, $p_2$, $p_3$, and $p_4$, and demonstrate that if either $p_1$ or $p_2$ is chosen as the leader, the adversary $\adv$ can manipulate the protocol to start over instead of terminating. Since any legitimate set $P_2$ consists of three parties, it must include at least one of $p_1$ and $p_2$. Therefore, in any possible set $P_2$, there is an honest party whose selection as the leader does not guarantee termination. This shows the hopelessness of finding a valid argument for the weaker claim.

Assuming the preconditions for the $\pselect$ protocol hold, we now consider adaptive corruption in this subsection. The adversary $\adv$ does not corrupt any party before the leader is selected. If the chosen leader is $p_2$, then $\adv$ follows the description of the static attack by corrupting $p_1$ and executing the attack as previously described.

However, if party $p_1$ is chosen as the leader, we can leverage symmetry and modify the adversary's description accordingly. We rename the parties as follows: $p'_1 \assign p_2$, $p'_2 \assign p_1$, $p'_3 \assign p_4$, and $p'_4 \assign p_3$. The adversary $\adv$ corrupts $p'_1$ and runs the same static adversary code on the participant set $\{p'_1, p'_2, p'_3, p'_4\}$. Furthermore, it is important to note that in the static attack, we previously specified that $Relay_{1,2} = Relay_{4,2} \assign \{p_1, p_2, p_3\}$, while allowing $Relay_{2,1}$ and $Relay_{3,1}$ to be arbitrary since they were not essential to the attack and could take any value. However, in order to achieve full symmetry in the static attack and apply the aforementioned approach, $\adv$ should now set $Relay_{2,1} = Relay_{3,1} \assign \{p_1, p_2, p_4\}$ in the description of the static attack.

With these adjustments, our adaptive adversary is complete. 
\section{The Asynchronous Turpin-Coan Extension}
\label{app:tc}

Here, we investigate when the generic extension of binary to multi-valued synchronous BA (for $t<n/3$) given by Turpin and Coan~\cite{turpin1984extending} works in the asynchronous setting with eventual delivery. It turns out that an asynchronous version of the extension---with appropriate modifications---is secure when $t<n/5$, but is provably insecure for any $t\ge n/5$ regardless of which binary A-BA protocol is used. We remark that while our negative result is not exactly a lower bound, as the attack that we present is on a specific extension protocol, the fact that we lose two additive factors of $t$ in resiliency (and even three, with a more \naive approach) gives some evidence that more sophisticated techniques are needed to maintain optimal resiliency, as in~\cite{mostefaoui2017signature}. Since our primary goal is to show that the Turpin-Coan extension cannot be used to get optimally resilient, multi-valued A-BA in expected-constant rounds, we choose to work with simpler, property-based definitions in this section.

We start by reviewing the original extension in~\cite{turpin1984extending}, which requires just two additional rounds. Parties first distribute their inputs amongst one another, over P2P channels. An honest party, based on how many of the received values disagree with his own, considers himself either ``perplexed'' or ``content,'' and announces this information. If enough parties claim to be perplexed, the honest party becomes ``alert.'' The parties now run a binary BA protocol to agree on this last state (effectively, to determine if they all started with the same input), and depending on the outcome they either output a default value, or are able to recover a common value from their local transcripts of the protocol.

\begin{figure}[ht!]
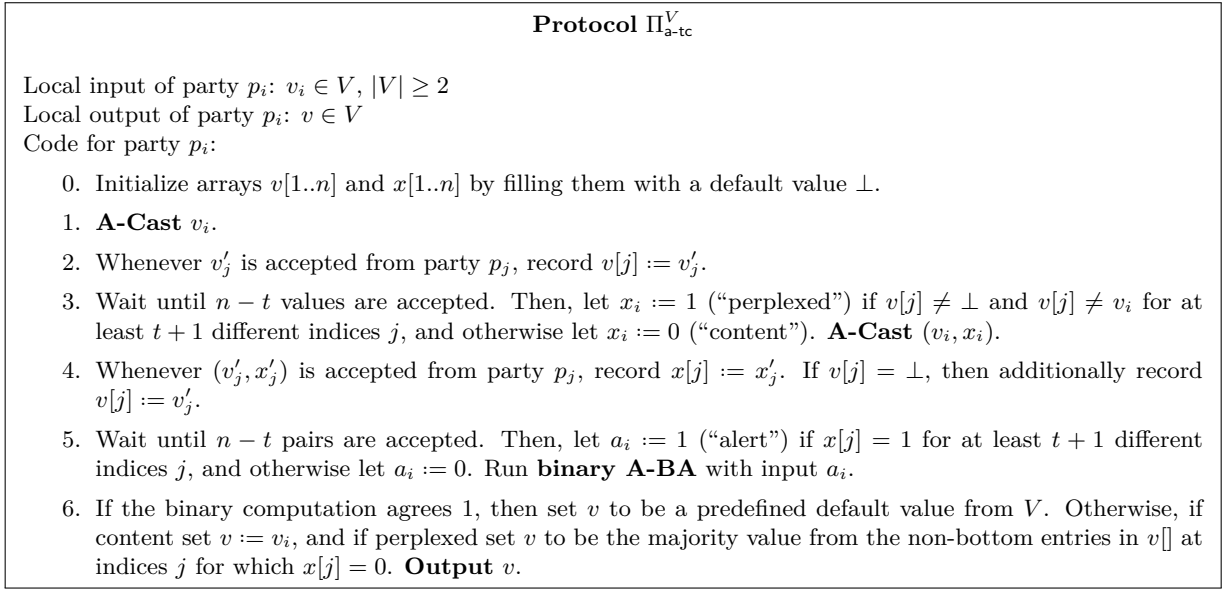

\centering
\protbox{$\patc^V$}{
    Local input of party $p_i$: $v_i\in V$, $|V|\ge2$

    Local output of party $p_i$: $v\in V$

    Code for party $p_i$:
    \begin{enumerate}
        \item[0.] Initialize arrays $v[1..n]$ and $x[1..n]$ by filling them with a default value $\bot$.
        \item \textbf{A-Cast} $v_i$.
        \item Whenever $v'_j$ is accepted from party $p_j$, record $v[j]\assign v'_j$.
        \item Wait until $n-t$ values are accepted. Then, let $x_i\assign1$ (``perplexed'') if $v[j]\not=\bot$ and $v[j]\not=v_i$ for at least $t+1$ different indices $j$, and otherwise let $x_i\assign0$ (``content''). \textbf{A-Cast} $(v_i,x_i)$.
        \item Whenever $(v'_j,x'_j)$ is accepted from party $p_j$, record $x[j]\assign x'_j$. If $v[j]=\bot$, then additionally record $v[j]\assign v'_j$.
        \item Wait until $n-t$ pairs are accepted. Then, let $a_i\assign1$ (``alert'') if $x[j]=1$ for at least $t+1$ different indices $j$, and otherwise let $a_i\assign0$. Run \textbf{binary A-BA} with input $a_i$.
        \item If the binary computation agrees $1$, then set $v$ to be a predefined default value from $V$. Otherwise, if content set $v\assign v_i$, and if perplexed set $v$ to be the majority value from the non-bottom entries in $v[]$ at indices $j$ for which $x[j]=0$. \textbf{Output} $v$.
    \end{enumerate}
}
\caption{The asynchronous Turpin-Coan extension.}
\label{fig:prot-atc}
\end{figure}

Our asynchronous version of this extension, $\patc$, is shown in Figure~\ref{fig:prot-atc}. We simply replace the basic message distribution mechanism with A-Cast, and account for the fact that adversarial A-Cast's may never terminate (i.e., honest parties should only expect to receive $n-t$ values). To avoid a somewhat artificial attack wherein the adversary can manipulate delays so that $t$ entries are missing in an honest party's $v[]$ array and a \emph{different} set of $t$ entries are missing in the $x[]$ array, we require parties to re-send their input value in the second A-Cast. Note that this additional modification cannot \emph{decrease} security, since honest parties will always use their actual input. In the following proposition, we prove the security of $\patc$ for $t<n/5$. We emphasize that in the asynchronous setting we must impose relaxed termination requirements. However, since the modified extension is still deterministic and constant-round, any termination guarantees provided by the underlying binary A-BA protocol should be preserved; accordingly, we only prove the agreement and validity properties.

\begin{proposition}
    For any input domain $V$, protocol $\patc^V$ securely implements multi-valued A-BA from A-Cast and binary A-BA, in constant rounds and in the presence of an adaptive and malicious $t$-adversary, provided $t<n/5$.
\end{proposition}

\begin{proof}
The round complexity is easily established by inspection (considering A-Cast and binary A-BA to be ideal primitives). We now prove correctness:

\paragraph{Validity.} If all honest parties begin with the same initial value $v$, then no honest party is perplexed and all honest parties are not alert. The binary computation agrees $0$ and all honest parties (which are content) output $v$.

\paragraph{Agreement.} There are two cases: the binary computation agrees either $1$ or $0$. In the former case, all honest parties output a common default value. For the latter case, we show that all content parties have the same initial value, and that this value is correctly deduced by all the perplexed parties.

Note that each content party's input appears in at least $n-t-t=n-2t$ positions in its $v[]$ array. Since at most $t$ of these positions correspond to corrupted parties, we have that each content party has the same initial value as at least $n-3t$ honest parties, which represent a majority of the honest parties when $n>5t$. Of course there cannot be two distinct majorities, so all content parties must have the same initial value.

Furthermore, since the binary computation agrees $0$, there must be at least $2t+1$ honest content parties, for otherwise there would be at least $n-3t$ honest perplexed parties and all honest parties would be alert (at least $n-4t\ge t+1$ of the corresponding entries in any honest party's $x[]$ array would not be $\bot$) and agreement would be on $1$. Each perplexed party has $x[j]=0$ or $x[j]=\bot$ for all honest content parties and possibly for some corrupted parties. Even in the worst case, when $x[j]=\bot$ for $t$ of the honest content parties, there are still at least $2t+1-t=t+1$ honest content parties left but at most $t$ corrupted parties, so the content parties are a majority of those for which $x[j]=0$ and $v[j]\not=\bot$.
\end{proof}

It is not difficult to see that when the predefined default value comes from outside of $V$, the resulting multi-valued A-BA protocol is actually \emph{intrusion-tolerant} (see Section~\ref{sec:results}). We now show that $\patc$ is insecure when $t\ge n/5$ (i.e., that the analysis above is tight). Our attack does not make any assumptions regarding the underlying binary A-BA protocol, other than that it satisfies the properties of A-BA. Moreover, we require only static corruptions, and somewhat surprisingly the attack considers only binary inputs (i.e., using the extension when $|V|=2$ actually breaks the security of the binary protocol for $n/5\le t<n/3$).

\begin{proposition}
    For all input domains $V$, protocol $\patc^V$ does NOT securely implement multi-valued A-BA from A-Cast and binary A-BA, in the presence of a static and malicious $t$-adversary, for any $t\ge n/5$.
\end{proposition}

\begin{proof}
It suffices to consider the case $n=5t$. We construct a distribution of inputs, a static adversary $\adv$, and a message scheduling such that agreement is broken.

Let $\{p_1,\ldots,p_{5t}\}$ be the set of parties. Define four sets $A\assign\{p_1,\ldots,p_{2t}\}$, $B\assign\{p_{2t+1},\ldots,p_{3t}\}$, $T\assign\{p_{3t+1},\ldots,p_{4t}\}$, and $C\assign\{p_{4t+1},\ldots,p_{5t}\}$. The parties in $T$ are statically corrupted by $\adv$ and instructed to start the protocol with input $0$. Suppose the parties in $A$ and $B$ start with input $1$, and parties in $C$ with input $0$. Now consider a message scheduling in which: 1) both A-Casts from parties in $C$ get delayed for parties in $A$, 2) the A-Casted input values from $p_1,\ldots,p_{t}\in A$ get delayed for parties in $B$ and $C$, 3) the A-Casted pairs from parties in $C$ get delayed for parties in $B$, and 4) the A-Casted pairs from $p_{t+1},\ldots,p_{2t}\in A$ get delayed for parties in $C$.

Observe that all parties in $A$ will A-Cast that they are content. $\adv$ also instructs parties in $T$ to maliciously A-Cast that they are content. On the other hand, the parties in both $B$ and $C$ will A-Cast that they are perplexed. Then, no party in $A$ or $B$ will be alert, whereas all parties in $C$ will be alert. We summarize all this information in Figure~\ref{fig:attack}.

\begin{figure}[ht!]
\centering
\begin{tabular}{|c||c|c|c|c|c|c}
    \cline{1-6}
     & $A_1$ & $A_2$ & $B$ & $\color{red}T$ & $C$ & \\
    \cline{1-6}
    \multirow{2}{1em}{\centering$A$} & $1$ & $1$ & $1$ & $\color{red}0$ & \cellcolor{lightgray} \\
    & $(1,\text{c})$ & $(1,\text{c})$ & $(1,\text{p})$ & $\color{red}(0,\text{c})$ & \cellcolor{lightgray} & not alert \\
    \cline{1-6}
    \multirow{2}{1em}{\centering$B$} & \cellcolor{lightgray} & $1$ & $1$ & $\color{red}0$ & $0$ \\
    & $(1,\text{c})$ & $(1,\text{c})$ & $(1,\text{p})$ & $\color{red}(0,\text{c})$ & \cellcolor{lightgray} & not alert \\
    \cline{1-6}
    \multirow{2}{1em}{\centering$C$} & \cellcolor{lightgray} & $1$ & $1$ & $\color{red}0$ & $0$ \\
    & $(1,\text{c})$ & \cellcolor{lightgray} & $(1,\text{p})$ & $\color{red}(0,\text{c})$ & $(0,\text{p})$ & alert \\
    \cline{1-6}
\end{tabular}
\caption{Summary of the A-Casts received by parties in $A$, $B$, and $C$, and their resulting alert status. The content and perplexed values are denoted by ``c'' and ``p,'' respectively. Note that each column corresponds to a set of $t$ parties ($A_1$ and $A_2$ denote $\{p_1,\ldots,p_t\}$ and $\{p_{t+1},\ldots,p_{2t}\}$, respectively).}
\label{fig:attack}
\end{figure}

Consequently, if $\adv$ instructs parties in $T$ to run the binary A-BA protocol with input $0$ (not alert), and delays messages from parties in $C$ during the computation, then agreement will be on $0$. Therefore, parties in $A$ will output their initial value $1$, while the parties in $B$ and $C$ will try to deduce the correct output from their arrays. All parties in $B$ can do this successfully ($2t$ of the $3t$ relevant entries in the $v[]$ array will be $1$). However, for parties in $C$, $t$ of the $2t$ relevant entries will be $1$ while the other $t$ will be $0$.
\end{proof}

It is possible to show that without the small modification in which parties A-Cast their input twice, the protocol is secure if and only if $t<n/6$. We have a method to extend binary to multi-valued A-BA in constant rounds for $t<n/4$, which is not much more complicated than the synchronous extension in~\cite{turpin1984extending} (although it requires \emph{two} invocations of the binary protocol). However, as discussed in Section~\ref{sec:functionalities}, to re-gain optimal resiliency $t<n/3$ we use the asynchronous extension of Most\'{e}faoui and Raynal~\cite{mostefaoui2017signature}. We leave open the possibility of obtaining expected-constant-round multi-valued A-BA via other approaches, such as substituting our multi-valued OCC in the binary A-BA protocol from~\cite{canetti1993fast} (and making appropriate changes). In any case, we prefer the solution in~\cite{mostefaoui2017signature} as it is a \emph{black-box} reduction.

\end{document}